\documentclass[11pt]{article}
\usepackage{macros}
\usepackage{fullpage}
\usepackage{authblk}
\usepackage{enumitem}
\usepackage{afterpage}
\usepackage{caption}
\usepackage{subcaption}
\usepackage{cancel}
\usepackage{amsmath}
\usepackage{thmtools}
\usepackage{thm-restate}

\usepackage{enumitem}
\usepackage{lmodern} 

\usepackage{microtype}
\usepackage{amssymb}
\usepackage{mathtools}
\usepackage{latexsym}
\usepackage[square,sort,comma,numbers]{natbib} 
\usepackage{algorithmic}
\usepackage{multirow}
\usepackage{multicol}

\renewenvironment{proof}{{\bfseries Proof.}}{\qed}

\usepackage{tikz}
\usetikzlibrary{arrows}
\usepackage{graphicx}

\DeclarePairedDelimiter{\ceil}{\lceil}{\rceil}
\DeclarePairedDelimiter{\floor}{\lfloor}{\rfloor}

\DeclareMathOperator*{\E}{\mathbb{E}}
\DeclareMathOperator*{\EV}{\mathcal{E}}

\renewcommand{\vec}[1]{\textbf{#1}}

\newcommand{\RLA}{\mathsf{RLA}}

\newcommand{\PM}[1]{\mathsf{PM}[\textbf{F}, #1]}
\newcommand{\DR}[1]{\mathsf{DR}[\textbf{f}, #1]}

\newcommand{\RR}[2]{\mathsf{R}[#1, #2]}
\newcommand{\FF}{\textbf{F}}
\newcommand{\HH}{\textbf{H}}

\newcommand{\ff}{\textbf{f}}

\newcommand{\VW}{\mathsf{VW}}
\newcommand{\EW}{\mathsf{EW}}
\newcommand{\WS}{\mathsf{WS}}

\newcommand{\euler}{\mathbf{\mathsf{e}}} 

\newcommand{\myEdit}[1]{{\color{black} #1}}
\newcommand{\del}{\delta}
\newcommand{\alp}{\alpha}

\newcommand{\pa}{\partial}
\def \ALG {\ensuremath{\operatorname{ALG}}\xspace}
\def \OPT {\ensuremath{\operatorname{OPT}}\xspace}
\def \LP {\ensuremath{\operatorname{LP}}\xspace}

\usepackage[toc, page]{appendix}

\usepackage[colorlinks=true,breaklinks=true,bookmarks=true,urlcolor=blue,
     citecolor=red,linkcolor=blue,bookmarksopen=false,draft=false]{hyperref}

 \title{Online Stochastic Matching: New Algorithms and Bounds\footnote{A preliminary version of this appeared in the European Symposium on Algorithms (ESA), 2016. Research supported NSF Awards CNS 1010789, CCF 1422569, CCF-1749864 and research awards from Adobe, Amazon, and Google.}}

\author{Brian Brubach\thanks{\textbf{Email: } \texttt{bbrubach@cs.umd.edu}}}
\author{Karthik A. Sankararaman\thanks{\textbf{Email: }\texttt{kabinav@cs.umd.edu}}}
\author{Aravind Srinivasan\thanks{\textbf{Email: } \texttt{srin@cs.umd.edu}}}
\author{Pan Xu\thanks{\textbf{Email: } \texttt{panxu@cs.umd.edu}}}
\affil{Department of Computer Science, University of Maryland, College Park, MD}

\date{First Version: May 2016 \\ This Version: July 2019}

\begin{document}

\maketitle

\begin{abstract}
	Online matching has received significant attention over the last $15$ years due to its close connection to Internet advertising. As the seminal work of Karp, Vazirani, and Vazirani has an optimal $(1 - 1/e)$ competitive ratio in the standard adversarial online model, much effort has gone into developing useful online models that incorporate some stochasticity in the arrival process. One such popular model is the ``known I.I.D. model'' where different customer-types arrive online from a known distribution. We develop algorithms with improved competitive ratios for some basic variants of this model with integral arrival rates, including: (a) the case of general weighted edges, where we improve the best-known ratio of $0.667$ due to Haeupler, Mirrokni and Zadimoghaddam (WINE 2011) to $0.705$; and (b) the vertex-weighted case, where we improve the $0.7250$ ratio of Jaillet and Lu (\emph{Mathematics of Operations Research}, 2013) to $0.7299$. 
	
	We also consider an extension of stochastic rewards, a variant where each edge has an
 independent probability of being present. For the setting of stochastic rewards with non-integral arrival rates, we present a simple \emph{optimal} non-adaptive algorithm with a ratio of $1-1/e$. For the special case where each edge is unweighted and has a uniform constant probability of being present, we improve upon $1-1/e$ by proposing a strengthened LP benchmark.

 One of the key ingredients of our improvement is the following (offline) approach to bipartite-matching polytopes with additional constraints. We first add several valid constraints in order to get a good fractional solution $\ff$; however, these give us less control over the structure of $\ff$. We next \emph{remove} all these additional constraints and randomly move from $\textbf{f}$ to a feasible point on the matching polytope with all coordinates being from the set $\{0, 1/k, 2/k, \ldots, 1\}$ for a chosen integer $k$. The structure of this solution is inspired by Jaillet and Lu (\emph{Mathematics of Operations Research}, 2013) and is a tractable structure for algorithm design and analysis. The appropriate random move preserves many of the removed constraints (approximately with high probability and exactly in expectation). This underlies some of our improvements and could be of independent interest.
\end{abstract}

\section{Introduction.}

Applications to Internet advertising have driven the study of online matching problems in recent years 
\cite{mehtaBook}. In these problems, we consider a bipartite graph $G = (U, V, E)$ in which the set of vertices $U$ is available \emph{offline} while the set of vertices in $V$ arrive \emph{online}. Whenever some vertex $v$ arrives, it must be matched immediately (and irrevocably) to (at most) one vertex in $U$. Each offline vertex $u$ can be matched to at most one $v$. In the context of Internet advertising, $U$ is the set of advertisers and $V$ is the set of impressions. The edges $E$ define the impressions that interest a particular advertiser. When an impression $v$ arrives, we must choose an available advertiser (if any) to match with it. We consider the case where $v \in V$ can be matched at most once upon arriving. Since advertising forms the key source of revenue for many large Internet companies, finding good matching algorithms and obtaining even small performance gains can have high impact. 

In the \emph{stochastic known I.I.D.} model of arrival, we are given a bipartite graph $G=(U,V,E)$ and a finite online time horizon $T$ (in most cases, we assume $T = |V| = n$ and say the online phase takes place over $n$ rounds). In each round, a vertex $v$ is sampled with replacement from a known distribution over $V$. The sampling distributions are independent and identical over all of the $T$ online rounds. This captures the fact that we often have historical data about the impressions and can predict the frequency with which each type of impression will arrive. 
Edge-weighted matching \cite{feldmanEdgeWeighted} is a general model in the context of advertising: every advertiser gains a given revenue for being matched to a particular type of impression. Here, a \emph{type} of impression refers to a class of users (e.g., a demographic group) who are interested in the same subset of advertisements. Each arrival of a type $v \in V$ is considered a distinct vertex (user) that can be matched to up to one $u \in U$. For example, if the same $v$ arrives three times, we consider this three separate vertices (or \textit{copies} of $v$) that can potentially be matched to three different vertices in $U$. A special case of this model is vertex-weighted matching \cite{aggarwalVertex}, where weights are associated only with the advertisers (the offline set $U$). In other words, a given advertiser has the same revenue generated for matching any of the user types interested in it. 
In some modern business models, revenue is not generated upon matching advertisements, but only when a user \emph{clicks} on the advertisement: this is the \emph{pay-per-click} model. From historical data, one can assign the probability of a particular advertisement being clicked by a type of user.  Works including \cite{mehta2012online, mehtaonline} capture this notion of \textit{stochastic rewards} by assigning a probability to each edge.

One unifying theme in most of our approaches is the use of an LP benchmark with additional valid constraints that hold for the respective stochastic-arrival models. We use the optimal solution to this \LP to guide our online actions. In most cases, we use various modifications of dependent randomized rounding to convert the fractional LP solution into a suitable guide for our online algorithms.


\section{Preliminaries and technical challenges.}
\label{sec:prelim}

In the \emph{Unweighted Online Known I.I.D. Stochastic Bipartite Matching} problem, we are given a bipartite graph $G = (U, V, E)$. The set $U$ is available offline while the vertex set $V$ represent the online vertices. Each edge $e \in E$ is associated with a weight $w_e$. Thus, this represents the \emph{input graph}. The vertices $v$ arrive online and are drawn with replacement from an I.I.D. distribution on $V$. For each $v \in V$, we are given an \textit{arrival rate} $r_v$, which is the expected number of times $v$ will arrive. With the exception of Section~\ref{sec:nonint}, this paper will focus on the integral arrival rates setting where all $r_v \in \mathbb{Z}^+$.  For reasons described in~\cite{bib:Haeupler}, we can further assume WLOG that each $v$ has $r_v=1$ under the assumption of integral arrival rates. In particular, a vertex type $v$ with an integral arrival rate $k > 1$, can be split into $k$ different vertex types each with an arrival rate of $1$. In this case, we have that $|V|=n$ where $n$ is the total number of online rounds. 


In the \textbf{vertex-weighted} variant, every offline vertex $u \in U$ has a weight $w_u$ (alternatively, for any vertex $u \in U$ all edges incident to $u$ have the same weight) and we seek a maximum weight matching. In the \textbf{edge-weighted} variant, every edge $e \in E$ has a weight $w_e$ and we again seek a maximum weight matching. 

In the \textbf{stochastic rewards} variant, each edge has a probability $p_e$ of being present once we probe edge $e$ 
and we seek to maximize the expected size or weight of the matching. 
The edge realization process is independent for different edges. At each step, the algorithm ``probes" an edge $e$. With probability $p_e$ the edge $e$ exists and with the remaining probability it does not. Once realization of an edge is determined, it does not affect the random realizations for the rest of the edges. We consider the query-commit model where an edge that is probed and found to exist must be matched. For a single arriving vertex, each edge can only be probed once. However, we remind the reader that multiple arrivals of the same vertex \textit{type} are considered distinct vertices. Suppose the first arrival of a vertex type $v \in V$ probes an edge to some offline vertex $u \in U$ and the edge does not exist. Then later, another copy of type $v$ might arrive and also probe an edge to $u$ because each arrival is a distinct vertex with its own hidden edge realizations.

\textbf{Asymptotic assumption and notation.} We will always assume $n$ is large and analyze algorithms as $n$ goes to infinity: e.g., if $x \leq 1 - (1-2/n)^n$, we will just write this as ``$x \leq 1 - 1/\euler^2$'' instead of the more-accurate ``$x \leq 1 - 1/\euler^2 + o(1)$''. These suppressed $o(1)$ terms will subtract at most $o(1)$ from our competitive ratios. Note the we use $\euler$ for Euler's constant in contrast with $e$ which denotes an edge.
Throughout, we use ``$\WS$'' to refer to the worst case instance for various algorithms. 

\myEdit{
\textbf{Competitive ratio.} The \textit{competitive ratio} is defined slightly differently than usual for this set of problems (similar to the notation used in \cite{mehtaBook}). In particular, it is defined as 
	$\frac{\mathbb{E}[\ALG]}{\mathbb{E}[\OPT]}$.
Here, $\mathbb{E}[\ALG]$ is the expected performance of our online algorithm with respect to the random online vertex arrivals and any internal randomness the algorithm may use; and for the stochastic rewards variant the random edge realizations, arrival sequence and internal randomness of the algorithm. Similarly, $\mathbb{E}[\OPT]$ is the expected performance of an optimal offline matching algorithm which knows the random vertex arrivals in advance. In the case of stochastic rewards, we compare to an optimal offline stochastic matching algorithm which can probe edges in any order, but does not know the outcomes of these probes and can only probe one neighbor of each vertex from the ``online'' partition.
	
\textbf{Adaptivity.} Algorithms can be \textit{adaptive} or \textit{non-adaptive}. When $v$ arrives, an adaptive algorithm can modify its online actions based on the realization of the online vertices (and edges in the stochastic rewards model) thus far, but a non-adaptive algorithm has to specify all of its actions before the start of the online phase. 
}


\subsection{LP benchmark for deterministic rewards.}
\label{lp:stoch-match}
As in prior work (\emph{e.g,} see \cite{mehtaBook}), we use the following LP to upper bound the optimal offline expected performance and also use it to guide our algorithm in the cases where rewards are deterministic. For the case of stochastic rewards, we use slightly modified LPs, whose definitions we defer until Sections \ref{sec:nonint} and \ref{sec:uniform_p}. We first show an LP for the unweighted variant, then describe the changes for the vertex-weighted and edge-weighted settings. As usual, we have a variable $f_e$ for each edge. Let $\partial(w)$ be the set of edges adjacent to a vertex $w \in U \cup V$ and let $f_w =  \sum_{e \in \partial(w)} f_e$.  Constraint~\eqref{cons:single} is used in ~\cite{bib:Manshadi} and~\cite{bib:Haeupler}. 


\vspace*{-0.1in}
\begin{alignat}{2}
\text{maximize}    &\quad  \sum_{e \in E} f_e     \\
\text{subject to}  
	& \quad\sum_{e \in \partial(u)} f_e \leq 1	&\ & \quad\forall u \in U \label{cons:umatch}\\
	& \quad\sum_{e \in \partial(v)} f_e \leq 1	&\ & \quad\forall v \in V \label{cons:vmatch}\\
	& \quad 0 \leq f_e \leq 1-1/\euler	 &\ & \quad\forall e \in E \label{cons:single}\\
	& \quad f_e + f_{e'} \leq 1-1/\euler^2	 &\ & \quad\forall e,e' \in \partial(u), \forall u \in U \label{cons:pair}
\end{alignat}

\xhdr{Variants: } The objective function is to $\text{maximize}$ $\sum_{u \in U} \sum_{e \in \partial(u)} f_e w_u$ in the vertex-weighted variant and $\text{maximize}$ $\sum_{e \in E} f_e w_e$ in the edge-weighted variant (here $w_e$ refers to $w_{(u, v)}$).


\myEdit{
\begin{lemma}
	\label{label:LPUpperbound}
	Let $\OPT$ denote the total weight obtained by the best offline algorithm. Let $\mathbf{f}^*$ denote the optimal solution to the above $\LP$. Then $ \sum_{e \in E} f^*_e \geq \mathbb{E}[\OPT]$.
\end{lemma}
}
\begin{proof}
\myEdit{We prove this as follows. Let $Y_e$ denote the indicator random variable for the event that edge $e \in E$ is matched in the optimal solution for a given arrival sequence $\mathcal{A}$. Let $y_e := \mathbb{E}_{\mathcal{A}}[Y_e]$ for every edge $e \in E$. We will now argue that the vector $\vec{y} := (y_e)_{e \in E}$ is a feasible solution to the LP. Consider a vertex $u \in U$. We have that $\sum_{ e \in \partial(u)} Y_e \leq 1$. Taking expectations on both sides and using the linearity of expectation we have $\sum_{e \in \partial(u)} Y_e \leq 1$. This shows that $\vec{y}$ is feasible to the constraint~\eqref{cons:umatch}. Let $R_v$ denote the random variable for the number of times a vertex $v \in V$ arrived in a given arrival sequence $\mathcal{A}$. Then we have, for every $v \in V$, $\sum_{e \in \partial(v)} Y_e \leq R_v$. From the integral arrival rates assumption, $\mathbb{E}_{\mathcal{A}}[R_v] = 1$ for every $v \in V$. Thus, from linearity of expectation we obtain $\sum_{e \in \partial(v)} Y_e \leq 1$. This shows that $\vec{y}$ is feasible to the constraint~\eqref{cons:vmatch}. For any edge $e=(u, v)$, let $\mathbb{I}[R_v = 0]$ be an indicator for the event that a vertex $v \in V$ never arrives in the $T$ rounds. Thus, for any arrival sequence $\mathcal{A}$, we have $Y_e \leq \mathbb{I}[R_v \neq 0]$. Taking expectations on both sides we get $Y_e \leq \mathbb{E}_{\mathcal{A}}[\mathbb{I}[R_v \neq 0]$. The probability that a vertex $v$ never arrives in $T$ rounds is $\left(1-\frac{1}{T}\right)^T \leq 1/\euler$. Thus, $\mathbb{E}_{\mathcal{A}}[\mathbb{I}[R_v \neq 0] \leq 1-1/\euler$. This shows that $\vec{y}$ is feasible to the constraint~\eqref{cons:single}. Consider two edges $e, e' \in \partial(u)$ for some $u \in U$. Let $e=(u, v)$ and $e'=(u, v')$ and as before let $\mathbb{I}[R_v \neq 0]$ and $\mathbb{I}[R(v') \neq 0]$ denote the indicator for the events that $v, v'$ arrives at least once in the $T$ rounds, respectively. For any arrival sequence $\mathcal{A}$ we have that $Y_e + Y(e') \leq \mathbb{I}[R_v \neq 0] \wedge \mathbb{I}[R(v') \neq 0]$. Taking expectations on both sides we get $Y_e + y(e') \leq \mathbb{E}_{\mathcal{A}}[\mathbb{I}[R_v \neq 0] \wedge \mathbb{I}[R(v') \neq 0]]$. The probability that both $v$ and $v'$ never arrive in the $T$ rounds is given by $\left( 1- \frac{2}{T} \right)^T \leq \frac{1}{\euler^2}$. Thus, we get $Y_e + y(e') \leq 1-\frac{1}{\euler^2}$ which shows that $\vec{y}$ is feasible to the constraint~\ref{cons:pair}. 
			
	The expected weight of the optimal solution is $\mathbb{E}_{\mathcal{A}}[\sum_{e \in E} w_e Y_e]$ which from linearity of expectation gives $\sum_{e \in E} w_e Y_e$. Since $\vec{y}$ is a feasible solution we have that the optimal value to LP is at least as large as the expected optimal solution.}
\end{proof}
\myEdit{We compare the performance of our algorithm to this $\LP$. Suppose that $\vec{f}^*$ is the optimal solution to the above $\LP$. We prove the following lemma which shows that it suffices to analyze the competitive ratio \emph{edge-wise}.
	\begin{lemma}
		\label{lem:edgeCR}
		If $\min_{e \in E, f_e^* >0} \frac{\Pr[\text{$e$ is included in the matching}]}{f_e^*} \geq \alpha$, then this implies that the competitive ratio is at least $\alpha$.
	\end{lemma}
	\begin{proof}
		 From linearity of expectation we have that 
		 \begin{align*}
		 	 \mathbb{E}[\ALG] &= \sum_{e \in E}  \Pr[\text{$e$ is included in the matching}]\\
		 	 &  \geq \alpha \sum_{e \in E} f^*_e \\
		 	 &  \geq \alpha \mathbb{E}[\OPT].
		 \end{align*}
	\end{proof}
	In what follows, we only compute a lower-bound on the probability that any edge $e \in E$ is included in the final matching (we call this quantity \emph{competitive ratio of edge $e$}) which would imply a lower-bound on the competitive ratio.
	
	In the vertex-weighted setting (Section~\ref{sec:vweight}) we compute a lower-bound on the probability that a vertex $u \in U$ is matched in any randomized online algorithm. Analogous to Lemma~\ref{lem:edgeCR}, the following lemma connects the lower-bound on this probability to the competitive ratio.
	
	\begin{lemma}
		\label{lem:vertexCR}
		Define $F_u := \sum_{e \in \partial(u)} f_e$. If $\min_{u \in U, F^*_u >0} \frac{\Pr[\text{$u$ is matched}]}{F^*_u} \geq \alpha$, then this implies that the competitive ratio is at least $\alpha$.
	\end{lemma}
	\begin{proof}
		 From linearity of expectation we have that 
		 \begin{align*}
		 	 \mathbb{E}[\ALG] &= \sum_{u \in U}  \Pr[\text{$u$ is matched}]\\
		 	 &  \geq \alpha \sum_{u \in U} F^*_u \\
		 	 & = \alpha \sum_{u \in U}  \sum_{e \in \partial(u)} f^*_e \\
		 	 &  \geq \alpha \mathbb{E}[\OPT].
		 \end{align*}
	\end{proof}
	}
	
Note that the work of \cite{bib:Manshadi} does not use an \LP to upper-bound the optimal value of the offline instance. Instead they use Monte-Carlo simulations wherein they simulate the arrival sequence and compute the vector $\vec{f}$ by approximating (via Monte-Carlo simulation) the probability of matching an edge $e$ in the offline optimal solution. We do not use a similar approach for our problems for a few reasons. (1) For the \emph{weighted} variants, namely the edge and vertex-weighted versions, the number of samples depends on the maximum value of the weight, making it expensive. (2) In the unweighted version, the running time of the sampling based algorithm is $O(|E|^2 n^4)$; on the other hand, we show in Section \ref{subsec:runtime} that the \LP based algorithm can be solved much faster, $\tilde{O}(|E|^2)$ time in the worst case and even faster than that in practice. (3) For the stochastic rewards setting, the offline problem is not known to be polynomial-time solvable\myEdit{, which is required for \cite{bib:Manshadi} since they rely on solving instances of the offline problem on simulated graphs}. \cite{assadi2016stochastic} show that under the assumption of constant $p$ \emph{and $\OPT = \omega(1/p)$}, we can obtain a $(1-\epsilon)$-approximation to the optimal solution. However, these assumptions are too strong to be used in our setting.

For the stochastic-rewards setting, one might be tempted to use an \LP to achieve the same property obtained from Monte-Carlo simulation via adding extra constraints. In the context of uniform stochastic rewards where each edge $e$ is associated with a uniform constant probability $p$, what we really need is:
 \begin{equation}
 \forall S \subseteq \pa(u), ~~ \sum_{e \in S} f_e \leq \frac{ 1-\exp(-|S|p)}{p}
 \label{cons:expstoch}
 \end{equation}
  To guarantee this via the \LP, a straightforward approach is to add this family of constraints to the \LP. However, the number of such constraints is exponential and there seems to be no \emph{obvious} separation oracle. We overcome this challenge by showing it suffices to ensure that inequality~\eqref{cons:expstoch} above holds for all $S$ with $|S| \le 2/p$, which is a constant, thus making the resultant LP polynomial-time solvable.



\subsection{Overview of edge-weighted algorithm and contributions.}

The previous best result due to~\cite{bib:Haeupler} for the edge-weighted problem was $0.667$. They used two matchings, $M_1$ and $M_2$, from the offline graph to guide the online algorithm and leverage the \textit{power of two choices}. When a vertex $v$ arrives for the first time, it can be matched to its neighbor in $M_1$ and on its second arrival it can be matched to its neighbor in $M_2$. However, these two matchings may not be edge disjoint, leaving some arriving vertices with only one choice. In fact, choosing two guiding matchings that maximize both the edge weights and the number of disjoint edges is a major challenge that arises in applying the \textit{power of two choices} to this setting.

When the same edge $(u, v)$ is included in both matchings $M_1$ and $M_2$, the copy of $(u, v)$ in $M_2$ can offer no benefit and a second arrival of $v$ is wasted. To use an example from related work, Haeupler \emph{et al.}~\cite{bib:Haeupler} choose two matchings in the following way. $M_1$ is attained by solving an LP with constraints~\eqref{cons:umatch},~\eqref{cons:vmatch} and~\eqref{cons:single} and rounding to an integral solution. $M_2$ is constructed by finding a maximum-weight matching and removing any edges which have already been included in $M_1$. A key element of their proof is showing that the probability of an edge being removed from $M_2$ is at most $1-1/\euler \approx 0.63$.

The approach in this paper is to construct two or three matchings together in a correlated manner to reduce the probability that some edge is included in all matchings. We show a general technique to construct an ordered set of $k$ matchings where $k$ is an easily adjustable parameter. For $k = 2$, we show that the probability of an edge appearing in both $M_1$ and $M_2$ is at most $1 - 2/\euler \approx 0.26$. 

For the algorithms presented, we first solve an LP on the input graph. We then round the LP solution vector to a sparse integral vector and use this vector to construct a randomly ordered set of matchings which will guide our algorithm during the online phase. We begin Section~\ref{sec:eweight} with a simple warm-up algorithm which uses a set of two matchings as a guide to achieve a $0.688$ competitive ratio, improving the best known result for this problem. We follow it up with a slight variation that improves the ratio to $0.7$ and a more complex $0.705$-competitive algorithm which relies on a convex combination of a $3$-matching algorithm and a separate \textit{pseudo-matching} algorithm.

\subsection{Overview of vertex-weighted algorithm and contributions.}

The previous best results due to~\cite{bib:Jaillet} for the vertex-weighted and unweighted problems were $0.725$ and $1 - 2\euler^{-2} \approx 0.729$, respectively. 
\myEdit{
They used a clever LP which guaranteed they could find a solution wherein each edge variable was assigned a value in $\{0, 1/3, 2/3\}$ as opposed to an arbitrary fractional value. This property, which we will call a $\{0, 1/3, 2/3\}$ solution, was required by their adaptive online algorithm. However, their special LP was a slightly weaker upper bound on the optimal solution than the LP we describe in Section~\ref{lp:stoch-match}.

Another key challenge encountered by~\cite{bib:Jaillet} was that solutions to their special LP could lead to length-four cycles of type $C_1$ shown in Figure~\ref{fig:JailletC1}. In fact, they used this case to show that no algorithm could perform better than $1 - 2\euler^{-2} \approx 0.7293$ using their LP as an upper bound. They mentioned that tighter LP constraints such as~\eqref{cons:single} and~\eqref{cons:pair} in the LP from Section~\ref{lp:stoch-match} could avoid this bottleneck, but did not propose a technique to use them. Note that the $\{0, 1/3, 2/3\}$ solution produced by their specific LP was an essential component of their \emph{Random List} algorithm.

To address this challenge, we show a randomized rounding algorithm to construct a similar, simplified $\{0, 1/3, 2/3\}$ vector from the solution of a stronger benchmark LP. This allows for the inclusion of additional constraints, most importantly constraint~\eqref{cons:pair}. Using our rounding algorithm combined with tighter constraints, we can upper-bound the probability of a vertex appearing in the cycle $C_1$ from Figure~\ref{fig:JailletC1} at $2 - 3/\euler \approx 0.89$ (See Lemma \ref{lem:fu=1}). By constant, cycles of type $C_1$ occur deterministically in~\cite{bib:Jaillet}.

Additionally, we note briefly that there are other length four cycles with different variable weights, defined as types $C_2$ and $C_3$ (See Figure~\ref{fig:JailletCycles} in Section~\ref{sec:vw-modification}). These cycles could be problematic, but we show how to deterministically break them in Section~\ref{sec:vw-modification} without creating any new cycles of type $C_1$ (This can happen if the cycle breaking is not done carefully). Finally, we describe an algorithm which utilizes these techniques to improve previous results in both the vertex-weighted and unweighted settings.
}

For this problem, we first solve the LP in Section~\ref{sec:prelim} on the input graph. In Section~\ref{sec:vweight}, we show how to use the technique in Section~\ref{subsec:drk} to obtain a sparse fractional vector. We then present a randomized online algorithm (similar to the one in~\cite{bib:Jaillet}) which uses the sparse fractional vector as a guide to achieve a competitive ratio of $0.7299$. 

Previously, there was a gap between the best unweighted algorithm with a ratio of $1 - 2\euler^{-2}$ due to~\cite{bib:Jaillet} and the negative result of $1 - \euler^{-2}$ due to~\cite{bib:Manshadi}. We take a step toward closing this gap by showing that an algorithm can achieve $0.7299 > 1 - 2\euler^{-2}$ for both the unweighted and vertex-weighted variants with integral arrival rates. In doing so, we make progess on Open Questions $3$ and $4$ from the book~\cite{mehtaBook}. \footnote{Open Questions~$3$ and $4$ state the following: ``In general, close the gap between the upper and lower bounds. In some sense, the ratio of $1 - 2\euler^{-2}$ achieved in~\cite{bib:Jaillet} for the integral case, is a nice `round' number, and one may suspect that it is the correct answer.''}

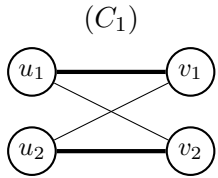
\begin{figure}
	\centering
	{\caption{This cycle is the source of the negative result described by Jaillet and Lu~\cite{bib:Jaillet}. \myEdit{It results from the edge variable assignments in their special LP.} Thick edges have $f_e = 2/3$ while thin edges have $f_e = 1/3$. \myEdit{This structure and variable assignment leads to a gap of $1 - 2\euler^{-2}$ between the LP solution and the best possible solution of any online algorithm.}}\label{fig:JailletC1}}
	{\begin{tikzpicture}
[
	xscale=1,yscale=1,auto,thick,
  	gray node/.style={circle,
  		inner sep=0pt,minimum size=14pt, 
  		fill=black!20,draw,font=\small},  
  	white node/.style={circle,
  		inner sep=0pt,minimum size=14pt, 
  		fill=white,draw,font=\small},
  node distance=24pt 
]


	\node [white node] (c1u1) at (0,0) {$u_1$};
	\node [white node] (c1u2) [below of=c1u1, xshift=0, yshift=0] {$u_2$};
	\node [white node] (c1v1) [right of=c1u1, xshift=30, yshift=0] {$v_1$};
	\node [white node] (c1v2) [right of=c1u2, xshift=30, yshift=0] {$v_2$};
	\draw[ultra thick] (c1u1)  -- (c1v1) node [midway, yshift=10] {($C_1$)};
	\draw[thin] (c1u1)  -- (c1v2);
	\draw[thin] (c1u2)  -- (c1v1);
	\draw[ultra thick] (c1u2)  -- (c1v2);

\end{tikzpicture}}
\end{figure}

\subsection{Overview of stochastic rewards  and contributions.}

Our algorithm for the more general problem allowing stochastic rewards and non-integral arrival rates (Algorithm~\ref{alg:non-integral}) is presented in Section~\ref{sec:nonint} and~\ref{sec:uniform_p}. We believe the known I.I.D. model with stochastic rewards is an interesting new direction motivated by the work of~\cite{mehta2012online} and~\cite{mehtaonline} in the adversarial model. We introduce a new, more general LP (see \LP~\eqref{lp2:stoch-match}) specifically for this setting and show that a simple algorithm using the LP solution directly can achieve a competitive ratio of $1 - 1/\euler$. This ratio is optimal among all non-adaptive algorithms for the case of non-integral arrival rates even without stochastic rewards~\cite{bib:Manshadi}. 
In~\cite{mehta2012online}, it is shown that no randomized algorithm can achieve a ratio better than 0.62 $< 1-1/\euler$ in the adversarial model \myEdit{when comparing to a problem called Budgeted Allocation as the offline optimal. While our work instead compares to offline stochastic matching as the offline optimal, the benchmark LP we use in Section~\ref{sec:nonint} (LP~\ref{lp2:stoch-match}) upper bounds Budgeted Allocation}. Hence, achieving $1-1/\euler$ for the I.I.D. model shows that this lower bound does not extend to the I.I.D. model. Further, the paper \cite{brubach2017} shows that using  \LP~\eqref{lp2:stoch-match} one cannot achieve a ratio better than $1-1/\euler$. We discuss some challenges relating to why the techniques used in prior work do not directly extend to this model. 

To take a step toward addressing these challenges in Section~\ref{sec:uniform_p}, we consider a restricted version of the problem where each edge is unweighted and has a uniform constant probability $p \in (0,1]$ under integral arrival rates. By proposing a family of valid constraints, we are able to show that in this restricted setting, one can indeed beat $1-1/\euler$. \myEdit{We note that this result cannot be compared to the work of~\cite{mehta2012online} since we use a tighter benchmark (LP~\ref{lp3:stoch-match}) which does not upper bound Budgeted Allocation.} 

\begin{table*}[!t]
\caption{Summary of Contributions}\label{summary}
\setlength{\tabcolsep}{\textwidth/80}
\renewcommand{\arraystretch}{2}
\begin{tabular}[!t]{|c|c|c|}
	\hline
	Problem & Previous Work & This Paper \\ \hline
	Edge-Weighted (Section \ref{sec:eweight}) & $0.667$ \cite{bib:Haeupler} & $0.705$  \\ \hline
	Vertex-Weighted (Section \ref{sec:vweight}) & $0.725$ \cite{bib:Jaillet} & $0.7299$ \\ \hline
	Unweighted & $1-2/\euler^2$ \cite{bib:Jaillet} & $0.7299~(> 1-2/\euler^2)$ \\
	\hline
	\multirow{ 2 }{*}{Stochastic Rewards (Section \ref{sec:nonint} and \ref{sec:uniform_p})} & \multirow{2}{*}{N/A} & $1-\euler^{-1} $ for general version \\
	&& $0.702$ for the restricted version\\ \hline
\end{tabular} 
\end{table*}


\subsection{Running time of the algorithms.}
\label{subsec:runtime}
In this section, we discuss the implementation details of our algorithms. All of our algorithms solve an \LP in the pre-processing step.
The dimension of the \LP is determined by the constraint matrix which consists of $O(|E|^2 + |U| + |V|)$ rows and $O(|E|)$ columns. However, note that the number of non-zero entries in this matrix is of the order $O(|E|^2)$ because each edge is subject to $O(|E|)$ constraints primarily due to LP constraint~\ref{cons:pair}. Some recent work (\emph{e.g.,} \cite{lee2015efficient}) shows that such sparse programs can be solved in time $\tilde{O}(|E|^2)$ using interior point methods (which are known to perform very well in practice). This sparsity in the \LP is the reason we can solve very large instances of the problem. The second critical step in pre-processing is to perform randomized rounding. Note that we have $O(|E|)$ variables and that in each step of the randomized rounding due to \cite{bib:Gandhi}, they incur a running time of $O(|E|)$. Hence the total running time to obtain a rounded solution is of the order $O(|E|^2)$. Additionally, both of these operations are part of the pre-processing step. In the online phase, the algorithm incurs a per-time-step running time of at most $O(|U|)$ for the stochastic rewards case (in fact, a smarter implementation using binary search runs as fast as $O(\log |U|)$ time) and $O(1)$ for the edge-weighted and vertex-weighted algorithms in Sections~\ref{sec:eweight} and \ref{sec:vweight}.

\subsection{LP rounding technique $\DR{k}$.}
\label{subsec:drk}

For the algorithms presented, we first solve the benchmark LP in sub-section \ref{lp:stoch-match} for the input instance to get a fractional solution vector $\ff$. We then round $\ff$ to an integral solution $\FF$ using a two step process we call $\DR{k}$. The first step is to multiply $\ff$ by $k$. The second step is to apply the dependent rounding techniques of Gandhi, Khuller, Parthasarathy, and Srinivasan~\cite{bib:Gandhi} to this new vector. In this paper, it suffices to consider $k =2$ or $k=3$. 

While dependent rounding is typically applied to values between $0$ and $1$, the useful properties extend naturally to our case in which $k f_e$ may be greater than $1$ for some edge $e$. To understand this process, it is easiest to imagine splitting each $k f_e$ into two edges with the integer value $f'_{e} = \floor{k f_e}$ and fractional value $f''_{e} = k f_e - \floor{k f_e}$. The former will remain unchanged by the dependent rounding since it is already an integer while the latter will be rounded to $1$ with probability $f''_e$ and $0$ otherwise. Our final value $F_e$ would be the sum of those two rounded values. The two properties of dependent rounding we use are:
\begin{enumerate}
	\item 
		\textbf{Marginal distribution: } For every edge $e$, let $p_e = k f_e - \floor{k f_e}$. Then, $\Pr[F_e = \ceil{k f_e}] = p_e$ and $\Pr[F_e = \floor{k f_e}] = 1 - p_e$.
	\item
		\textbf{Degree-preservation: } For any vertex $w \in U \cup V$, let its fractional degree $k f_w$ be $\sum_{e \in \partial(w)} k f_e$ and integral degree be the random variable $F_w = \sum_{e \in \partial(w)} F_e$. Then $F_w \in \{\floor{k f_w}, \ceil{k f_w}\}$. 
\end{enumerate}

\subsection{Related work.} 
 The study of online matching began with the seminal work of Karp, Vazirani, Vazirani \cite{kvv}, where they gave an optimal online algorithm for a version of the unweighted bipartite matching problem in which vertices arrive in adversarial order. Following that, a series of works have studied various related models. The book by Mehta \cite{mehtaBook} gives a detailed overview. The vertex-weighted version of this problem was introduced by Aggarwal, Goel and Karande \cite{aggarwalVertex}, where they give an optimal $\left( 1- \frac{1}{\euler} \right)$ ratio for the adversarial arrival model. The edge-weighted setting has been studied in the adversarial model by Feldman, Korula, Mirrokni and Muthukrishnan \cite{feldmanEdgeWeighted}, where they consider an additional relaxation of ``free-disposal".

In addition to the adversarial and known I.I.D. models, online matching is also studied under several other variants such as random arrival order, unknown distributions, and known adversarial distributions. In the setting of random arrival order, the arrival sequence is assumed to be a random permutation over all online vertices, see e.g., \cite{devanur2009adwords,kesselheim2013optimal,korula2009algorithms,mahdian2011online}. In the case of unknown distributions, in each round an item is sampled from a fixed but unknown distribution. If the sampling distributions are required to be the same during each round, it is called  unknown I.I.D. (\cite{devanur2011near, devanur2012asymptotically}); otherwise, it is called adversarial stochastic input (\cite{devanur2011near}). As for known adversarial  distributions, in each round an item is sampled from a known distribution, which is allowed to change over time (\cite{alaei2012online, alaei2013online}). Another variant of this problem is when the edges have stochastic rewards. Models with stochastic rewards have been previously studied by \cite{mehta2012online, mehtaonline} among others, but not in the known I.I.D. model.


\xhdr{Related work in the vertex-weighted/unweighted settings.} The vertex-weighted and unweighted settings have many results starting with Feldman, Mehta, Mirrokni and Muthukrishnan \cite{bib:Feldman} who were the first to beat $1-1/\euler$ with a competitive ratio of $0.67$ for the unweighted problem. This was improved by Manshadi, Gharan, and Saberi~\cite{bib:Manshadi} to $0.705$ with an adaptive algorithm. In addition, they showed that even in the unweighted variant with integral arrival rates, no algorithm can achieve a ratio better than $1 - \euler^{-2} \approx 0.86$. Finally, Jaillet and Lu~\cite{bib:Jaillet} presented an adaptive algorithm which used a clever LP to achieve $0.725$ and $1 - 2\euler^{-2} \approx 0.729$ for the vertex-weighted and unweighted problems, respectively. 

\xhdr{Related work in the edge-weighted setting.}
For this model, Haeupler, Mirrokni, Zadimoghaddam~\cite{bib:Haeupler} were the first to beat $1-1/\euler$ by achieving a competitive ratio of $0.667$. They use a \textit{discounted LP} with tighter constraints than the basic matching LP (a similar LP can be seen in \ref{lp:stoch-match}) and they employ the \textit{power of two choices} by constructing two matchings offline to guide their online algorithm.

 \xhdr{Other related work.}
Devanur \emph{et al} \cite{devanur2012asymptotically} gave an algorithm which achieves a ratio of $1-k!/(k^k e^k)$ for the Adwords problem \footnote{\myEdit{In the Adwords problem, every $u \in U$ has a total budget $B_u$. Each edge $e$ has a bid $b_e$ which represents the weight. The goal is to maximize the total weight of the edges matched subject to the constraint that for any vertex $u \in U$ the sum of the total weight of the edges matched to $u$ is at most $B_u$.}}  in the Unknown I.I.D. arrival model with knowledge of the optimal budget utilization and when the bid-to-budget ratios are at most $1/k$, where $k$ is some positive integer. Alaei \emph{et al.} \cite{alaei2012online} considered the Prophet-Inequality Matching problem, in which $v$ arrives from a distinct (known) distribution $\mathcal{D}_t$, in each round $t$. They gave a $1-1/\sqrt{k+3}$ competitive algorithm, where $k$ is the minimum capacity of $u$.

\section{Edge-weighted matching with integral arrival rates}
\label{sec:eweight}

	\myEdit{As mentioned in the introduction, we make the simplifying assumption that the arrival rates $r_v=1$ for every online vertex $v \in V$. Our techniques critically relies on this assumption; in particular the rounding procedure breaks when $r_v$ is fractional. If all the arrival rates are fractional, one can still employ the rounding procedure by approximating it to the \emph{nearest} rational and lose a tiny constant in the competitive ratio. However, the worst-case is when all arrival rates are $\frac{1}{n}$ in which case the competitive ratio obtained by working with an approximate arrival rate is $o(1)$.}

\subsection{Warm-up: 0.688-competitive algorithm.}

As a warm-up, we describe a simple algorithm which achieves a competitive ratio of 0.688 and introduces the key ideas in our approach.  We begin by solving the LP in sub-section \ref{lp:stoch-match} to get a fractional solution vector $\ff$ and applying $\DR{2}$ as described in Subsection~\ref{subsec:drk} to get an integral vector $\FF$. We construct a bipartite graph $G_{\FF}$ with $F_e$ copies of each edge $e$. Note that $G_{\FF}$ will have max degree $2$ since for all $w \in U \cup V$, $F_w \leq \ceil{2 f_w} \leq 2$ and thus we can decompose it into two matchings using a greedy algorithm and \emph{Hall's Theorem}. \myEdit{The exact choice of the two matchings is not critical to the algorithm as long as the union contains all edges in $G_{\FF}$}. Finally, we randomly permute the two matchings into an ordered pair of matchings, $[M_1, M_2]$. These matchings serve as a guide for the online phase of the algorithm, similar to \cite{bib:Haeupler}. 
The entire warm-up algorithm for the edge-weighted model, denoted by $\mathsf{EW}_{0}$, is summarized in Algorithm \ref{alg:ew0}.

	\IncMargin{1em}
	\begin{algorithm}[!h]
		\caption{\protect{$[\mathsf{EW}_{0}$]}}
		\label{alg:ew0}
		\DontPrintSemicolon
		Construct and solve the benchmark LP in sub-section \ref{lp:stoch-match} for the input instance. \;
		Let ${\bf f}$ be an optimal fractional solution vector. Call $\DR{2}$ to get a \myEdit{random} integral vector $\FF$. \;
		Create the graph $G_{\FF}$ with $F_e$ copies of each edge $e \in E$ and decompose it into two matchings as described in text. \;
		Randomly permute the matchings to get a {\it random ordered} pair of matchings, say $[M_{1}, M_{2}]$.  \;
		When a vertex $v$ arrives for the first time, attempt to match $v$ to $u_{1}$ if $(u_{1},v) \in M_{1}$; when $v$ arrives for the second time, attempt to match $v$ to $u_{2}$ if $(u_{2},v) \in M_{2}$.\;
		When a vertex $v$ arrives for the third time or more, do nothing in that step.
	\end{algorithm}
	\DecMargin{1em}

\subsubsection{Analysis of algorithm $\mathsf{EW}_{0}$.}
\label{subsec:analysisEW0}
 
We will show that $\mathsf{EW}_{0}$ (Algorithm~\ref{alg:ew0}) achieves a competitive ratio of $0.688$. Let $[M_{1}, M_{2}]$ be our randomly ordered pair of matchings. Note that there might exist some edge $e$ which appears in both matchings due to having $f_{e} > 1/2$, which could be rounded up to $F_e = 1$. Therefore, we consider three types of edges. We say an edge $e$ is of type $\psi_{1}$, denoted by $e \in \psi_{1}$, if and only if $e$ appears \textit{only} in $M_{1}$. Similarly $e \in \psi_{2}$, if and only if $e$ appears \textit{only} in $M_{2}$. Finally, $e \in \psi_{b}$, if and only if $e$ appears in \textit{both} $M_{1}$ and $M_{2}$. 
Let $P_{1}$, $P_{2}$, and $P_{b}$ be the probabilities of getting matched for $e \in \psi_{1}$, $e \in \psi_{2}$, and $e \in \psi_{b}$, respectively. According to the result in Haeupler \emph{et al.}~\cite{bib:Haeupler}, Lemma~\ref{lem:ew0} bounds these probabilities.

\begin{lemma}[Proof details in section 3 of \cite{bib:Haeupler}]\label{lem:ew0}
\myEdit{For any two matchings $M_1$ and $M_2$ steps (5) and (6) in Algorithm~\ref{alg:ew0} implies that we have (1) $P_{1} > 0.5808$; (2) $P_{2} > 0.14849$ and (3) $P_{b} >0.6321$.}
\end{lemma}

We can use Lemma~\ref{lem:ew0} to prove that the warm-up algorithm $\mathsf{EW}_{0}$ achieves a ratio of $0.688$ by examining the probability that a given edge becomes type $\psi_{1}$, $\psi_{2}$, or $\psi_{b}$. \\

	\xhdr{Analysis of $\mathsf{EW}_{0}$.}
Consider the following two cases.
\begin{itemize}
	\item 
		\textbf{Case 1: } $0\le f_{e} \le 1/2$: By the marginal distribution property of dependent rounding, there can be at most one copy of $e$ in $G_{\FF}$ and the probability of including $e$ in $G_{\FF}$ is $2 f_e$. Since an edge in $G_{\FF}$ can appear in either $M_1$ or $M_2$ with equal probability $1/2$, we have $\Pr[e \in \psi_{1}]=\Pr[e \in \psi_{2}]=f_{e}$. Thus, the ratio is $(f_{e}P_{1}+f_{e}P_{2})/f_{e}=P_{1}+P_{2} = 0.729$.
	\item
			\textbf{Case 2: } $1/2\le f_{e} \le 1-\frac{1}{\euler}$: Similarly, by marginal distribution, $\Pr[e \in \psi_{b}] = \Pr[F_e = \ceil{2 f_e}] = 2 f_e - \floor{2 f_e} = 2f_{e} - 1$. It follows that $\Pr[e \in \psi_{1}] = \Pr[e \in \psi_{2}] = (1/2)(1-(2f_{e} - 1)) = 1 - f_e$. Thus, the ratio is (noting that the first term is from case 1 while the second term is from case 2) 
			$((1-f_{e})(P_{1}+P_{2})+(2f_{e}-1)P_{b})/f_{e} \ge 0.688 $, where the $\WS$ is for an edge $e$ with $f_{e}=1-\frac{1}{\euler}$. \qed
\end{itemize}

\subsection{Improved algorithm: 0.7-competitive algorithm.}
\label{subsec:attenuation}

In this section, we describe an improvement upon the previous warm-up algorithm to get a competitive ratio of $0.7$. We start by making an observation about the performance of the warm-up algorithm. After solving the LP, let edges with $f_e > 1/2$ be called \textit{large} and edges with $f_e \leq 1/2$ be called \textit{small}. Let $L$ and $S$, be the sets of large and small edges, respectively. Notice that in the previous analysis, small edges achieved a much higher competitive ratio of $0.729$ versus $0.688$ for large edges. This is primarily due to the fact that we may get two copies of a large edge in $G_{\FF}$. In this case, the copy in $M_1$ has a better chance of being matched, since there is no edge which can ``block'' it (i.e. an edge with the same offline neighbor that gets matched first), but the copy that is in $M_2$ has no chance of being matched.

To correct this imbalance, we make an additional modification to the $f_e$ values \textit{before} applying $\DR{k}$. The rest of the algorithm is exactly the same. \myEdit{Let $\eta$ be a parameter to be optimized in the analysis. For all large edges $\ell \in L$ such that $f_{\ell}^* > 1/2$, we set $\tilde{f}_{\ell}^*(\ell) = f_{\ell}^* + \eta$. For all small edges $s \in S$ which are adjacent to some large edge, let $\ell \in L$ be the largest edge adjacent to $s$ such that $f_{\ell}^* > 1/2$. Note that it is possible for $s$ to have two large neighbors, but we only care about the larger of the two. We set $\tilde{f}_s^* = f_s^* \left( \frac{1-\tilde{f}_{\ell}^*}{1 - f_{\ell}^*} \right)$.}

In other words, we increase the values of large edges while ensuring that for all $w \in U \cup V$, $f_w \leq 1$ by reducing the values of neighboring small edges proportional to their original values. Note that it is not possible for two large edges to be adjacent since they must both have $f_e > 1/2$. For all other small edges which are not adjacent to any large edges, we leave their values unchanged. We then apply $\DR{2}$ to this new vector, multiplying by 2 and applying dependent rounding as before.

\subsubsection{Analysis.}

\begin{theorem}
	\label{thm:ewSimple}
\myEdit{For edge-weighted online stochastic matching with integral arrival rates, $\mathsf{EW}(0.0142)$ achieves a competitive ratio of at least $0.7$.}
\end{theorem}

\myEdit{
\begin{proof}
	As in the warm-up analysis, we'll consider large and small edges separately
	
		\begin{itemize}
		\item \textbf{Scenario 1:} $0 \leq f_s^* \leq \frac{1}{2}$: 
		
		Here we have two cases
		\begin{itemize}
			\item \textbf{Case 1:} $s$ is not adjacent to any large edges.
				
				In this case, the analysis is the same as Case 1 in the warm-up analysis.  Thus, the probability that edge $s$ is added to the matching is $0.729 f_e^*$.
			
			\item \textbf{Case 2:} $s$ is adjacent to some large edge $\ell$.
			
			For this case, let $f_{\ell}^*$ be the value of the largest neighboring edge in the original LP solution. Then the probability that edge $s$ is added to the matching is
			\[
				f_s^* \left( \frac{1 - (f_{\ell}^* + \eta)}{1 - f_{\ell}^*} \right) (0.1484 + 0.5803).
			\]
			This follows from Lemma~\ref{lem:ew0}; in particular, the first two terms are the result of how we set $\tilde{f}_s$ in the algorithm, while the two numbers, $0.1484$ and $0.5803$, are the probabilities that $s$ is matched when it is in $M_2$ and $M_1$, respectively.
			Note that for $f_{\ell}^* \in [0,1)$ this is a decreasing function in $f_{\ell}^*$. So the worst case is when $f_{\ell}^* = 1-\frac{1}{\euler}$ (due to third constraint in the LP~\eqref{cons:single}) Thus, the probability that edge $s$ is added to the matching is
			\begin{equation*}
				f_s^* \left( \frac{1 - (1-\frac{1}{\euler} + \eta)}{1 - (1 - \frac{1}{\euler})} \right) (0.1484 + 0.5803).
			\end{equation*}
			Since $\eta = 0.0142$, this evaluates to,
			\begin{equation}
				\label{eq:EWCase1Final}
				0.701 f_s^*.
			\end{equation}
		\end{itemize}
		
		\item $\frac{1}{2} < f_{\ell}^* \leq 1-\frac{1}{\euler}$:
		Here, the probability that $\ell$ is added to the matching is,
		$[1 - (f_{\ell}^*(\ell) + \eta)][P_{1}+P_{2}]+[2(f_{\ell}^* + \eta) - 1]P_{b}$. This can re-arranged to obtain 
		\begin{equation}
			\label{eq:EWPart1}
				(P_1 + P_2)(1-\eta) + (2 \eta - 1) P_b + f_{\ell}^*[2 P_b - P_1 - P_2].
		\end{equation}
		Since $\eta = 0.0142$ using Lemma~\ref{lem:ew0} we have $(P_1 + P_2)(1-\eta) + (2 \eta - 1) P_b = 0.1048$. Similarly, using Lemma~\ref{lem:ew0} we have $2 P_b - P_1 - P_2 = 0.535$. Thus, Eq.~\eqref{eq:EWPart1} simplifies to,
		\begin{equation}
			\label{eq:EWPart2}
			0.1048 + f_{\ell}^* 0.535
		\end{equation}
		We can write Eq.~\eqref{eq:EWPart2} as $f_{\ell}^* [0.1048/f_{\ell}^* + 0.535]$. Note that $\frac{1}{2} < f_{\ell}^* \leq 1-\frac{1}{\euler}$. Thus, Eq.~\eqref{eq:EWPart2} can be lower-bounded by 
		\begin{equation}
			\label{eq:EWCase2Final}
			0.701 f_{\ell}^* .
		\end{equation}
		\end{itemize}
		
		Thus combining Eq.~\eqref{eq:EWCase1Final} and \eqref{eq:EWCase2Final} with Lemma~\ref{lem:edgeCR} we get a competitive ratio of $0.7$.
	
		We now show that the chosen value of $\eta = 0.0142$ ensures that both $\tilde{f}_{\ell}^*$ and $\tilde{f}_s^*$ are less than $1$ \emph{after} modification. Since $f_{\ell}^* \leq 1-\frac{1}{\euler}$ we have that $f_{\ell}^* + \eta \leq 1-\frac{1}{\euler}+0.0142 \leq 1$. Note that $f_{\ell}^* \geq 1/2$. Hence, the modified value $\tilde{f}_s^*$ is always less than or equal to the original value, since $\left( \frac{1-(f_{\ell}^* + \eta)}{1 - f_{\ell}^*} \right)$ is decreasing in the range $f_{\ell}^* \in [1/2, 1-\frac{1}{\euler}]$ and has a value less than $0.98$ at $f_{\ell}^* = 1/2$.
		\end{proof}
	}

\subsection{Final algorithm: Roadmap.}

In the next few subsections, we describe our final edge-weighted algorithm with all of the attenuation factors. To keep it modular, we give the following guide to the reader. We note that the definition of large and small edges given below in Subsection~\ref{apx:hybrid-alg} is different from the definition in the previous subsection.

\label{guide:edge}

\begin{itemize}
	\item 
	Section~\ref{apx:hybrid-alg} describes the main algorithm which internally invokes two algorithms, $\mathsf{EW}_{1}$ and $\mathsf{EW}_{2}$, which are described in sections~\ref{append:ew1} and \ref{append:ew2}, respectively. 
	\item
	Theorem \ref{thm:ew7} proves the final competitive ratio. This proof depends on the performance guarantees of $\mathsf{EW}_{1}$ and 
	$\mathsf{EW}_{2}$, which are given by Lemmas~\ref{lem:ew1} and \ref{lem:ew2}, respectively.
	\item
	The proof of Lemma \ref{lem:ew1} depends on claims \ref{cl:ew1-a}, \ref{cl:ew1-b}, and \ref{cl:ew1-c} (Found in the Appendix). Each of those claims is a careful case-by-case analysis. Intuitively, \ref{cl:ew1-a} refers to the case where the offline vertex $u$ is incident to one large edge and one small edge (here the analysis is for the large edge), \ref{cl:ew1-b} refers to the case where $u$ is incident to three small edges and \ref{cl:ew1-c} refers to the case where $u$ is incident to a small edge and large edge (here the analysis is for the small edge).
	\item
	The proof of Lemma \ref{lem:ew2} depends on claims \ref{cl:ew2-a} and \ref{cl:ew2-b} (Found in the Appendix). Again, both of those claims are proven by a careful case-by-case analysis. Since there are many cases, we have given a diagram of the cases when we prove them.
\end{itemize}

\subsubsection{Algorithm $\EW$: 0.705-competitive algorithm.}
\label{apx:hybrid-alg}

In this section, we describe an algorithm $\EW$ (Algorithm ~\ref{alg:ew}), that achieves a competitive ratio of 0.705. The algorithm first solves the benchmark LP in subsection \ref{lp:stoch-match} and obtains a fractional optimal solution $\ff$. By invoking $\DR{3}$, it obtains a random integral solution $\FF$. Notice that from LP Constraint~\eqref{cons:single} we see $f_e \leq 1-\frac{1}{\euler} \leq 2/3$. Therefore after $\DR{3}$, each $F_e \in \{0, 1, 2\}$. We say an edge $e$ is {\it large} if $F_e=2$ and {\it small} if $F_e=1$ (note that this differs from the definition of large and small in Subsection~\ref{subsec:attenuation}). 

We design two non-adaptive algorithms, denoted by $\EW_{1}$ and $\EW_{2}$, which take the sparse graph $G_{\FF}$ as input. 
The difference between the two algorithms $\EW_1$ and $\EW_2$ is that $\EW_{1}$ favors the \emph{small edges} while $\EW_{2}$ favors the \emph{large edges}. The final algorithm is to take a convex combination of $\EW_{1}$ and $\EW_{2}$, \ie run $\EW_1$ with probability $q$ and $\EW_2$ with probability $1-q$.

\begin{theorem}
	\label{thm:ew7}
For edge-weighted online stochastic matching with integral arrival rates, the algorithm $\mathsf{EW}[q]$  with $q=0.149251$ achieves a competitive ratio of at least $0.70546$.
\end{theorem}	

\IncMargin{1em}
\begin{algorithm}[!h]
	\caption{\protect{$\EW$[q]}}\label{alg:ew}
	\DontPrintSemicolon
	Solve the benchmark LP in sub-section \ref{lp:stoch-match} for the input. Let $\ff$ be the optimal solution vector. \;
	Invoke $\DR{3}$ to obtain the vector $\FF$. \;
	Independently run \myEdit{$\mathsf{EW}_{1}[0.538]$ and $\mathsf{EW}_{2}[0.687, 1]$} with probabilities $q$ and $1-q$ respectively on $G_{\FF}$.
\end{algorithm}
\DecMargin{1em}
The details of algorithm $\EW_1$ and $\EW_2$ and the proof of Theorem  \ref{thm:ew7} are presented in the following sections.

\subsubsection{Sub-routine 1: algorithm $\EW_{1}$.}
\label{append:ew1}
In this section, we describe the randomized algorithm $\EW_{1}$ (Algorithm~\ref{alg:ew1}). Let $\PM{3}$ refer to the process of constructing the graph $G_\FF$ with $F_e$ copies of each edge $e$, decomposing it into three matchings\footnote{\myEdit{The natural way of repeatedly computing the maximum matching can be used to find this decomposition.}}, and randomly permuting the matchings. 
$\EW_{1}$ first invokes $\PM{3}$ to obtain a \emph{random ordered} triple of matchings, say $[M_{1}, M_{2}, M_{3}]$. Notice that from LP Constraint~\eqref{cons:single} and the properties of $\DR{3}$ and $\PM{3}$, an edge will appear in at most two of the three matchings. For a small edge $e=(u,v)$ with $F_e=1$, we say $e$ is of type $\Gamma_{1}$ if  $u$ has two other neighbors $v_1$ and $v_2$ with $F_{(u,v_1)}=F_{(u,v_2)}=1$.
We say $e$ is of type $\Gamma_{2}$ if $u$ has exactly one other neighbor $v_{1}$ with $F_{(u,v_{1})}=2$. WLOG we can assume that for every $u$, $F_{u}=\sum_{e \in \partial(u)} F_e=3$; otherwise, we can add a dummy node $v'$ to the neighborhood of $u$. \myEdit{Similarly, we assume $F_v=\sum_{e \in \partial(v)} F_e=3$ by adding dummy nodes $u'$. Note that when we assign $v$ to a dummy node $u'$, it essentially means rejection of $v$ when it arrives. Since all $v$ has $F_v=3$, we can safely assume that each $v$ has exactly one edge in each of the three matchings output by $\PM{3}$.} We use the terminology, *assign $v$ to $u$*,  to denote that \emph{we assign $v$ to $u$ if $u$ is not matched and reject $v$ otherwise}. 

\IncMargin{1em}
\begin{algorithm}[!h]
	\caption{\protect{$\EW_{1}[h]$}}
	\label{alg:ew1}
	\DontPrintSemicolon
	Invoke $\PM{3}$ to obtain a \emph{random ordered } triple matchings, say $[M_{1}, M_{2}, M_{3}]$. \;

	When a vertex $v$ comes for the first time, assign $v$ to $u_{1}$ with $(u_{1},v) \in M_{1}$. \;
	When $v$ comes for the second time, assign $v$ to $u_{2}$ with $(u_{2},v) \in M_{2}$.\;
	\myEdit{When $v$ comes for the third time, let $e=(u_3,v)$ be the edge in $M_3$. If $e$ is either a large edge or a small edge of type $\Gamma_{1}$ then assign $v$ to $u_{3}$. However, if $e$ is a small edge of type $\Gamma_{2}$, then \textit{with probability $h$} assign $v$ to $u_{3}$ and do nothing otherwise.} \;
	When $v$ comes for the fourth or more time, do nothing in that step.\;
\end{algorithm}
\DecMargin{1em}

Let $\RR{\EW_{1}}{1/3}$ and $\RR{\EW_{1}}{2/3}$ be the competitive ratio for a small edge and large edge respectively.  

\begin{lemma}\label{lem:ew1}
	For $h=0.537815$, $\EW_{1}[h]$ achieves a competitive ratio $\RR{\EW_{1} }{2/3}=0.679417$,   $\RR{\EW_{1}} { 1/3} =0.751066$ for a large and small edge respectively.
\end{lemma}		

\begin{proof}
	In case of the large edge $e$, we divide the analysis into three cases where each case corresponds to $e$ being in one of the three matchings. And we combine these conditional probabilities using Bayes' theorem to get the final competitive ratio for $e$. For each of the two types of small edges, we similarly condition them based on the matching they can appear in, and combine them using Bayes' theorem.
	A complete version of proof can be found in Section \ref{apx:analy-EW1} of Appendix.
\end{proof}					


\subsubsection{Sub-routine 2: algorithm $\EW_{2}$.}
\label{append:ew2}

Algorithm $\EW_{2}$ (Algorithm~\ref{alg:ew2}) is a non-adaptive algorithm which takes $G_{\FF}$ as input and performs well on the \emph{large edges} with $F_e=2$. 
\myEdit{Recall that $\FF$ is an integral vector output by $\DR{3}$ with $F_{e}\in \{0, 1, 2\}$ for each $e$. WLOG, we can assume that $F_{v}=3$ for every $v$ in $G_{\FF}$; otherwise we can add \emph{dummy} vertices to ensure the case. Unlike $\EW_{1}$, $\EW_{2}$ will invoke a routine, denoted by $\mathsf{PM}^{*}[\FF,2]$ (Algorithm~\ref{PM:EW2}), to generate a pair of pseudo matchings from $\FF$.} 


\IncMargin{1em}
\begin{algorithm}[!h]
	
	\caption{\protect{$\mathsf{PM}^{*}[\FF,2][y_{1},y_{2}]$}}
	\label{PM:EW2}
	\DontPrintSemicolon
	Suppose $v$ has two neighbors, say $u_{1}, u_{2}$, with $e_{1}=(u_{1},v)$ being a large edge while $e_{2}=(u_{2},v)$ being a small edge. Add $e_{1}$ to the primary matching $M_{1}$ and $e_{2}$ to the secondary matching $M_{2}$.\;
	\myEdit{Suppose $v$ has three neighbors and the incident edges are $\partial(v)=(e_{1},e_{2},e_{3})$.  Take a random permutation of $\partial(v)$, say $(\pi_{1},\pi_{2},*) \in \Pi(\partial(v))$. \myEdit{Independently} add $\pi_{1}$ to $M_{1}$ with probability $y_{1}$ and $\pi_{2}$ to $M_{2}$ with probability $y_{2}$. (We use only the first two edges in the random permutation.)} 
\end{algorithm}

\myEdit{Note that the pair of matchings generated by $\mathsf{PM}^{*}[\FF,2]$ can be pseudo-matchings. Consider the following case: (1) $v$ has a large edge $e=(u,v)$; (2) $u$ has a small edge $e'=(u,v')$ other than $e$; and (3) $v'$ has two other small edges excluding $e'$. From $\mathsf{PM}^{*}[\FF,2][y_{1},y_{2}]$, we see that with probability $1$, $e =(u,v)\in M_1$ and with probability $y_1/3$ ($e'$ appears first in the random permutation and get selected in $M_1$), $e'=(u,v') \in M_1$. In that case, $u$ will have two neighbors in $M_1$.}

Algorithm \ref{alg:ew2} describes $\EW_{2}$ which uses Algorithm~\ref{PM:EW2} as a sub-routine.

\begin{algorithm}[!h]
	\caption{\protect{$\mathsf{EW}_{2}[y_{1},y_{2}]$}}\label{alg:ew2}
	\DontPrintSemicolon
	Invoke $\mathsf{PM}^{*}[\FF,2][y_{1},y_{2}]$ to generate a {\it random ordered} pair of pseudo-matchings, say $[M_{1}, M_{2}]$.  \;
	When a vertex $v$ comes for the first time, assign $v$ to some $u_{1}$ if $(u_{1},v) \in M_{1}$; When $v$ comes for the second time, try to assign $v$ to some $u_{2}$ if $(u_{2},v) \in M_{2}$.\;
	When a vertex $v$ comes for the third or more time, do nothing in that step.
\end{algorithm}

Let $\RR{\EW_{2}}{1/3}$ and $\RR{\EW_{2}}{2/3}$ be the competitive ratios for small edges and large edges, respectively.  

\begin{lemma}\label{lem:ew2}
	For $y_{1}=0.687$ and $y_{2}=1$, $\EW_{2}[y_{1},y_{2}]$ achieves a competitive ratio of $\RR{\EW_{2}}{ 2/3} = 0.8539 $ and $\RR{\EW_{2}}{ 1/3} = 0.4455 $ for a large and small edge respectively.
\end{lemma}		

\begin{proof}
	We analyze this on a case-by-case basis by considering the local neighborhood of the edge. A large edge can have two possible cases in its neighborhood, while a small edge can have eight possible cases. This is because of the fact that a large edge can have only small edges in its neighborhood while a small edge can have both large and small edges in its neighborhood. Choosing the worst case among the two for large edge and the worst case among the eight for the small edge, we prove the claim. Complete details of the proof can be found in section \ref{apx:analy-EW2} of Appendix.
\end{proof}		

\subsubsection{Convex combination of $\EW_{1}$ and $\EW_{2}$.}	
In this section, we prove Theorem \ref{thm:ew7}.

\begin{proof}
	Let $(a_{1}, b_{1})$ be the competitive ratios achieved by $\EW_{1}$ for large and small edges, respectively. \myEdit{From Lemma~\ref{lem:ew1} we have that $a_1 = 0.751$ and $b_1 = 0.679$.} Similarly, let $(a_{2},b_{2})$ denote the same for $\EW_{2}$. \myEdit{From Lemma~\ref{lem:ew2} we have $a_2 = 0.854$ and $b_2 = 0.445$.}
	
	We have the following two cases.
	\begin{itemize}
		\item $0 \leq f_e \leq \frac{1}{3}$: By marginal distribution property of $\DR{3}$,  we know that $\Pr[F_{e}=1]=3f_{e}$. Thus, the final ratio is 
		$$3f_{e}(qb_{1}/3+(1-q)b_{2}/3)/f_{e}=qb_{1}+(1-q)b_{2}$$
		
		\item $1/3 \leq f_e \leq 1-\frac{1}{\euler}$: By the same properties of $\DR{3}$, we know that $\Pr[F_{e}=2]=3f_{e}-1$ and 
		$\Pr[F_{e}=1]=2-3f_{e}$. Thus, the final ratio is
		$$\Big( (3f_{e}-1)(2qa_{1}/3+2(1-q)a_{2}/3)+(2-3f_{e})(qb_{1}/3+(1-q)b_{2}/3)  \Big)/f_{e}$$
	\end{itemize}
	
	The competitive ratio of the convex combination is maximized at $q=0.149251$ with a value of $0.70546$.
\end{proof}

\section{Vertex-weighted stochastic I.I.D. matching with integral arrival rates.}
\label{sec:vweight}

In this section, we consider vertex-weighted online stochastic matching on a bipartite graph $G$ under the known I.I.D. model with integral arrival rates. We present an algorithm in which each offline vertex $u$ has a competitive ratio of at least $0.72998 > 1-2\euler^{-2}$. 

\myEdit{Recall from Section~\ref{subsec:drk} that after invoking $\DR{3}$, we can obtain  a ({\it random}) integral vector ${\FF}$ with $F_e \in \{0, 1, 2\}$. Define $\HH=\FF/3$ and thus $H_e \in \{0, 1/3, 2/3\}$. Notice that $F_u =\sum_{e \in \pa(u)} F_e\le 3$  due to the degree preservation property from $\DR{3}$ and $H_u\doteq \sum_{e \in \pa(u)} H_e \le 1$. Let $G(\FF)$ and $G(\HH)$ be the induced sub-graphs of $G$ determined by $F_e$ and $H_e$ respectively. In particular, all edges $e$ with $F_e=0$ and $H_e$=0 are removed from the respective graphs.}

The main idea of our algorithm is as follows.

\begin{enumerate}

 \item  Solve the vertex-weighted benchmark LP in Section \ref{lp:stoch-match}. Let $\ff$ be an optimal solution vector.   
 
 \item Invoke $\DR{3}$ to obtain  an integral vector ${\FF}$ and a fractional vector $\HH$ with $\HH=\FF/3$. 
 
 \item Apply a series of modifications to $\HH$ and transform it to another solution $\HH'$ (See Section \ref{sec:vw-modification}).
 
 \item Run the Randomized List Algorithm~\cite{bib:Jaillet} with parameter $\HH'$, denoted by $\mathsf{RLA}[\HH']$, on $G(\HH)$. \\
\end{enumerate} 

\myEdit{
We first myefly describe how we overcome the vertex-weighted and unweighted \emph{bottleneck} cases for the algorithm in \cite{bib:Jaillet} and then explain the algorithm in full detail. Recall that~\cite{bib:Jaillet} analyze their algorithm by considering cases for various neighborhood structures at a given offline vertex.

The $\WS$ for the vertex-weighted case in~\cite{bib:Jaillet} is shown in Figure~\ref{fig:JailletWS} (left), which happens at a node $u$ with a competitive ratio of $0.725$. Jaillet and Lu described and analyzed this case in Claim 5 within the proof of Lemma 7 from~\cite{bib:Jaillet}. However, also from their analysis, we have that the node $u_1$ in Figure \ref{fig:JailletWS} (left) has a competitive ratio of at least 0.736. Hence, we can \emph{boost} the performance of $u$ at the cost of $u_1$ by making $u$ more likely to match and $u_1$ less likely. Specifically, we increase the value of $H_{(u,v_1)}$ and decrease the value $H_{(u_1,v_1)}$. Cases (10) and (11) in Figure \ref{fig:secondModification} illustrate this.

After this modification, we will later show that the new $\WS$ for vertex-weighted is now the $C_1$ cycle shown in both Figures~\ref{fig:JailletC1} and~\ref{fig:JailletWS} (right) and defined formally in Section~\ref{sec:cyclebreak}. In fact, this is also the $\WS$ for the unweighted problem in~\cite{bib:Jaillet} as well. Jaiillet and Lu give the following explaining in their ``Tight example'' section~\cite{bib:Jaillet}:
\begin{quote}
It is worth mentioning that the ratio of $1-2\euler^{-2}$ is tight for this algorithm. The ratio can be achieved with the following example: Consider the case of the complete bipartite graph $K_{n,n}$, where $n$ is an even number. One optimal solution to [the LP from~\cite{bib:Jaillet}] consists of a disjoint union of $n/2$ cycles of length $4$; within each cycle, two edges carry $1/3$ flow, and two carry $2/3$ flow. Since the underlying graph is $K_{n,n}$, the optimal offline solution is $n$. On the other hand, for any cycle in the offline optimal solution, the expected number of matches is $2(1 - \euler^{-2}$). Therefore, the competitive ratio in this instance is $1 - 2\euler^{-2} \approx 0.729$.
\end{quote}
 However, Lemma \ref{lem:fu=1} implies that $C_1$ cycles can be avoided with probability at least $\frac{3}{\euler}-1$ due to our LP and rounding procedure. This helps us improve the ratio even for the unweighted case in \cite{bib:Jaillet}. Lemma~\ref{lem:fu=1} describes this formally.
}

\begin{lemma}
\label{lem:fu=1}
For any given $u \in U$, $u$ appears in a $C_1$ cycle after $\DR{3}$ with probability at most $2 - \frac{3}{\euler}$.
\end{lemma}
\begin{proof}
Consider the graph $G(\FF)$ with $\FF$ obtained from $\DR{3}$. Notice that for some vertex $u$ to appear in a $C_1$ cycle, it must have a neighboring edge with $H_e = 2/3$. Now we try to bound the probability of this event. It is easy to see that for some $e \in \partial(u)$ with $f_{e} \le 1/3$, $F_e \le 1$ after $\DR{3}$, and hence $H_e=F_e/3 \le 1/3$. Thus only those edges $e \in \partial(u)$ with $f_{e}>1/3$ will possibly be rounded to $H_{e}=2/3$. Note that, there can be at most two such edges in $\partial(u)$, since $\sum_{e \in \partial(u)} f_{e} \leq 1$. Hence, we have the following two cases. 

\begin{itemize}
	\item 
		\textbf{Case 1: } $\partial(u)$ contains only one edge $e$ with $f_{e} > 1/3$. 
		Let $q_{1}=\Pr[H_e=1/3]$ and $q_{2}=\Pr[H_{e}=2/3]$ after $\DR3$. By $\DR{3}$, we know that $\mathbb{E}[H_{e}] =\E[F_e]/3 = q_{2}(2/3) + q_{1}(1/3) = f_{e}
		$.

		Notice that $q_{1}+q_{2}=1$ and hence
		$q_{2}= 3f_{e} - 1$.
		Since this is an increasing function of $f_{e}$ and $f_{e} \leq 1 - \frac{1}{\euler}$ from LP constraint \eqref{cons:single}, we have $q_{2} \leq 3(1 -\frac{1}{\euler}) - 1 = 2 - \frac{3}{\euler}$.
	\item
		\textbf{Case 2: } $\partial(u)$ contains two edges $e_{1}$ and $e_{2}$ with $f_{e_{1}} > 1/3$ and $f_{e_{2}} > 1/3$.
		Let $q_{2}$ be the probability that after $\DR{3}$, either $H_{e_{1}}= 2/3$ or $H_{e_{2}} = 2/3$. Note that, these two events are mutually exclusive since $H_{u} \leq 1$. Using the analysis from case 1, it follows that $q_{2}= (3f_{e_{1}} - 1) + (3f_{e_{2}} - 1) = 3(f_{e_{1}}+f_{e_{2}})-2  
		$.

		From LP constraint \eqref{cons:pair}, we know that $ f_{e_{1}}+f_{e_{2}} \leq 1 - \frac{1}{\euler^2}$, and hence $q_{2} \le 3(1 - \frac{1}{\euler^2})-2<2-\frac{3}{\euler}$. \\
\end{itemize}
\end{proof}

Now we present the details of  $\mathsf{RLA}$ based on a given $\HH'$ in Section \ref{apx:RLA-alg} and then discuss the two modifications transforming $\HH$ to $\HH'$ in Section \ref{sec:vw-modification}. We give a formal statement of our algorithm in Section \ref{sec:vw_last} and the related analysis. 


\subsection{$\mathsf{RLA}$ algorithm.}
\label{apx:RLA-alg}

\myEdit{We describe how to apply the $\RLA$ algorithm with parameter $\HH'$. WLOG assume that $H'_v \doteq \sum_{e \in \pa(v)} H'_e=1$.\footnote{We can add a dummy node $u'$ to $v$ if $H'_v<1$ and assignment $v$ to $u'$ simply means rejection of $v$.} Let $\delta_{\HH'}(v)$ be the set of neighbors of $v$ in $G(\HH')$ with $H'_{u,v}>0$. Thus, $|\del_\HH'(v)| \ge 2$ since each non-zero $H'_e$ satisfies $H'_e \in \{1/3,2/3\}$.}

Each time when a vertex $v$ comes, $\RLA$ first generates a random list $\cR_v$, which is a permutation over $\delta_{\HH'}(v)$, as follows.
\begin{itemize}
	\item If $|\delta_{\HH'}(v)|=2$, say $\delta_{\HH'}(v)=(u_1,u_2)$, then sample a random list $\cR_v$ such that
	\begin{equation}\label{eqn:rl_1}
	\Pr[\mathcal{R}_{v}=(u_{1},u_{2})]=H'_{(u_{1},v)},~~\Pr[\mathcal{R}_{v}=(u_{2},u_{1})]=H'_{(u_{2},v)}
	\end{equation}

	\item  If $|\delta_{\HH'}(v)|=3$, say $\delta_{\HH'}(v)=(u_1,u_2, u_3)$. Then we sample a permutation of $(i,j,k)$ over $\{1,2,3\}$ such that 
	\begin{equation}\label{eqn:rl_2}
	\Pr[\mathcal{R}_{v}=(u_{i},u_{j}, u_{k})]=H'_{(u_{i},v)}\frac{H'_{(u_{j},v)}}{H'_{(u_{j},v)}+H'_{(u_{k},v)}}
	\end{equation}
\end{itemize}

We can verify that the sampling distributions described in Equations \eqref{eqn:rl_1} and \eqref{eqn:rl_2} are valid since $H'_v=\sum_{e \in \pa(v)} H'_e=1$ and no $H'_e=1$. \myEdit{(Both properties are guaranteed in the two modifications shown in Section \ref{sec:vw-modification}.)} The full details of the Random List Algorithm, $\mathsf{RLA}[\HH']$, are shown in Algorithm \ref{alg:RLA}.
\begin{algorithm}[h!]	
	\caption{ (\cite{bib:Jaillet}) \protect{$\mathsf{RLA}[{\HH'}] $}  }\label{alg:RLA}
	\DontPrintSemicolon
	When a vertex $v$ comes, generate a random list $\mathcal{R}_{v}$ satisfying Equation \eqref{eqn:rl_1} or \eqref{eqn:rl_2}. \;
	If all $u$ in the list are matched, then drop the vertex $v$; otherwise, assign $v$ to the first unmatched $u$ in the list.\;
\end{algorithm}

 
\subsection{Two kinds of modifications to $\HH$.} 
\label{sec:vw-modification}

As stated earlier, we first modify $\vec{H}$ before running the $\mathsf{RLA}$ algorithm. In this section, we describe the modifications.

\subsubsection{The first modification to $\HH$: cycle breaking.}
\label{sec:cyclebreak}

\begin{figure}[h]
	\centering
	\begin{tikzpicture}
[
	xscale=0.7,yscale=1,auto,thick,
  	gray node/.style={circle,
  		inner sep=0pt,minimum size=14pt, 
  		fill=black!20,draw,font=\small},  
  	white node/.style={circle,
  		inner sep=0pt,minimum size=14pt, 
  		fill=white,draw,font=\small},
  node distance=24pt 
]


	\node [gray node] (u) at (-5,-1) {$u$};
	\node [white node] (v1) [right of=u, xshift=30, yshift=15] {$v_1$};
	\node [white node] (v2) [right of=u, xshift=30, yshift=-15] {$v_2$};
	\draw[ultra thick] (u)  -- (v1);
	\draw[thin] (u)  -- (v2);
	\node [white node] (u1) [above of=u, yshift=-5] {$u_1$};
		\draw[thin] (v1) -- (u1);
	\node [white node] (u2) [below of=u, yshift=5] {$u_2$};
		\draw[thin] (v2) -- (u2);
	\node [white node] (u3) [below of=u2, yshift=5] {$u_3$};
		\draw[thin] (v2) -- (u3);
	\node (v3) [right of=u1, xshift=5, yshift=10] {};
		\draw[ultra thick] (u1) -- (v3) node [yshift=10] {($\WS$)};
	\node (fu) [left of=u, xshift=10, yshift=0] {$1$};
	\node (fu1) [left of=u1, xshift=10, yshift=0] {$1$};
	\node (fu2) [left of=u2, xshift=10, yshift=0] {$1$};
	\node (fu3) [left of=u3, xshift=10, yshift=0] {$1$};


	\node [white node] (c1u1) at (0,0) {$u_1$};
	\node [white node] (c1u2) [below of=c1u1, xshift=0, yshift=0] {$u_2$};
	\node [white node] (c1v1) [right of=c1u1, xshift=25, yshift=0] {$v_1$};
	\node [white node] (c1v2) [right of=c1u2, xshift=25, yshift=0] {$v_2$};
	\draw[ultra thick] (c1u1)  -- (c1v1) node [midway, yshift=10] {($C_1$)};
	\draw[thin] (c1u1)  -- (c1v2);
	\draw[thin] (c1u2)  -- (c1v1);
	\draw[ultra thick] (c1u2)  -- (c1v2);


	\node [white node] (c2u1) at (3.5,0) {$u_1$};
	\node [white node] (c2u2) [below of=c2u1, xshift=0, yshift=0] {$u_2$};
	\node [white node] (c2v1) [right of=c2u1, xshift=25, yshift=0] {$v_1$};
	\node [white node] (c2v2) [right of=c2u2, xshift=25, yshift=0] {$v_2$};
	\draw[ultra thick] (c2u1)  -- (c2v1) node [midway, yshift=10] {($C_2$)};
	\draw[thin] (c2u1)  -- (c2v2);
	\draw[thin] (c2u2)  -- (c2v1);
	\draw[thin] (c2u2)  -- (c2v2) node (c2start) [midway, below, yshift=-5] {};


	\node [white node] (bc2u1) at (3.5,-2.5) {$u_1$};
	\node [white node] (bc2u2) [below of=bc2u1, xshift=0, yshift=0] {$u_2$};
	\node [white node] (bc2v1) [right of=bc2u1, xshift=25, yshift=0] {$v_1$};
	\node [white node] (bc2v2) [right of=bc2u2, xshift=25, yshift=0] {$v_2$};
	\draw[thin] (bc2u1)  -- (bc2v1) node (c2end) [midway, yshift=5] {};
	\draw[ultra thick] (bc2u1)  -- (bc2v2);
	\draw[ultra thick] (bc2u2)  -- (bc2v1);

\draw[->] (c2start) -- (c2end);


	\node [white node] (c3u1) at (7,0) {$u_1$};
	\node [white node] (c3u2) [below of=c3u1, xshift=0, yshift=0] {$u_2$};
	\node [white node] (c3v1) [right of=c3u1, xshift=25, yshift=0] {$v_1$};
	\node [white node] (c3v2) [right of=c3u2, xshift=25, yshift=0] {$v_2$};
	\draw[thin] (c3u1)  -- (c3v1) node [midway, yshift=10] {($C_3$)};
	\draw[thin] (c3u1)  -- (c3v2);
	\draw[thin] (c3u2)  -- (c3v1);
	\draw[thin] (c3u2)  -- (c3v2) node (c3start) [midway, below, yshift=-5] {};


	\node [white node] (bc3u1) at (7,-2.5) {$u_1$};
	\node [white node] (bc3u2) [below of=bc3u1, xshift=0, yshift=0] {$u_2$};
	\node [white node] (bc3v1) [right of=bc3u1, xshift=25, yshift=0] {$v_1$};
	\node [white node] (bc3v2) [right of=bc3u2, xshift=25, yshift=0] {$v_2$};
	\draw[ultra thick] (bc3u1)  -- (bc3v1) node (c3end) [midway, yshift=5] {};
	\draw[ultra thick] (bc3u2)  -- (bc3v2);

\draw[->] (c3start) -- (c3end);




  



\end{tikzpicture}
	\caption{Left: The $\WS$ for Jaillet and Lu~\cite{bib:Jaillet} for their vertex-weighted case. Right: The three possible types of cycles of length $4$ after applying $\DR{3}$. Thin edges have $H_e = 1/3$ and thick edges have $H_e = 2/3$. \myEdit{The arrows show how cycles $C_2$ and $C_3$ are broken.}}
	\label{fig:JailletCycles}
	\label{fig:JailletWS}
\end{figure}
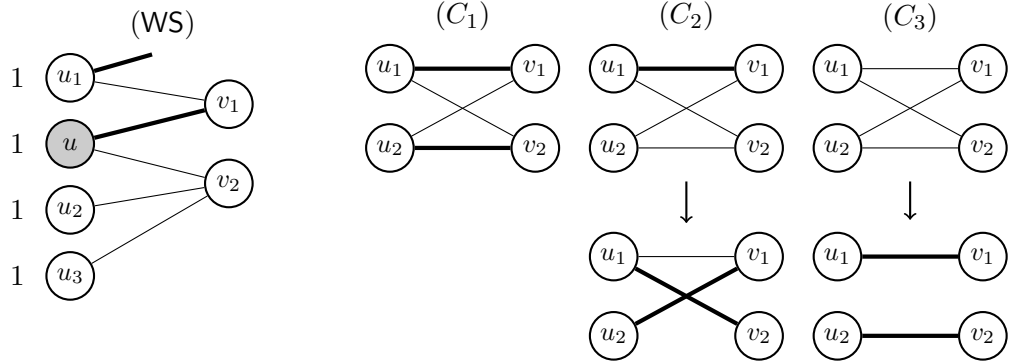

\myEdit{
The first modification is to break the cycles of length $4$ deterministically. There are three possible cycles of length $4$ in the graph $G_{\HH}$, denoted $C_1$, $C_2$, and $C_3$ in the righthand side of Figure~\ref{fig:JailletCycles} and defined as follows. 

\begin{definition}[Cycle type $C_1$]
This length $4$ cycle is a complete bipartite graph on two offline vertices and two online vertices. It has two vertex-disjoint edges with $H_e = 2/3$ and the remaining edges have $H_e = 1/3$.
\end{definition}

\begin{definition}[Cycle type $C_2$]
This length $4$ cycle is a complete bipartite graph on two offline vertices and two online vertices. It has one edge with $H_e = 2/3$ and the remaining edges have $H_e = 1/3$.
\end{definition}
 
\begin{definition}[Cycle type $C_3$]
This length $4$ cycle is a complete bipartite graph on two offline vertices and two online vertices. All edges have $H_e = 1/3$.
\end{definition}
 
}

In \cite{bib:Jaillet}, they give an efficient way to break $C_{2}$ and $C_{3}$, as shown in Figure~\ref{fig:JailletCycles}. Cycle $C_{1}$ cannot be modified further and hence, is the bottleneck for their unweighted case. Notice that, while breaking the cycles of type $C_{2}$ and $C_{3}$, new cycles of $C_1$ can be created in the graph. Since our randomized construction of solution $\HH$ gives us control on the probability of cycles $C_1$ occurring, we would like to break $C_{2}$ and $C_{3}$ in a controlled way, so as not to create any new $C_1$ cycles. This procedure is summarized in Algorithm~\ref{alg:CycleBreak} and its correctness is proved in Lemma~\ref{lem:alg-cycle}.

\begin{algorithm}[!t]
	
	\caption{ \protect{[Cycle breaking algorithm] Offline Phase} }\label{alg:CycleBreak}
	\DontPrintSemicolon
	While there is some cycle of type $C_2$ or $C_3$, Do:\;
	Break all cycles of type $C_2$. \;
	Break one cycle of type $C_3$ and return to the first step. \;
\end{algorithm}

\myEdit{
\label{apex:proof-de-cycle}
The proof of Lemma~\ref{lem:alg-cycle} will follow from three claims which we state and prove below.

\begin{claim}
	\label{cl:cycleFw}
	Breaking cycles will not change the value $H_w$ for any $w \in U \cup V$.
\end{claim}
\begin{proof}
	As shown in Figure~\ref{fig:JailletCycles}, we increase and decrease edge values $f_e$ in such a way that their sums $H_w$ at any vertex $w$ will be preserved.
\end{proof}

\begin{claim}
	\label{cl:c2}
	After breaking a cycle of type $C_2$, the vertices $u_1$, $u_2$, $v_1$, and $v_2$ can never be part of any length four cycle.
\end{claim}
\begin{proof}
	Consider the structure after breaking a cycle of type $C_2$. Note that the edge $(u_2, v_2)$ has been permanently removed and hence, these four vertices together can never be part of a cycle of length four. The vertices $u_1$ and $v_1$ have $H_{u_1} = 1$ and $H_{v_1} = 1$ respectively. So they cannot have any other edges and therefore cannot appear in any length four cycle. The vertices $u_2$ and $v_2$ can each have one additional edge, but since the edge $(u_2, v_2)$ has been removed, they can never be part of any cycle with length less than six. 
\end{proof}

\begin{claim}
	\label{cl:c3}
	When all length four cycles are of type $C_1$ or $C_3$, breaking exactly one cycle of type $C_3$ cannot create a new cycle of type $C_1$. 
\end{claim}	
\begin{proof}
	First, we note that since no edges will be added during this process, we cannot create a new cycle of length four or join with a cycle of type $C_1$. Therefore, the only cycles which could be affected are of type $C_3$. However, every cycle $c$ of type $C_3$ falls into one of two cases:
	\begin{description}
	\item[Case 1:]~~$c$ is the cycle we are breaking. In this case, $c$ cannot become a cycle of type $C_1$ since we remove two of its edges and break the cycle.
	
	\item[Case 2:]~~$c$ is not the cycle we are breaking. In this case, $c$ can have at most one of its edges converted to a $2/3$ edge. Let $c'$ be the length four cycle we are breaking. Note that $c$ and $c'$ will differ by at least one vertex. When we break $c'$, the two edges which are converted to $2/3$ will cover all four vertices of $c'$. Therefore, at most one of these edges can be in $c$.

Note that breaking one cycle of type $C_3$ could create cycles of type $C_2$, but these cycles are always broken in the next iteration, before breaking another cycle of type $C_3$. 
	\end{description}
\end{proof}

\begin{lemma}\label{lem:alg-cycle}
After applying Algorithm~\ref{alg:CycleBreak} to $G({\HH})$, we have (1) the value $H_{w}$ is preserved for each $w \in U \cup V$; (2) no cycle of type $C_{2}$ or $C_{3}$ exists; (3) no new cycle of type $C_{1}$ is added.
\end{lemma}
\begin{proof}
The proof follows from Claims~\ref{cl:cycleFw},~\ref{cl:c2}, and~\ref{cl:c3}. 
Notice that $C_2$ cycles can be freely broken without creating new $C_1$ cycles. After removing all cycles of type $C_2$, removing a single cycle of type $C_3$ cannot create any cycles of type $C_1$. Hence, Algorithm~\ref{alg:CycleBreak} removes all $C_2$ and $C_3$ cycles without creating any new $C_1$ cycles.
\end{proof}
}

\subsubsection{The Second modification to $\HH$: balancing the worst case.}
  
\myEdit{Informally, this second modification decreases $H_e$ values on $u$ with $H_u=1/3$ or $H_u=2/3$ and increases  $H_e$ values on $u$ with $H_u=1$. We will illustrate this intuition on the following example.}

\begin{figure}[h]
	\centering
	\begin{tikzpicture}
[
	xscale=1,yscale=1,auto,thick,
  	gray node/.style={circle,
  		inner sep=0pt,minimum size=18pt, 
  		fill=black!20,draw,font=\small},  
  	white node/.style={circle,
  		inner sep=0pt,minimum size=18pt, 
  		fill=white,draw,font=\small},
  node distance=30pt 
]


	\node [gray node] (u) at (0,0) {$u$};
	\node [white node] (v1) [right of=u, xshift=30, yshift=15] {$v_1$};
	\node [white node] (v2) [right of=u, xshift=30, yshift=-15] {$v_2$};
	\draw[ultra thick] (u)  -- (v1);
	\draw[thin] (u)  -- (v2);
	\node [white node] (u1) [above of=u, yshift=0] {$u_1$};
		\draw[thin] (v1) -- (u1);
	\node [white node] (u2) [below of=u, yshift=0] {$u_2$};
		\draw[ultra thick] (v2) -- (u2);
	\node (fu) [left of=u, xshift=10, yshift=0] {$1$};
	\node (fu1) [left of=u1, xshift=10, yshift=0] {$1/3$};
	\node (fu2) [left of=u2, xshift=10, yshift=0] {$2/3$};



  



\end{tikzpicture}
	\hspace{8em}\begin{tikzpicture}
[
	xscale=1,yscale=1,auto,thick,
  	gray node/.style={circle,
  		inner sep=0pt,minimum size=14pt, 
  		fill=black!20,draw,font=\small},  
  	white node/.style={circle,
  		inner sep=0pt,minimum size=14pt, 
  		fill=white,draw,font=\small},
  node distance=24pt 
]


	\node [white node] (u1) at (0,0) {$u_1$};
	\node [white node] (v) [right of=u, xshift=30, yshift=-15] {$v_1$};
	\draw[thin] (u1)  -- (v) node[pos=0.3]{$0.1$};
	\node [gray node] (u) [below of=u, yshift=0] {$u$};
		\draw[ultra thick] (u) -- (v) node[pos=0, xshift=35, yshift=-13]{$0.9$};
	\node (fu1) [left of=u1, xshift=6, yshift=0] {$1/3$};
	\node (fu) [left of=u, xshift=10, yshift=0] {$1$};



  



\end{tikzpicture}
	\caption{An example of the need for the second modification.  For the left: competitive analysis shows that in this case, $u_1$ and $u_2$ can achieve a high competitive ratio at the expense of $u$. For the right: an example of balancing strategy by making $v_{1}$ slightly more likely to pick $u$ when it comes.} 
	\label{fig:AttenuationExample}
\end{figure}
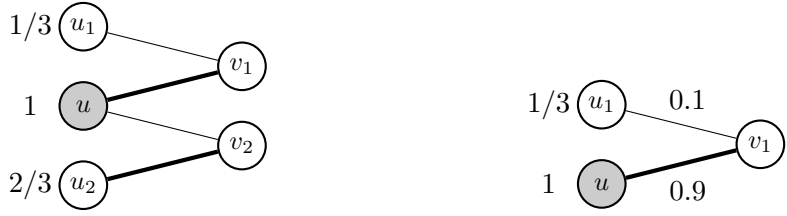

\myEdit{Consider the two graphs, denoted by $G_L$ and $G_R$ in Figure \ref{fig:AttenuationExample}, where thin and thick edges represent $H_{e}=1/3$ and $H_{e}=2/3$ respectively. We now compute the competitive ratio after applying $\mathsf{RLA}$ on $G_L$.  For each node $w$, let $\del(w)$ be the set of neighbors of $w$ in $G_L$. Let $A_u$ be the event that $u$ is matched in $\mathsf{RLA}$. Let $A_{u,1}$ denote the event that among the $n$ random arrival lists, there exists a list starting with $u$. For each $v \in \del(u)=\{v_1,v_2\}$, let $A_{u,2}^{v}$ denote the event that among the $n$ online arrival lists, there exists successive lists such that (I) Each of those lists starts with a $u' \neq u$ and $u' \in \delta(v)$ and (II) The lists arrive in an order which ensures $u$ will be matched by the algorithm. From lemma 4 and Corollary 1 in \cite{bib:Jaillet}, we have the following.

\begin{lemma}[\cite{bib:Jaillet}]
	\label{lem:cycle}
	Suppose $u$ is not a part of any cycle of length $4$. We have 
	$$\Pr[A_u]=1-(1-\Pr[A_{u,1}])\prod_{v \in \del(u)} (1-\Pr[A_{u,2}^{v}]) +o(1/n)$$
\end{lemma}

The validity of the above lemma can be seen as follows: the probability that $u$ is not matched ($\neg A_u$) can be approximated up to $o(1/n)$ by the probability that none of lists arrives staring with $u$ ($\neg A_{u,1}$) and  none of events described in (II) occurs $(\wedge_{v \in \del(u)} \neg A_{u,2}^v)$. 

For the node $u$ in $G_L$, we have $\Pr[A_{u,1}]=1-\euler^{-1}$. 
From the definition, $A_{u,2}^{v_{1}}$ is the event that among the $n$ online lists, the random list $\mathcal{R}_{v_{1}}=(u_{1},u)$ comes at least twice. Notice that the list $\mathcal{R}_{v_{1}}=(u_{1},u)$ arrives with probability $\frac{1}{3n}$ each round. Thus we have 
$\Pr[A_{u,2}^{v_{1}}]=\Pr[X \ge 2]=1-\euler^{-1/3}(1+1/3)$, where $X \sim \operatorname{Pois} (1/3)$. Similarly, we can get $\Pr[A_{u,2}^{v_{2}}]=1-\euler^{-2/3}(1+2/3)$ and the resultant 
$\Pr[A_u]=1-\frac{20}{9e^{2}} \sim 0.699$. Observe that for $u_1$ and $u_2$, 
$\Pr[A_{u_1} \ge \Pr[A_{u_1,1}]=1-\euler^{-1/3}$ and $\Pr[A_{u_2}] \ge \Pr[A_{u_2,1}]=1-\euler^{-2/3}$.

Let $\RR{\RLA}{1}$, $\RR{\RLA}{1/3}$ and $\RR{\RLA}{2/3}$ be the competitive ratio
achieved by $\mathsf{RLA}$ for $u$, $u_{1}$ and $u_{2}$ respectively. Hence, we have $\RR{\RLA}{1}\sim 0.699$ while $\RR{\RLA}{1/3} \ge 3(1-\euler^{-1/3})\sim 0.8504$ and  $\RR{\RLA}{2/3} \ge 0.729$. 

Intuitively, one can improve the worst case ratio by increasing the arrival rate for $\mathcal{R}_{v_{1}}=(u,u_{1})$ while reducing that for $\mathcal{R}_{v_{1}}=(u_{1},u)$. Suppose we modify $H_{(u_{1},v_{1})}$ and $H_{(u,v_{1})}$ to $H'_{(u_{1},v_{1})}=0.1$ and $H'_{(u,v_{1})}=0.9$ as shown in $G_R$, the arrival rate for $\mathcal{R}_{v_{1}}=(u,u_{1})$ and $\mathcal{R}_{v_{1}}=(u_{1},u)$ gets modified to $0.1/n$ and $0.9/n$ respectively.  The updated values are $\Pr[A_{u,1}] = 1-\euler^{-0.9-1/3}$, $\Pr[A_{u,2}^{v_1}]=1-\euler^{-0.1}(1+0.1)$, $\RR{\RLA}{1}=0.751$, $\Pr[A_{u_1,1}] = 1-\euler^{-1/3}$, $\Pr[A_{u_{1},2}^{v_{1}}] \sim 0.227$ and $\RR{\RLA}{1/3} \ge 0.8$. Hence, the performance on $\WS$ instance improves. Notice that after the modifications, 
$H'_{u}=H'_{(u,v_{1})}+H_{(u,v_{2})}=0.9+1/3$.
} 

Figure \ref{fig:secondModification} describes the various modifications applied to $\HH$ vector. The values on top of the edge, denote the new values. Cases (11) and (12) help improve upon the $\WS$ described in Figure \ref{fig:JailletWS}.
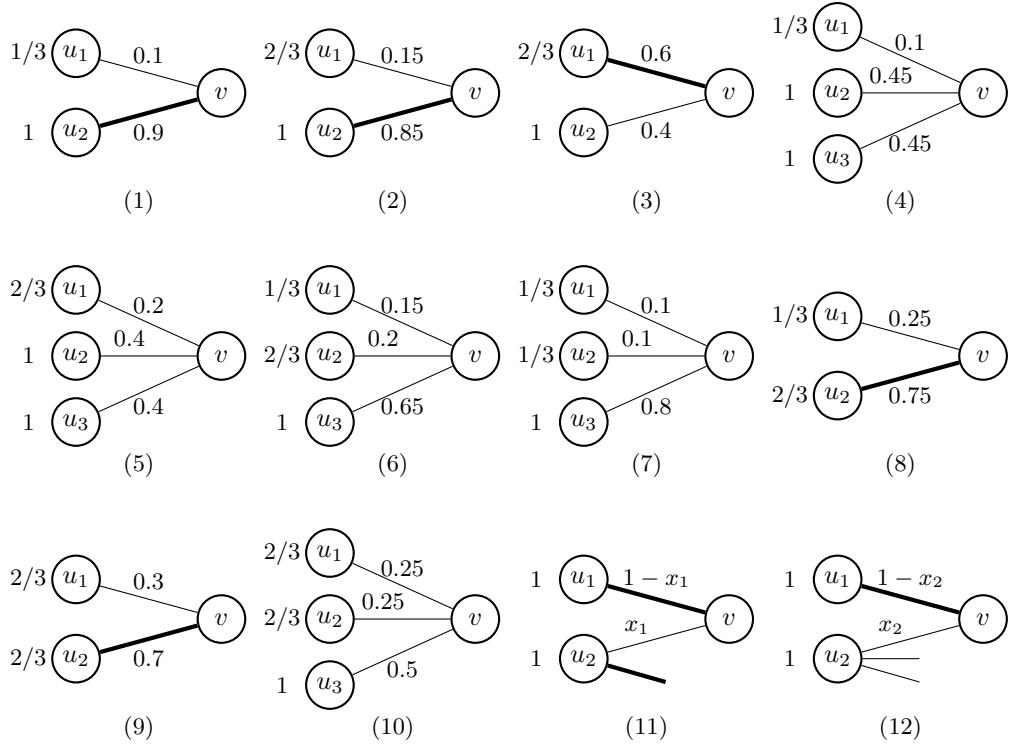
\begin{figure}[!h]
	\centering
	\begin{tikzpicture}
[
	xscale=0.8,yscale=1,auto,thick,font=\footnotesize,
  	gray node/.style={circle,
  		inner sep=0pt,minimum size=14pt, 
  		fill=black!20,draw,font=\small},  
  	white node/.style={circle,
  		inner sep=0pt,minimum size=14pt, 
  		fill=white,draw,font=\small},
  node distance=24pt 
]


	\node [white node] (1v) at (0,0) {$v$};
	\node [white node] (1u1) [left of=1v, xshift=-25, yshift=15] {$u_1$};
	\node [white node] (1u2) [left of=1v, xshift=-25, yshift=-15] {$u_2$};
	\draw[thin] (1u1)  -- (1v) node [midway, above] {$0.1$};
	\draw[ultra thick] (1u2)  -- (1v) node [midway, below] {$0.9$};
	\node (1fu1) [left of=1u1, xshift=8, yshift=0] {$1/3$};
	\node (1fu2) [left of=1u2, xshift=12, yshift=0] {$1$};
	\node [below left of=1v, xshift=-10, yshift=-20] {(1)};


	\node [white node] (2v) at (3.75,0) {$v$};
	\node [white node] (2u1) [left of=2v, xshift=-25, yshift=15] {$u_1$};
	\node [white node] (2u2) [left of=2v, xshift=-25, yshift=-15] {$u_2$};
	\draw[thin] (2u1)  -- (2v) node [midway, above] {$0.15$};
	\draw[ultra thick] (2u2)  -- (2v) node [midway, below] {$0.85$};
	\node (2fu1) [left of=2u1, xshift=8, yshift=0] {$2/3$};
	\node (2fu2) [left of=2u2, xshift=12, yshift=0] {$1$};
	\node [below left of=2v, xshift=-10, yshift=-20] {(2)};


	\node [white node] (3v) at (7.5,0) {$v$};
	\node [white node] (3u1) [left of=3v, xshift=-25, yshift=15] {$u_1$};
	\node [white node] (3u2) [left of=3v, xshift=-25, yshift=-15] {$u_2$};
	\draw[ultra thick] (3u1)  -- (3v) node [midway, above] {$0.6$};
	\draw[thin] (3u2)  -- (3v) node [midway, below] {$0.4$};
	\node (3fu1) [left of=3u1, xshift=8, yshift=0] {$2/3$};
	\node (3fu2) [left of=3u2, xshift=12, yshift=0] {$1$};
	\node [below left of=3v, xshift=-10, yshift=-20] {(3)};


	\node [white node] (4v) at (11.25,0) {$v$};
	\node [white node] (4u1) [left of=4v, xshift=-25, yshift=25] {$u_1$};
	\node [white node] (4u2) [left of=4v, xshift=-25, yshift=0] {$u_2$};
	\node [white node] (4u3) [left of=4v, xshift=-25, yshift=-25] {$u_3$};
	\draw[thin] (4u1)  -- (4v) node [midway, above] {$0.1$};
	\draw[thin] (4u2)  -- (4v) node [above, pos=0.3] {$0.45$};
	\draw[thin] (4u3)  -- (4v) node [midway, below] {$0.45$};
	\node (4fu1) [left of=4u1, xshift=8, yshift=0] {$1/3$};
	\node (4fu2) [left of=4u2, xshift=12, yshift=0] {$1$};
	\node (4fu3) [left of=4u3, xshift=12, yshift=0] {$1$};
	\node [below left of=4v, xshift=-10, yshift=-20] {(4)};


	\node [white node] (5v) at (0,-3.5) {$v$};
	\node [white node] (5u1) [left of=5v, xshift=-25, yshift=25] {$u_1$};
	\node [white node] (5u2) [left of=5v, xshift=-25, yshift=0] {$u_2$};
	\node [white node] (5u3) [left of=5v, xshift=-25, yshift=-25] {$u_3$};
	\draw[thin] (5u1)  -- (5v) node [midway, above] {$0.2$};
	\draw[thin] (5u2)  -- (5v) node [above, pos=0.3] {$0.4$};
	\draw[thin] (5u3)  -- (5v) node [midway, below] {$0.4$};
	\node (5fu1) [left of=5u1, xshift=8, yshift=0] {$2/3$};
	\node (5fu2) [left of=5u2, xshift=12, yshift=0] {$1$};
	\node (5fu3) [left of=5u3, xshift=12, yshift=0] {$1$};
	\node [below left of=5v, xshift=-10, yshift=-20] {(5)};


	\node [white node] (6v) at (3.75,-3.5) {$v$};
	\node [white node] (6u1) [left of=6v, xshift=-25, yshift=25] {$u_1$};
	\node [white node] (6u2) [left of=6v, xshift=-25, yshift=0] {$u_2$};
	\node [white node] (6u3) [left of=6v, xshift=-25, yshift=-25] {$u_3$};
	\draw[thin] (6u1)  -- (6v) node [midway, above] {$0.15$};
	\draw[thin] (6u2)  -- (6v) node [above, pos=0.3] {$0.2$};
	\draw[thin] (6u3)  -- (6v) node [midway, below] {$0.65$};
	\node (6fu1) [left of=6u1, xshift=8, yshift=0] {$1/3$};
	\node (6fu2) [left of=6u2, xshift=8, yshift=0] {$2/3$};
	\node (6fu3) [left of=6u3, xshift=12, yshift=0] {$1$};
	\node [below left of=6v, xshift=-10, yshift=-20] {(6)};


	\node [white node] (7v) at (7.5,-3.5) {$v$};
	\node [white node] (7u1) [left of=7v, xshift=-25, yshift=25] {$u_1$};
	\node [white node] (7u2) [left of=7v, xshift=-25, yshift=0] {$u_2$};
	\node [white node] (7u3) [left of=7v, xshift=-25, yshift=-25] {$u_3$};
	\draw[thin] (7u1)  -- (7v) node [midway, above] {$0.1$};
	\draw[thin] (7u2)  -- (7v) node [above, pos=0.3] {$0.1$};
	\draw[thin] (7u3)  -- (7v) node [midway, below] {$0.8$};
	\node (7fu1) [left of=7u1, xshift=8, yshift=0] {$1/3$};
	\node (7fu2) [left of=7u2, xshift=8, yshift=0] {$1/3$};
	\node (7fu3) [left of=7u3, xshift=12, yshift=0] {$1$};
	\node [below left of=7v, xshift=-10, yshift=-20] {(7)};


	\node [white node] (8v) at (11.25,-3.5) {$v$};
	\node [white node] (8u1) [left of=8v, xshift=-25, yshift=15] {$u_1$};
	\node [white node] (8u2) [left of=8v, xshift=-25, yshift=-15] {$u_2$};
	\draw[thin] (8u1)  -- (8v) node [midway, above] {$0.25$};
	\draw[ultra thick] (8u2)  -- (8v) node [midway, below] {$0.75$};
	\node (8fu1) [left of=8u1, xshift=8, yshift=0] {$1/3$};
	\node (8fu2) [left of=8u2, xshift=8, yshift=0] {$2/3$};
	\node [below left of=8v, xshift=-10, yshift=-20] {(8)};


	\node [white node] (9v) at (0,-7) {$v$};
	\node [white node] (9u1) [left of=9v, xshift=-25, yshift=15] {$u_1$};
	\node [white node] (9u2) [left of=9v, xshift=-25, yshift=-15] {$u_2$};
	\draw[thin] (9u1)  -- (9v) node [midway, above] {$0.3$};
	\draw[ultra thick] (9u2)  -- (9v) node [midway, below] {$0.7$};
	\node (9fu1) [left of=9u1, xshift=8, yshift=0] {$2/3$};
	\node (9fu2) [left of=9u2, xshift=8, yshift=0] {$2/3$};
	\node [below left of=9v, xshift=-10, yshift=-20] {(9)};


	\node [white node] (10v) at (3.75,-7) {$v$};
	\node [white node] (10u1) [left of=10v, xshift=-25, yshift=25] {$u_1$};
	\node [white node] (10u2) [left of=10v, xshift=-25, yshift=0] {$u_2$};
	\node [white node] (10u3) [left of=10v, xshift=-25, yshift=-25] {$u_3$};
	\draw[thin] (10u1)  -- (10v) node [midway, above] {$0.25$};
	\draw[thin] (10u2)  -- (10v) node [above, pos=0.3] {$0.25$};
	\draw[thin] (10u3)  -- (10v) node [midway, below] {$0.5$};
	\node (10fu1) [left of=10u1, xshift=8, yshift=0] {$2/3$};
	\node (10fu2) [left of=10u2, xshift=8, yshift=0] {$2/3$};
	\node (10fu3) [left of=10u3, xshift=12, yshift=0] {$1$};
	\node [below left of=10v, xshift=-10, yshift=-20] {(10)};


	\node [white node] (11v) at (7.5,-7) {$v$};
	\node [white node] (11u1) [left of=11v, xshift=-25, yshift=15] {$u_1$};
	\node [white node] (11u2) [left of=11v, xshift=-25, yshift=-15] {$u_2$};
	\node (11v1) [right of=11u2, xshift=5, yshift=-10] {};
	\draw[ultra thick] (11u1)  -- (11v) node [midway, above] {$1-x_1$};
	\draw[thin] (11u2)  -- (11v) node [above, pos=0.3] {$x_1$};
	\draw[ultra thick] (11u2)  -- (11v1);
	\node (11fu1) [left of=11u1, xshift=12, yshift=0] {$1$};
	\node (11fu2) [left of=11u2, xshift=12, yshift=0] {$1$};
	\node [below left of=11v, xshift=-10, yshift=-20] {(11)};


	\node [white node] (12v) at (11.25,-7) {$v$};
	\node [white node] (12u1) [left of=12v, xshift=-25, yshift=15] {$u_1$};
	\node [white node] (12u2) [left of=12v, xshift=-25, yshift=-15] {$u_2$};
	\node (12v1) [right of=12u2, xshift=5, yshift=0] {};
	\node (12v2) [right of=12u2, xshift=5, yshift=-10] {};
	\draw[ultra thick] (12u1)  -- (12v) node [midway, above] {$1-x_2$};
	\draw[thin] (12u2)  -- (12v) node [above, pos=0.3] {$x_2$};
	\draw[thin] (12u2)  -- (12v1);
	\draw[thin] (12u2)  -- (12v2);
	\node (12fu1) [left of=12u1, xshift=12, yshift=0] {$1$};
	\node (12fu2) [left of=12u2, xshift=12, yshift=0] {$1$};
	\node [below left of=12v, xshift=-10, yshift=-20] {(12)};

\end{tikzpicture}
	\caption{Illustration for second modification to $\HH$. The value assigned to each edge represents the value after the second modification. Here, $x_1 = 0.2744$ and $x_2 = 0.15877$.}\label{fig:secondModification}
\end{figure}

\subsection{Vertex-Weighted Algorithm $\VW$}\label{sec:vw_last}
\subsubsection{Analysis of algorithm $\VW$.}
The full details of our vertex-weighted algorithm are stated as follows. 
\begin{algorithm}[h!]
	\caption{ \protect{$\VW$} }\label{alg:vertexweighted}
	\DontPrintSemicolon
	Construct and solve the LP in sub-section \ref{lp:stoch-match} for the input instance. \;
	Invoke $\DR3$ to output $\FF$ and $\HH$. \;
	\myEdit{Apply the first modification shown in Algorithm~\ref{alg:CycleBreak} to transfer $\HH$ to $\widehat{\HH}$.\;
	Apply the second modification shown in Figure \ref{fig:secondModification} to morph $\widehat{\HH}$ to $\HH'$.\;}
	Run $\mathsf{RLA}[\HH']$.
\end{algorithm}

The algorithm $\VW$ consists of two different random processes: sub-routine $\DR3$ in the offline phase and  $\mathsf{RLA}$ in the online phase. Consequently, the analysis consists of two parts. First, for a given graph $G_{\HH}$, we analyze the ratio of $\mathsf{RLA}[\HH']$ for each node $u$ with $H_{u}=1/3, H_u=2/3$ and $H_{u}=1$. The analysis is similar to \cite{bib:Jaillet}. Second, we  
analyze the probability that $\DR3$ transforms each $u$, with fractional $f_{u}$ values, into the three discrete cases seen in the first part.  By combining the results from these two parts we get the final ratio.

Let us first analyze the competitive ratio for $\mathsf{RLA}[\HH']$.
For a given $\HH$ and $G({\HH})$, let $\mathsf{P}_{u}$ be the probability that $u$ gets matched in $\mathsf{RLA}[\HH']$. Notice that the value $\mathsf{P}_{u}$ is determined not just by the algorithm $\mathsf{RLA}$ itself, but also the modifications applied to $\HH$.  We define the competitive ratio of a vertex $u$ achieved by $\mathsf{RLA}$ as $\mathsf{P}_{u}/H_u$,  after modifications. Lemma \ref{lem:offline} gives the respective ratio values. The proof can be found in section \ref{apx:offline-vw} in the Appendix.

\begin{lemma}
	\label{lem:offline}
	Consider a given $\HH$ and a vertex $u$.
	The respective ratios achieved by $\mathsf{RLA}$ after the modifications are as described below. 
	
	\begin{itemize}
		
		\item  If $H_{u}=1$, then the competitive ratio $\RR{\RLA}{1} =  1-2\euler^{-2}\sim 0.72933$ if $u$ is in the first cycle $C_{1}$ and $\RR{\RLA}{1} \ge 0.735622$ otherwise.
		
		\item If $H_{u}=2/3$, then the competitive ratio $\RR{\RLA}{2/3} \ge 0.7847$.
		
		\item  If $H_{u}=1/3$, then  competitive ratio $\RR{\RLA}{1/3} \ge 0.7622$.

	\end{itemize}
\end{lemma}

Now we have all ingredients to state and prove Theorem \ref{thm:vertex}.
\begin{theorem}
\label{thm:vertex}
For vertex-weighted online stochastic matching with integral arrival rates,  online algorithm $\mathsf{VW}$ achieves a competitive ratio of at least $0.7299$.
\end{theorem}

\begin{proof}
	From Lemmas \ref{lem:fu=1} and \ref{lem:alg-cycle}, we know that any $u$ is present in cycle $C_{1}$ with probability at most $(2-\frac{3}{\euler})$. 
	
	Consider a node $u$ with $2/3 \le f_{u} \le 1$ and let $q_{1}, q_{2}, q_{3}$ be the probability that after $\DR{3}$ and the first modification, $H_{u}=1$ and $u$ is in the first cycle $C_{1}$, $H_{u}=1$ and $u$ is not in $C_{1}$, $H_{u}=2/3$ respectively. \myEdit{ From the marginal distribution of $\DR{3}$, we have that $q_1+q_2+q_3(2/3)=\E[\FF_u]/3=3f_u/3=f_u$. From Lemma \ref{lem:offline}, we get that the final ratio for $u$ is
	
	\begin{align*}
	\frac{1}{f_u}\Pr[\mbox{$u$ is matched}]&=\frac{1}{f_u} \Big(
	q_1\Pr[\mbox{$u$ is matched} | H_u=1, u \in C_1]  \\
	&+ q_2\Pr[\mbox{$u$ is matched} | H_u=1, u \notin C_1] \\
	&+q_3\Pr[\mbox{$u$ is matched} | H_u=2/3, u \in C_1 ]\Big) \\
	&\ge \frac{ 0.72933q_1+ 0.735622q_2 +  (2/3) *0.7847q_3 }{q_{1}+q_{2}+ (2/3)q_3}
	\end{align*}
	}
	
	Minimizing the above expression subject to (1) $q_{1}+q_{2}+q_{3}=1$; (2) $0 \le q_{i}, 1 \le i \le 3$; (3) $q_{1} \le 2-\frac{3}{\euler}$, we get a minimum value of $0.729982$ for $q_{1}=2-\frac{3}{\euler}$ and $q_{2}=\frac{3}{\euler}-1$.
	
	For any node $u$ with $0 \le u \le 2/3$, we know that the ratio is at least the min value of $\RR{\RLA}{2/3}$ and $\RR{\RLA}{1/3}$, which is $0.7622$.
	This completes the proof of Theorem \ref{thm:vertex}.
\end{proof}

\section{Non-integral arrival rates with stochastic rewards.}
\label{sec:nonint}

\newcommand{\y}{\vec{y}}

The setting here is strictly generalized over the previous sections in the following ways. Firstly, it allows an arbitrary arrival rate (say $r_v$) which can be fractional for each stochastic vertex $v$. Notice that, $\sum_{v} r_{v}=n$ where $n$ is the total number of rounds. 
Secondly, each $e=(v,u) \in E$ is associated with a value $p_{e}$, which captures the probability that the edge $e =(u, v)$ is present when we \emph{probe} it. We assume this process is independent of the stochastic arrival of each $v$. We will show that the simple non-adaptive algorithm introduced in~\cite{bib:Haeupler} can be extended to this general case. This achieves a competitive ratio of $(1- \frac{1}{\euler})$. Note that Manshadi \emph{et al.}~\cite{bib:Manshadi} show that no non-adaptive algorithm can possibly achieve a ratio better than $(1-1/\euler)$ for the non-integral arrival rates, even for the case of all $p_{e}=1$. Thus, our algorithm is an optimal non-adaptive algorithm for this model.  \\

We use an LP similar to \cite{bib:Jaillet} for the case of non-integral arrival rates. For each $e=(u,v) \in E$, let $f_{e}$ be the expected number of probes on edge $e$. When there are multiple copies of $v$, we count the sum of probes among all copies of $e$ in the offline optimal matching and thus some realizations of $f_e$ can be greater than $1$.  Consider the below LP:

\begin{eqnarray}
	\label{lp2:stoch-match}
	\max && \sum_{e \in E} w_{e} f_e p_e \\
	\mbox{s.t.} &&
	\sum_{e \in \partial(u)} f_e p_{e} \leq 1 \qquad \forall u \in U \label{con:lp2-1} \\
	&& \sum_{e \in \partial(v)} f_e \leq r_{v} \qquad \forall v \in V  \label{con:lp2-2} 
	\\
	&& 0 \le f_e  \qquad \forall e \in E \label{con:lp2-3}
\end{eqnarray}

Similar to Lemma~\ref{label:LPUpperbound}, we have the below lemma.
\begin{lemma}
	\label{lem:non-integ}
	Let $\OPT$ denote the expected weight obtained by an offline optimal algorithm. Let $\vec{f}^*$ denote the optimal solution to the above $\LP$. Then $ \sum_{e \in E} w_e f_e^* p_e \geq \mathbb{E}[\OPT]$.
	\end{lemma}
\begin{proof}
For each edge $e$, let $Y_e$ indicate if $e$ is probed (not necessarily matched) in an offline optimal algorithm after observing the full arrival sequence $\cA$. Let $y_e \doteq \E_\mathcal{A}[Y_e]$ for every edge $e \in E$. Note that $\E[\OPT]=\sum_{e \in E}  w_e y_e p_e$. Now we show that $\y \doteq (y_e)_{e\in E}$ is feasible solution to  \LP~\eqref{lp2:stoch-match}.

Consider a given $u$. Let $Z_e$ indicate if $e$ is present when probed with mean $p_e$. Observe that $\sum_{e \in \pa(u)} Y_e Z_e$ indicate if $u$ is matched in $\OPT$. For any given realization of $\cA$, we have $\sum_{e \in \pa(u)} Y_e Z_e \le 1$ since $u$ can be matched at most once. Thus, by linearity of expectation, we have $\E[\sum_{e \in \pa(u)} Y_e Z_e \le 1] \le 1$, which implies that $\sum_{e \in \pa(u)} y_e p_e \le 1$. Thus, Constraint~\eqref{con:lp2-1} is valid. 

Consider a given $v$. Let $R_v$ be the (random) number of copies in $\cA$. Observe that $\sum_{e \in \pa(v)} Y_e \le R_v$. By taking expectation over randomness of $\cA$ on both sides, we get $\E[\sum_{e \in\pa(v)} Y_e ] \le \E[R_v]=r_v$. Thus, Constraint~\eqref{con:lp2-2} is valid. 

Hence, we have that the expected performance of an offline optimal is upper bounded by the optimal value to \LP~\eqref{lp2:stoch-match}.
\end{proof}

Our algorithm is summarized in Algorithm~\ref{alg:non-integral}. Notice that Constraint~\eqref{con:lp2-2} ensures that Step \ref{alg:step2} is valid. \myEdit{For a given $v$, recall that $\partial(v)$ is the set of edges incident to $v$ in $E$}.

\begin{algorithm}[!h]
	\caption{$\mathsf{SM}$}\label{alg:non-integral}
	\DontPrintSemicolon
	Construct and solve \LP~\eqref{lp2:stoch-match}. WLOG assume 
	$\{f_{e}| e \in E\}$ is an optimal solution. \;
	
	\myEdit{When a vertex $v$ arrives, sample an edge $e=(u,v) \in \partial(v)$ with probability $\frac{f_e}{r_v}$. Assign $v$ to $u$ if $u$ is not matched}.\; \label{alg:step2}
\end{algorithm}

\begin{theorem}\label{thm:non-integral}
For edge-weighted online stochastic matching with arbitrary arrival rates and stochastic rewards, online algorithm $\mathsf{SM}$ \eqref{alg:non-integral} achieves a competitive ratio of $1-1/\euler$, which is optimal all among all non-adaptive algorithms. 
\end{theorem}

\begin{proof}
	Let $B(u,t)$ be the event that $u$ is safe (\myEdit{\ie $u$ is not matched}) at beginning of round $t$ and $A(u, t)$ be the event that vertex $u$ is matched during the round $t$ conditioned on $B(u,t)$. From the algorithm, we know 
	$\Pr[A(u,t)] \leq \sum\limits_{ \myEdit{e=(u,v)\in \partial(u)}} \frac{r_v}{n} \frac{f_e}{r_v} p_e \leq \frac{1}{n}$, which is followed by $\Pr[B(u,t)] = \Pr \left[ \bigwedge_{i=1}^{t-1} (\neg A(u, i)) \right]  \ge \left( 1- \frac{1}{n} \right)^{t-1}$.
	
	Consider a given edge $e=(u,v)$ in the graph. Notice that the probability that $e$ gets matched in $\mathsf{SM}$ should be
	\begin{eqnarray*}		
		\Pr[e \text{ is matched} ]  &=& \sum\limits_{t=1}^{n} \Pr[\text{$v$ arrives at $t$ and  $B(u,t)$ }] \cdot \frac{f_ep_e}{r_v}\\
		&\geq &\sum\limits_{t=1}^{n} \left( 1- \frac{1}{n} \right)^{t-1} \frac{r_v}{n} \frac{f_e p_e}{r_v}\geq \left( 1-\frac{1}{\euler} \right) f_ep_e 
	\end{eqnarray*}		
	Note that Manshadi \emph{et al.}~\cite{bib:Manshadi} show that no non-adaptive algorithm can possibly achieve a ratio better than $(1-1/\euler)$ for the non-integral arrival rates, even for the case of all $p_{e}=1$. Thus, our algorithm is an optimal non-adaptive algorithm for this model. 
\end{proof}




\section{Integral arrival rates with uniform stochastic rewards.}
\label{sec:uniform_p}

In this section, we consider a special case of the model studied in Section~\ref{sec:nonint} and show that we can indeed surpass the $1-1/\euler$ barrier. We specialize the model in the following two ways. (1) We consider the unweighted case with uniform constant edge probabilities (\emph{i.e.,} $w_e=1$ and $p_e=p$ for some constant $p \in (0, 1]$ for all $e \in E$). The constant $p$ is arbitrary, but independent of the problem parameters. (2) Each vertex $v$ that comes online has an integral arrival rate $r_v$ (as usual WLOG $r_v = 1$ and $|V|=n$). We refer to this special model as \emph{unweighted online stochastic matching with integral arrival rates and uniform stochastic rewards}. Note that even for this special case, given an offline instance (\emph{i.e.,} the sequence of realizations for the online arrival), it is unclear if we can efficiently solve or approximate the exact offline optimal within $(1-\epsilon)$ without any extra assumptions. Hence we cannot directly apply the Monte-Carlo simulation technique in~\cite{bib:Manshadi} to approximate the exact expected offline optimal within an arbitrary desired accuracy. Here we present a strengthened LP as the benchmark to upper bound the offline optimal.
\begin{align}
	\label{lp3:stoch-match}
	\max & ~~p \cdot \sum_{e \in E} f_e && \\
	\mbox{s.t.} &
	\sum_{e \in \partial(u)} f_e \cdot p \leq 1 &&\forall u \in U \\
	& \sum_{e \in \partial(v)} f_e \leq 1 &&\forall v \in V \\
	& \sum_{\myEdit{e \in S}} f_e p \le 1-\exp(-|S|p) &&\forall S \subseteq \partial(u), |S| \le 2/p \label{lp3_extra} \\
	& \myEdit{0\le f_e } && \forall e\in E 
\end{align}

\begin{lemma}
\LP~\eqref{lp3:stoch-match} is a valid upper bound for the expected offline optimal.
\end{lemma}

\myEdit{\begin{proof}
It suffices to show that constraint \eqref{lp3_extra} is valid (the correctness of the other constraints follows from the previous section). Let $f_e$ represent the expected number of probes on edge $e$ in an offline optimal algorithm (denoted by $\OPT$). Consider a given $S \subseteq \partial(u)$ and let $X_S \in \{0,1\}^{|S|}$ be the indicators for edges in $S$ to be matched in $\OPT$. By definition we have $\E[X_S]=\sum_{e \in S} f_e \cdot p$. Let $Y_S$ be the (random) number of arrivals of all vertices incident to edges in $S$ during the online phase. Observe that $\E[X_S| Y_S] \le 1-(1-p)^{Y_S}$. Thus, we have,
$$\E[X_S]=\mathbb{E}_{Y}[\E[X_S| Y_S]] \le {\E}_{Y_S}[1-(1-p)^{Y_S}].$$

Note that for any constant size $|S| \le 2/p$, $Y_S$ follows a Poisson distribution with mean $|S|$ (since we assume that the total number of online rounds $n$ is sufficiently large). Therefore, we have
\begin{align*}
\E[X_S] & \le  {\E}_{Y_S}\Big[1-(1-p)^{Y_S}\Big] =1-{\E}_{Y_S}[(1-p)^{Y_S}]\\ 
&=1-\exp(-|S|)\sum_{k=0}^{\infty} \frac{|S|^k}{k!} (1-p)^k =1- 
\exp(-|S|)\sum_{k=0}^{\infty} \frac{\Big(|S| (1-p) \Big)^k}{k!}\\
&=1-\exp\Big( -|S|+|S|(1-p)\Big)=1-\exp(-p|S|)
\end{align*}

 Therefore we show that $\vec{f}$ is feasible to constraint \eqref{lp3_extra}.
\end{proof}}

Note that it is impossible to beat $1-1/\euler$ using LP~\eqref{lp3:stoch-match} as the benchmark without the extra constraint \eqref{lp3_extra} (see the hardness instance shown in \cite{brubach2017}). Our main idea in the online phase is based on \cite{bib:Manshadi}. In the offline phase, we first solve LP~\eqref{lp3:stoch-match} and get an optimal solution $\{f^*_e\}$. When a vertex $v$ arrives, we generate a random list of two choices based on $\{f_e^*|e\in \partial(v)\}$, denoted by $\cL_v=(\cL_v(1), \cL_v(2))$, where $\cL_v(1), \cL(2) \in \partial(v)$. Our online decision based on $\cL_v$ is as follows: if $\cL_v(1)=(u,v)$ is safe, i.e., $u$ is available, then match $v$ to $u$; else if the second choice $\cL_v(2)$ is safe match $v$ to $\cL_v(2)$. The random list $\cL_v$ generated based on $\{f_e^*|e\in \partial(v)\}$ satisfies the following two properties:

\begin{description}
\item [\textbf{(P1)}:] $\Pr[\cL_v(1)=e]=f_e^*$  and $\Pr[\cL_v(2)=e]=f_e^*$ for each $e \in \partial(v)$.
	\item	[\textbf{(P2)}:] $\Pr[\cL_v(1)=e \wedge \cL_v(2)= e] =\max\Big(2f_e-1, 0 \Big)$ for each $e \in \partial(v)$.
\end{description}

\begin{algorithm}[!h]
	\caption{}\label{alg:uniform_p}
	\DontPrintSemicolon
	Solve LP~\eqref{lp3:stoch-match} and let 
	$\{f^*_{e}| e \in E\}$ be an optimal solution. \;
When a vertex $v$ arrives, generate a random list $\cL_v$ of two choices based on $\{f_e^*|e\in \partial(v)\}$ such that $\cL_v$ satisfies Property \textbf{(P1)} and \textbf{(P2)}.\;
If $\cL_v(1)=(u,v)$ is safe, i.e., $u$ is available, then assign $v$ to $u$; else if the second choice $\cL_v(2)$ is safe, match $v$ to it.
\end{algorithm}

There are several ways to generate $\cL_v$ satisfying \textbf{(P1)} and  \textbf{(P2)}. One simple way is shown in Section 4 of~\cite{bib:Manshadi}. Another simple way of obtaining $\cL_v$ as required is by running $\mathsf{DR}[\mathbf{f}^*,2]$ and randomly permuting the two obtained matchings. We can verify that all of the calculations shown in \cite{bib:Manshadi} can be extended here if we incorporate the independent process that each $e$ will be present with probability $p$ after we assign $v$ to $u$. Hence, the final ratio is as follows (this can be viewed as a counterpart to Equation (15) on page 11 of~\cite{bib:Manshadi}).
\begin{multline}\label{eqn:sec6-1}
\frac{\E[\ALG]}{\E[\OPT]} \ge \\ \min_{u \in U} \left( \frac{(1-\euler^{-f'_u})+q'_u \euler^{-2}-(q'_u)^2 \euler^{-1}\Big(\frac{1}{2}-\euler^{-1}\Big)-\euler^{-2}f'_u(1-f'_u)}{f'_u} \right) \doteq F(f'_u, q'_u)
\end{multline}
 
 where $f'_u=\sum_{e \in \partial(u)}f^*_e \cdot p \le 1$ and $q'_u=p \cdot \Big(\sum_{e=(u,v) \in \pa(u)} \Pr[\cL_v(2)=e \wedge \cL_v(1) \neq e] \Big)$. Observe that 
 \[
 q'_u \le p \cdot \Big(\sum_{e=(u,v) \in \pa(u)} \Pr[\cL_v(2)=e]\Big)=p \cdot \Big(\sum_{e \in \pa(u)}f^*_e\Big) = f'_u \le 1 
 \]
We can verify that for each given $f'_u \le 1$, the RHS expression in inequality \eqref{eqn:sec6-1} is an increasing function of $q'_u$ during the interval $[0,1]$. Thus an important step is to lower bound $q'_u$ for a given $f'_u$. The following key lemma can be viewed as a counterpart to Lemma 4.7 of \cite{bib:Manshadi}:
 
 \begin{lemma}\label{lem:6-1}
 For each given $f'_u \ge  \ln 2/2$, we have that $q'_u \ge f'_u-(1-\ln 2)$.
 \end{lemma}
 
 \begin{proof}
 Consider a given $u$ with $f'_u \ge \ln2/2$. Define $\Delta=f'_u-q'_u$. Thus we have the following.
 \begin{align*}
 &\Delta&\\
 &=p \cdot \sum_{e=(u,v) \in \pa(u)}\Big( f^*_e-\Pr[\cL_v(2)=e \wedge \cL_v(1) \neq e] \Big) &\\
 &=p \cdot \sum_{e=(u,v) \in \pa(u)}\Big(\Pr[\cL_v(2)=e]- \Pr[\cL_v(2)=e \wedge \cL_v(1) \neq e] \Big) & \mbox{ From \textbf{P1}} \\
&=p \cdot \sum_{e=(u,v) \in \pa(u)}\Big( \Pr[\cL_v(2)=e \wedge \cL_v(1)= e] \Big) 
& \\
&=p \cdot \sum_{e\in \pa(u)}\max \Big( 2f^*_e-1,0 \Big) 
&\mbox{ From \textbf{P2}} 
 \end{align*}
 
 Thus to lower bound $q'_u$, we essentially need to maximize $\Delta$. Let $S^* \subseteq \pa(u)$ be the set of edges in $\pa(u)$ with $f^*_e \ge 1/2$, which is called a \emph{contributing} edge. 
  Thus we have 
\begin{equation}
\label{eqn:6-1}
\Delta=p \cdot \sum_{e\in \pa(u)}\max \Big( 2f^*_e-1,0 \Big) =p \cdot \sum_{e\in S^*}(2f^*_e-1)=\sum_{e\in S^*} 2p f^*_e-p |S^*|
\end{equation} 
 
 Observe that  
\begin{equation}  \label{eqn:6-2}
 \frac{p}{2} |S^*| \le \sum_{e \in S^*} f^*_e \cdot p \le f'_u \Rightarrow |S^*| \le \frac{2 f'_u}{p} \le \frac{2}{p}
 \end{equation} 
From Constraint \eqref{lp3_extra}, we have $\sum_{e\in S^*} (p f^*_e) \le 1-\exp(-|S^*| p)$. Substituting this inequality back into Equation \eqref{eqn:6-1}, we get 
$$\Delta \le 2-2\exp \big(-|S^*| \cdot p\big)-|S^*| \cdot p$$
It is easy to verify that when $f'_u \ge  \ln 2/2$, the above expression has a maximum value of $1-\ln 2$ when $|S^*| \cdot p=\ln 2$. Thus we have that $\Delta \le 1-\ln2$ and $q'_u \ge f'_u -(1-\ln 2)$. 
 \end{proof}

\begin{theorem}\label{thm:uniform_p}
For unweighted online stochastic matching with integral arrival rates and uniform constant stochastic rewards, there exists an adaptive algorithm which achieves a competitive ratio of at least $0.702$. 
\end{theorem}

\begin{proof}
We need to prove that $F(f'_u, q'_u)$ defined in \eqref{eqn:sec6-1} has a lower bound of $0.702$ for all $f'_u \in [0,1]$.
 
Consider the first case when $f'_u \le \ln 2/2$. It is easy to verify that $F(f'_u, q'_u) \ge F(f'_u, 0) \ge F(\ln2/2,0) \sim 0.8$. Consider the second case when $f'_u \ge \ln 2/2$. From Lemma \ref{lem:6-1}, we have $q'_u \ge f'_u-(1-\ln 2)$. Once again, simple calculations show that
\[
	F(f'_u, q'_u) \ge F\big(f'_u, f'_u-(1-\ln 2)\big) \ge F(1,1-(1-\ln 2)) \sim 0.702
\] 
 \end{proof}

\section{Conclusion and future directions.}
	In this paper, we gave improved algorithms for the Edge-Weighted and Vertex-Weighted models. Previously, there was a gap between the best unweighted algorithm with a ratio of $1 - 2\euler^{-2}$ due to~\cite{bib:Jaillet} and the negative result of $1 - \euler^{-2}$ due to~\cite{bib:Manshadi}. We took a step towards closing that gap by showing that an algorithm can achieve $0.7299 > 1 - 2\euler^{-2}$ for both the unweighted and vertex-weighted variants with integral arrival rates. In doing so, we made progress on Open Questions $3$ and $4$ in the online matching and ad allocation survey~\cite{mehtaBook}. This was possible because our approach of rounding to a simpler fractional solution allowed us to employ a stricter LP. For the edge-weighted variant, we showed that one can significantly improve the power of two choices approach by generating two matchings from the same LP solution. For the variant with edge weights, non-integral arrival rates, and stochastic rewards, we presented a $(1-1/\euler)$-competitive algorithm. This showed that the $0.62 < 1-1/\euler$ bound given in~\cite{mehtaonline} for the adversarial model with stochastic rewards does not extend to the known I.I.D. model.
		 
	 A natural next step in the edge-weighted setting is to use an \emph{adaptive} strategy. For the vertex-weighted problem, one can easily see that the stricter LP we use still has a gap. In addition, we only utilize fractional solutions $\{0, 1/3, 2/3\}$. However, dependent rounding gives solutions in $\{0, 1/k, 2/k, \ldots, \ceil{k(1-1/\euler)}/k \}$; allowing for random lists of length greater than three. Stricter LPs and longer lists could both yield improved results. In the stochastic rewards model with non-integral arrival rates, an open question is to either improve upon the $\left(1- \frac{1}{e} \right)$ ratio in the general case. In this work, we showed how for certain restrictions it is possible to beat $1-1/\euler$. However, the serious limitation comes from the fact that a polynomial sized \LP is insufficient to capture the complexity of the problem.

~\\
\noindent
\textbf{Acknowledgments. }
The authors would like to thank Aranyak Mehta and the anonymous reviewers for their valuable comments, which have significantly helped improve the presentation of this paper.

\bibliographystyle{alpha} 
\bibliography{refs}
	
	\appendix
		\section{Appendix}
\subsection{Supplementary materials in section \ref{sec:eweight} (Edge-weighted model).}
\label{apx:AlgEW1}

\subsubsection{Proof of Lemma \ref{lem:ew1} }
\label{apx:analy-EW1}

We will prove Lemma \ref{lem:ew1} using the following three Claims. Recall that we had one kind of large edge, while two kinds of small edges. Hence, the following claim characterizes the performance of each of them. \\
\begin{claim} \label{cl:ew1-a}
For a large edge $e$, $\EW_{1}[h]$ \eqref{alg:ew1} with parameter $h$ achieves a competitive ratio of $\RR{\EW_{1}}{ 2/3} = 0.67529+(1-h) \ast 0.00446$.\\
\end{claim}	

\begin{claim} \label{cl:ew1-b}
For a small edge $e$ of type $\Gamma_{1}$,  $\EW_{1}[h]$ \eqref{alg:ew1} achieves a competitive ratio of $\RR{\EW_{1}}{ 1/3} =0.751066$, regardless of the value $h$.\\
\end{claim}	

\begin{claim} \label{cl:ew1-c}
For a small edge $e$ of type $\Gamma_{2}$,  $\EW_{1}[h]$ \eqref{alg:ew1} achieves a competitive ratio of $\RR{\EW_{1}}{1/3}= 0.72933+ h \ast 0.040415$.\\
\end{claim}	
By setting $h = 0.537815$,  the two types of small edges have the same ratio and we get that $\EW_{1}[h]$ achieves $(\RR{\EW_{1} }{2/3}, \RR{\EW_{1}} { 1/3}) =(0.679417, 0.751066)$. Thus, this proves Lemma \ref{lem:ew1}.	\\

\xhdr{Proof of Claim \ref{cl:ew1-a}.}
Consider a large edge $e=(u, v_1)$ in the graph $G_{\FF}$. Let $e'=(u,v_{2})$ be the other small edge incident to $u$. Edges $e$ and $e'$ can appear in $[M_{1}, M_{2}, M_{3}]$ in the following three ways.		
				\begin{itemize}
					\item
						$\alpha_1$: $e \in M_1, e' \in M_2, e\in M_3$.
					\item
						$\alpha_2$: $e' \in M_1, e \in M_2, e\in M_3$.					\item
						$\alpha_3$: $e \in M_1, e \in M_2, e'\in M_3$.	
				\end{itemize}
Notice that the random triple of matchings $[M_{1}, M_{2}, M_{3}]$ is generated by invoking $\PM{3}$. Since $\PM{3}$ considers a uniform random permutation we have that
$\alpha_{i}$ will occur with probability $1/3$ for $1 \le i \le 3$. For $\alpha_{1}$ and $\alpha_{2}$, we can ignore the second copy of $e$ in $M_3$ and from Lemma \ref{lem:ew0} we have
$$\Pr[\text{$e$ is matched }|~\alpha_1]  \myEdit{=} 0.580831
\text{ and } 
\Pr[\text{$e$ is matched }|~\alpha_2]    \myEdit{=} 0.148499 $$
For $\alpha_{3}$, we have
\begin{eqnarray*}
\Pr[\text{$e$ is matched }|~\alpha_3] & \ge & \sum\limits_{t=1}^{n} \frac{1}{n} \left( 1 - \frac{2}{n} \right)^{t-1}  + \sum\limits_{t=1}^{n} \frac{1}{n} \left( \frac{t-1}{n} \right) \left( 1 - \frac{2}{n} \right)^{t-2}\\
&& + \sum\limits_{t=1}^{n} \frac{1}{n} \left( \frac{(t-1)(t-2)}{2n^2} \right) \left( 1 - \frac{2}{n} \right)^{t-3}\\
&&
+ (1-h) \sum\limits_{t=1}^{n} \frac{1}{n} \left( \frac{1}{n^3} \right) {t-1 \choose 3} \left( 1 - \frac{2}{n} \right)^{t-4}\\	
 &\geq & 0.621246 + (1-h)\ast 0.00892978						
\end{eqnarray*}
\myEdit{There are four terms in the summation above. The four terms denote the probabilities that $v_1$ comes for \emph{the first time} at some time $t \in [T]$ and $v_2$ arrives for $0$, $1$, $2$ and $3$ times before $t$ respectively. Note that in the last term when $v_2$ comes for a third time at some time before $t$, we need to ensure that $v_3$ never matches $u$ which occurs with probability $1-h$ as described in $\EW_1$}.

\myEdit{Recall that $\RR{\EW_{1}}{ 2/3}$ denotes the competitive ratio for a large edge. By definition, we have
\begin{align*}
\RR{\EW_{1}}{ 2/3}&=\frac{\Pr[e \text{~is matched}]}{2/3}\\
&=\frac{\frac{1}{3}\sum_{i=1}^3\Pr[e \text{~is matched} |~ \alp_i]}{2/3}\\
& \ge 0.67529 + (1-h)\ast 0.00446489
\end{align*}
}

\xhdr{Proof of Claims \ref{cl:ew1-b} and \ref{cl:ew1-c}.}
Consider a small edge $e=(u,v)$ of type $\Gamma_{1}$. Let $e_{1}$ and $e_{2} $ be the two other small edges incident to $u$. For a given triple of matchings $[M_{1}, M_{2}, M_{3}]$, we say $e$ is of type $\psi_{1}$ if $e$ appears in
$M_{1}$ while the other two in the remaining two matchings. Similarly, we define the type $\psi_{2}$ and $\psi_{3}$ for the case where $e$ appears in $M_{2}$ and $M_{3}$ respectively. Notice that the probability that $e$ is of type $\psi_{i},~1 \le i \le 3$ is $1/3$.

Similar to the calculations in the proof of Claim \ref{cl:ew1-a}, we have
$\Pr[\text{$e$ is matched} |~\psi_{1}] \ge 0.571861$,
 $\Pr[\text{$e$ is matched} |~\psi_{2}] \ge 0.144776$ and
 $ \Pr[\text{$e$ is matched} |~\psi_{3}] \ge 0.0344288$.
  Therefore we have 
  $$\Pr[ \text{$e$ is matched}] =\frac{1}{3}\sum\limits_{i=1}^{3} \Pr[\text{$e$ is matched }|~\psi_i]  \ge \frac{1}{3}\RR{\EW_{1}}{ 1/3} $$ 
where $\RR{\EW_{1}}{ 1/3}=0.751066$.

Consider a small edge $e=(u,v)$ of type $\Gamma_{2}$, we define 
type $\beta_{i}, 1 \le  i \le 3$, if $e$ appears in $M_{i}$ while the large edge $e'$ incident to $u$ appears in the remaining two matchings. Similarly,  we have 
$\Pr[\text{$e$ is matched} |~\psi_{1}] \ge 0.580831$,
$ \Pr[\text{$e$ is matched} |~\psi_{2}] \ge 0.148499$ and 
  $\Pr[\text{$e$ is matched} |~\psi_{3}] \ge h *0.0404154$.

Hence, the ratio for a small edge of type $\Gamma_{2}$ is $\RR{\EW_{1}}{ 1/3}=0.72933 + h*0.0404154$. 

\subsubsection{Proof of Lemma \ref{lem:ew2}}
\label{apx:analy-EW2}

We will prove Lemma \ref{lem:ew2} using the following two Claims. \\
\begin{claim} \label{cl:ew2-a}
For a large edge $e$, $\EW_{2}[y_{1},y_{2}]$ \eqref{alg:ew2} achieves a competitive ratio of 
$$\RR{\EW_{2}}{ 2/3} =\min\Big( 0.948183 - 0.099895 y_{1} - 0.025646 y_{2}, 0.871245\Big)$$\\
\end{claim}	
	
\begin{claim} \label{cl:ew2-b}
For a small edge $e$, $\EW_{2}[y_{1},y_{2}]$ \eqref{alg:ew2} achieves a competitive ratio of $\RR{\EW_{2}}{ 1/3} =0.4455$, when  $y_{1}=0.687, y_{2}=1$.\\
\end{claim}	

Therefore, by setting $y_{1}=0.687, y_{2}=1$ we get that $\RR{\EW_{2}}{ 2/3} = 0.8539 $ and $\RR{\EW_{2}}{ 1/3} = 0.4455 $, which proves Lemma \ref{lem:ew2}. \\


\xhdr{Proof of Claim \ref{cl:ew2-a}.}
Figure \ref{fig:ew-large} shows the two possible configurations for a large edge. 
	\begin{figure}[!h]
	\centering
	\begin{tikzpicture}
[
	xscale=1,yscale=1,auto,thick,
  	gray node/.style={circle,
  		inner sep=0pt,minimum size=14pt, 
  		fill=black!20,draw,font=\small},  
  	white node/.style={circle,
  		inner sep=0pt,minimum size=14pt, 
  		fill=white,draw,font=\small},
  node distance=24pt 
]


	\node [gray node] (A_u) at (0,0) {$u$};
	\node [gray node] (A_v1) [right of=A_u, xshift=30, yshift=15] {$v_1$};
	\node [white node] (A_v2) [right of=A_u, xshift=30, yshift=-15] {$v_2$};
	\draw[ultra thick] (A_u)  -- (A_v1);
	\draw[thin] (A_u)  -- (A_v2);
	\node (A_v1_1) [left of=A_v1, yshift=10] {};
		\draw[thin] (A_v1) -- (A_v1_1);
	\node (A_v2_1) [left of=A_v2, yshift=-10] {};
		\draw[thin] (A_v2) -- (A_v2_1);
	\node (A_v2_2) [left of=A_v2] {};
		\draw[thin] (A_v2) -- (A_v2_2);
	\node (A) [above of=A_u, xshift=0, yshift=0] {(A)};


	\node [gray node] (B_u) at (4,0) {$u$};
	\node [gray node] (B_v1) [right of=B_u, xshift=30, yshift=15] {$v_1$};
	\node [white node] (B_v2) [right of=B_u, xshift=30, yshift=-15] {$v_2$};
	\draw[ultra thick] (B_u)  -- (B_v1);
	\draw[thin] (B_u)  -- (B_v2);
	\node (B_v1_1) [left of=B_v1, yshift=10] {};
		\draw[thin] (B_v1) -- (B_v1_1);
	\node (B_v2_1) [left of=B_v2, yshift=-10] {};
		\draw[ultra thick] (B_v2) -- (B_v2_1);
	\node (B) [above of=B_u, xshift=0, yshift=0] {(B)};



  



\end{tikzpicture}
	\caption{Diagram of configurations for a large edge $e=(u,v_{1})$. Thin and Thick lines represent small and large edges respectively.}
	\label{fig:ew-large}
	\end{figure}
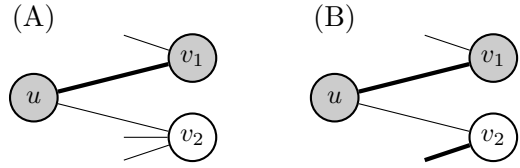

Consider a large edge $e=(u,v_{1})$ with the configuration $(A)$. From $\mathsf{PM}^{*}[\FF,2][y_{1},y_{2}]$, we know that $e$ will always be in $M_{1}$ while $e'=(u,v_{2})$ will be in $M_{1}$ and $M_{2}$ with probability $y_{1}/3$ and $y_{2}/3$ respectively. \\

We now have the following cases:
\myEdit{
\begin{itemize}
\item	$\alpha_1$: $e \in M_1$ and $e'\in M_1$. This happens with probability $y_{1}/3$. In this case, $e$ is matched if $v_1$ comes for the first time at some time $t \in [T]$ and $v_2$ never comes before $t$. Thus, 

$$\Pr[\text{$e$ is matched }|~\alpha_1]  =\sum_{t=1}^n \frac{1}{n} \Big(1-\frac{2}{n}\Big)^{t-1}\geq 0.432332$$

\item   $\alpha_2$: $e\in M_1$ and $e'\in M_2$. This happens with probability $y_{2}/3$.  In this case, $e$ is matched if $v_1$ comes for the first time at some time $t \in [T]$ and $v_2$ comes at most once before $t$. Note that this case is essentially the same as $P_1$ described in Lemma~\ref{lem:ew0}. Thus, we have

$$\Pr[\text{$e$ is matched }|~\alpha_2]  \geq  0.580831$$

\item	$\alpha_3$: $e\in M_1$ and $e'\not \in M_1,  
e' \not \in M_{2}$. This happens with probability $(1-y_{1}/3-y_{2}/3)$. In this case, $e$ is matched if $v_1$ comes at least once. Thus, $\Pr[\text{$e$ is matched}|~\alpha_1] =1-1/\euler \geq 0.632121$.\\
\end{itemize}
}

Therefore we have 
\begin{eqnarray*}
 \Pr[ \text{$e$ is matched}] &=&\Big(\frac{y_{1}}{3}\Pr[\text{$e$ is matched }|~\alpha_1]  +\frac{y_{2}}{3}\Pr[\text{$e$ is matched }|~\alpha_2]\\
 &&+ (1-\frac{y_{1}}{3}-\frac{y_{2}}{3})\Pr[\text{$e$ is matched }|~\alpha_3] \Big)\\
  &\ge &\frac{2}{3}  (0.948183 - 0.099895 y_{1} - 0.025646 y_{2})
  \end{eqnarray*}

Consider the configuration $(B)$. From $\mathsf{PM}^{*}[\FF,2][y_{1},y_{2}]$, we know that $e$ will always be in $M_{1}$ and $e'=(u,v_{2})$ will always  be in $M_{2}$. Thus we have 
$$\Pr[ \text{$e$ is matched}] =\Pr[\text{$e$ is matched }|~\alpha_2] =\frac{2}{3}*0.871245$$

Hence, this completes the proof of Claim \ref{cl:ew2-a}. 

\xhdr{Proof of Claim \ref{cl:ew2-b}.}
Figure \ref{fig:ew-small} shows all possible configurations for a small edge. 
	\begin{figure}[!h]
	\centering
	\begin{tikzpicture}
[
	xscale=0.8,yscale=1,auto,thick,
  	gray node/.style={circle,
  		inner sep=0pt,minimum size=14pt, 
  		fill=black!20,draw,font=\small},  
  	white node/.style={circle,
  		inner sep=0pt,minimum size=14pt, 
  		fill=white,draw,font=\small},
  node distance=24pt 
]


	\node [gray node] (1a_u) at (0,0) {$u$};
	\node [gray node] (1a_v1) [right of=1a_u, xshift=30, yshift=15] {$v_1$};
	\node [white node] (1a_v2) [right of=1a_u, xshift=30, yshift=-15] {$v_2$};
	\draw[thin] (1a_u)  -- (1a_v1);
	\draw[ultra thick] (1a_u)  -- (1a_v2);
	\node (1a_v1_1) [left of=1a_v1, yshift=10] {};
		\draw[ultra thick] (1a_v1) -- (1a_v1_1);
	\node (1a_v2_1) [left of=1a_v2, yshift=-10] {};
		\draw[thin] (1a_v2) -- (1a_v2_1);
	\node (1a) [above of=1a_u, xshift=0, yshift=0] {(1a)};


	\node [gray node] (1b_u) at (3.75,0) {$u$};
	\node [gray node] (1b_v1) [right of=1b_u, xshift=30, yshift=15] {$v_1$};
	\node [white node] (1b_v2) [right of=1b_u, xshift=30, yshift=-15] {$v_2$};
	\draw[thin] (1b_u)  -- (1b_v1);
	\draw[ultra thick] (1b_u)  -- (1b_v2);
	\node (1b_v1_1) [left of=1b_v1, yshift=10] {};
		\draw[thin] (1b_v1) -- (1b_v1_1);
	\node (1b_v1_2) [left of=1b_v1] {};
		\draw[thin] (1b_v1) -- (1b_v1_2);
	\node (1b_v2_1) [left of=1b_v2, yshift=-10] {};
		\draw[thin] (1b_v2) -- (1b_v2_1);
	\node (1b) [above of=1b_u, xshift=0, yshift=0] {(1b)};


	\node [gray node] (2a_u) at (7.5,0) {$u$};
	\node [gray node] (2a_v1) [right of=2a_u, xshift=30, yshift=25] {$v_1$};
	\node [white node] (2a_v2) [right of=2a_u, xshift=30, yshift=0] {$v_2$};
	\node [white node] (2a_v3) [right of=2a_u, xshift=30, yshift=-25] {$v_3$};
	\draw[thin] (2a_u)  -- (2a_v1);
	\draw[thin] (2a_u)  -- (2a_v2);
	\draw[thin] (2a_u)  -- (2a_v3);
	\node (2a_v1_1) [left of=2a_v1, yshift=10] {};
		\draw[ultra thick] (2a_v1) -- (2a_v1_1);
	\node (2a_v2_1) [left of=2a_v2, yshift=-10] {};
		\draw[ultra thick] (2a_v2) -- (2a_v2_1);
	\node (2a_v3_1) [left of=2a_v3, yshift=-10] {};
		\draw[ultra thick] (2a_v3) -- (2a_v3_1);
	\node (2a) [above of=2a_u, xshift=0, yshift=0] {(2a)};


	\node [gray node] (2b_u) at (11.25,0) {$u$};
	\node [gray node] (2b_v1) [right of=2b_u, xshift=30, yshift=25] {$v_1$};
	\node [white node] (2b_v2) [right of=2b_u, xshift=30, yshift=0] {$v_2$};
	\node [white node] (2b_v3) [right of=2b_u, xshift=30, yshift=-25] {$v_3$};
	\draw[thin] (2b_u)  -- (2b_v1);
	\draw[thin] (2b_u)  -- (2b_v2);
	\draw[thin] (2b_u)  -- (2b_v3);
	\node (2b_v1_1) [left of=2b_v1, yshift=10] {};
		\draw[thin] (2b_v1) -- (2b_v1_1);
	\node (2b_v1_2) [left of=2b_v1] {};
		\draw[thin] (2b_v1) -- (2b_v1_2);
	\node (2b_v2_1) [left of=2b_v2, yshift=-10] {};
		\draw[ultra thick] (2b_v2) -- (2b_v2_1);
	\node (2b_v3_1) [left of=2b_v3, yshift=-10] {};
		\draw[ultra thick] (2b_v3) -- (2b_v3_1);
	\node (2b) [above of=2b_u, xshift=0, yshift=0] {(2b)};


	\node [gray node] (3a_u) at (0,-3) {$u$};
	\node [gray node] (3a_v1) [right of=3a_u, xshift=30, yshift=25] {$v_1$};
	\node [white node] (3a_v2) [right of=3a_u, xshift=30, yshift=0] {$v_2$};
	\node [white node] (3a_v3) [right of=3a_u, xshift=30, yshift=-25] {$v_3$};
	\draw[thin] (3a_u)  -- (3a_v1);
	\draw[thin] (3a_u)  -- (3a_v2);
	\draw[thin] (3a_u)  -- (3a_v3);
	\node (3a_v1_1) [left of=3a_v1, yshift=10] {};
		\draw[ultra thick] (3a_v1) -- (3a_v1_1);
	\node (3a_v2_1) [left of=3a_v2, yshift=10] {};
		\draw[thin] (3a_v2) -- (3a_v2_1);
	\node (3a_v2_2) [left of=3a_v2, yshift=-10] {};
		\draw[thin] (3a_v2) -- (3a_v2_2);
	\node (3a_v3_1) [left of=3a_v3, yshift=-10] {};
		\draw[ultra thick] (3a_v3) -- (3a_v3_1);
	\node (3a) [above of=3a_u, xshift=0, yshift=0] {(3a)};


	\node [gray node] (3b_u) at (3.75,-3) {$u$};
	\node [gray node] (3b_v1) [right of=3b_u, xshift=30, yshift=25] {$v_1$};
	\node [white node] (3b_v2) [right of=3b_u, xshift=30, yshift=0] {$v_2$};
	\node [white node] (3b_v3) [right of=3b_u, xshift=30, yshift=-25] {$v_3$};
	\draw[thin] (3b_u)  -- (3b_v1);
	\draw[thin] (3b_u)  -- (3b_v2);
	\draw[thin] (3b_u)  -- (3b_v3);
	\node (3b_v1_1) [left of=3b_v1, yshift=10] {};
		\draw[thin] (3b_v1) -- (3b_v1_1);
	\node (3b_v1_2) [left of=3b_v1] {};
		\draw[thin] (3b_v1) -- (3b_v1_2);
	\node (3b_v2_1) [left of=3b_v2, yshift=10] {};
		\draw[thin] (3b_v2) -- (3b_v2_1);
	\node (3b_v2_2) [left of=3b_v2, yshift=-10] {};
		\draw[thin] (3b_v2) -- (3b_v2_2);
	\node (3b_v3_1) [left of=3b_v3, yshift=-10] {};
		\draw[ultra thick] (3b_v3) -- (3b_v3_1);
	\node (3b) [above of=3b_u, xshift=0, yshift=0] {(3b)};


	\node [gray node] (4a_u) at (7.5,-3) {$u$};
	\node [gray node] (4a_v1) [right of=4a_u, xshift=30, yshift=25] {$v_1$};
	\node [white node] (4a_v2) [right of=4a_u, xshift=30, yshift=0] {$v_2$};
	\node [white node] (4a_v3) [right of=4a_u, xshift=30, yshift=-25] {$v_3$};
	\draw[thin] (4a_u)  -- (4a_v1);
	\draw[thin] (4a_u)  -- (4a_v2);
	\draw[thin] (4a_u)  -- (4a_v3);
	\node (4a_v1_1) [left of=4a_v1, yshift=10] {};
		\draw[ultra thick] (4a_v1) -- (4a_v1_1);
	\node (4a_v2_1) [left of=4a_v2, yshift=10] {};
		\draw[thin] (4a_v2) -- (4a_v2_1);
	\node (4a_v2_2) [left of=4a_v2, yshift=-10] {};
		\draw[thin] (4a_v2) -- (4a_v2_2);
	\node (4a_v3_1) [left of=4a_v3, yshift=0] {};
		\draw[thin] (4a_v3) -- (4a_v3_1);
	\node (4a_v3_2) [left of=4a_v3, yshift=-10] {};
		\draw[thin] (4a_v3) -- (4a_v3_2);
	\node (4a) [above of=4a_u, xshift=0, yshift=0] {(4a)};


	\node [gray node] (4b_u) at (11.25,-3) {$u$};
	\node [gray node] (4b_v1) [right of=4b_u, xshift=30, yshift=25] {$v_1$};
	\node [white node] (4b_v2) [right of=4b_u, xshift=30, yshift=0] {$v_2$};
	\node [white node] (4b_v3) [right of=4b_u, xshift=30, yshift=-25] {$v_3$};
	\draw[thin] (4b_u)  -- (4b_v1);
	\draw[thin] (4b_u)  -- (4b_v2);
	\draw[thin] (4b_u)  -- (4b_v3);
	\node (4b_v1_1) [left of=4b_v1, yshift=10] {};
		\draw[thin] (4b_v1) -- (4b_v1_1);
	\node (4b_v1_2) [left of=4b_v1, yshift=0] {};
		\draw[thin] (4b_v1) -- (4b_v1_2);
	\node (4b_v2_1) [left of=4b_v2, yshift=10] {};
		\draw[thin] (4b_v2) -- (4b_v2_1);
	\node (4b_v2_2) [left of=4b_v2, yshift=-10] {};
		\draw[thin] (4b_v2) -- (4b_v2_2);
	\node (4b_v3_1) [left of=4b_v3, yshift=0] {};
		\draw[thin] (4b_v3) -- (4b_v3_1);
	\node (4b_v3_2) [left of=4b_v3, yshift=-10] {};
		\draw[thin] (4b_v3) -- (4b_v3_2);
	\node (4b) [above of=4b_u, xshift=0, yshift=0] {(4b)};



  



\end{tikzpicture}
	\caption{Diagram of configurations for a small edge $e=(u,v_{1})$. Thin and Thick lines represent small and large edges respectively.}\label{fig:ew-small}
	\end{figure}
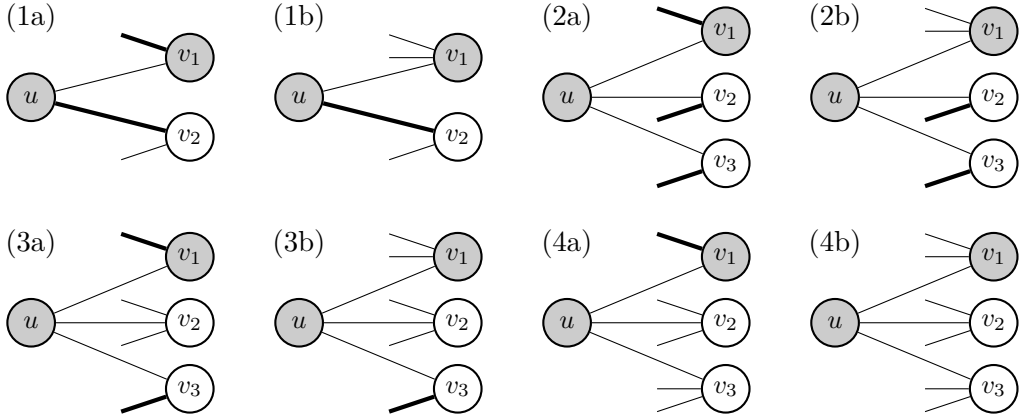

Similar to the proof of Claim \ref{cl:ew2-a}, we will do a case-by-case analysis on the various configurations. Let $e_{i}=(u,v_{i})$ for $1 \le i \le 3$ and $\mathcal{E}$
 be the event that $e_{1}$ gets matched. For a given $e_{i}$, denote $e_{i} \in M_{0}$ if $e_{i} \notin M_{1}, e_{i} \not \in M_{2}$.

\begin{itemize}
\item $(1a)$. Observe that $e_{1} \in M_{2}$ and $e_{2} \in M_{1}$. Thus we have $\Pr[ \EV] =\frac{1}{3}*0.44550$.\\

\item $(1b)$. Observe that we have two cases: $\{ \alpha_{1}: e_{2} \in M_{1}, e_{1} \in M_{1}\}$ and $\{ \alpha_{2}: e_{2} \in M_{1}, e_{1} \in M_{2}\}$. Case $\alpha_{1}$ happens with probability $y_{1}/3$ and the conditional probability is $\Pr[ \EV| \alpha_{1}]=0.432332$. Case $\alpha_{2}$ happens with probability $y_{2}/3$ and the conditional is $\Pr[ \EV| \alpha_{2}]=0.148499$. Thus we have
$$\Pr[\EV]=y_{1}/3 \ast \Pr[ \EV| \alpha_{1}] +y_{2}/3 \ast \Pr[ \EV| \alpha_{2}] \ge
\frac{1}{3}(0.432332 y_1 + 0.148499 y_2 )$$\\

\item $(2a)$. Observe that $e_{1} \in M_{2}, e_{2} \in M_{2}, e_{3} \in M_{2}$. $\Pr[ \EV] =\frac{1}{3}*0.601704$\\

\item $(2b)$. Observe that we have two cases: $\{ \alpha_{1}: e_{1} \in M_{1}, e_{2} \in M_{2}, e_{3} \in M_{2}\}$ and $\{ \alpha_{2}: e_{1} \in M_{2}, e_{2} \in M_{2}, e_{3} \in M_{2}\}$. 
Case $\alpha_{1}$ happens with probability $y_{1}/3$ and the conditional is $\Pr[ \EV| \alpha_{1}]=0.537432$. Case $\alpha_{2}$ happens with probability $y_{2}/3$ and conditional is $\Pr[ \EV| \alpha_{2}]=0.200568$. 
Thus we have
$$\Pr[\EV]=y_{1}/3 \ast \Pr[ \EV| \alpha_{1}] +y_{2}/3 \ast \Pr[ \EV| \alpha_{2}] \ge
\frac{1}{3}(0.537432 y_1 + 0.200568 y_2 )$$\\

\item $(3a)$.  Observe that we have three cases: $\{ \alpha_{1}: e_{1} \in M_{2}, e_{2} \in M_{1}, e_{3} \in M_{2}\}$, $\{ \alpha_{2}: e_{1} \in M_{2}, e_{2} \in M_{2}, e_{3} \in M_{2}\}$ and $\{ \alpha_{2}: e_{1} \in M_{2}, e_{2} \in M_{0}, e_{3} \in M_{2}\}$.
 Case $\alpha_{1}$ happens with probability $y_{1}/3$ and conditional is $\Pr[ \EV| \alpha_{1}]=0.13171$. Case $\alpha_{2}$ happens with probability $y_{2}/3$ and conditional is $\Pr[ \EV| \alpha_{2}]=0.200568$. Case $\alpha_{3}$ happens with probability $(1-y_{1}/3-y_{2}/3)$ and conditional is $\Pr[ \EV| \alpha_{3}]=0.22933$. \\

 Similarly, we have
$$\Pr[\EV]=y_{1}/3 \ast \Pr[ \EV| \alpha_{1}] +y_{2}/3 \ast \Pr[ \EV| \alpha_{2}]+
(1-y_{1}/3-y_{2}/3) \ast \Pr[ \EV| \alpha_{3}] $$
$$\ge
\frac{1}{3}(0.13171 y_1 + 0.200568 y_2+(3-y_{1}-y_{2}) 0.22933 )$$

\item $(3b)$. Observe that we have six cases.
\begin{itemize} 
\item $\alpha_{1}: e_{1} \in M_{1}, e_{2} \in M_{1}, e_{3} \in M_{2}$. $\Pr[\alpha_{1}]=y_{1}^{2}/9$ and $ \Pr[ \EV| \alpha_{1}]=0.4057$.\\
\item $\alpha_{2}: e_{1} \in M_{1}, e_{2} \in M_{2}, e_{3} \in M_{2}$. $\Pr[\alpha_{2}]=y_{1}y_{2}/9$ and $\Pr[ \EV| \alpha_{2}]=0.5374$.\\

\item $\alpha_{3}: e_{1} \in M_{1}, e_{2}  \in M_{0}, e_{3} \in M_{2}$.  $\Pr[\alpha_{3}]=y_{1}/3(1-y_{1}/3-y_{2}/3)$ and $\Pr[ \EV| \alpha_{3}]=0.58083$.\\

\item $\alpha_{4}: e_{1} \in M_{2}, e_{2} \in M_{1}, e_{3} \in M_{2}$. $\Pr[\alpha_{4}]=y_{1}y_{2}/9, \Pr[ \EV| \alpha_{4}]=0.1317$.\\

\item $\alpha_{5}: e_{1} \in M_{2}, e_{2} \in M_{2}, e_{3} \in M_{2}$. $\Pr[\alpha_{5}]=y_{2}^{2}/9, \Pr[ \EV| \alpha_{5}]=0.2006$.\\

\item $\alpha_{6}: e_{1} \in M_{2}, e_{2}  \in M_{0}, e_{3} \in M_{2}$. $\Pr[\alpha_{6}]=y_{2}/3(1-y_{1}/3-y_{2}/3)/3$ and $\Pr[ \EV| \alpha_{6}]=0.22933$.\\

\end{itemize}

Therefore we have
\begin{eqnarray*}
\Pr[\EV] &\ge&  \frac{1}{3} \Big(0.135241y_1^2 +
0.223033 y_1 y_2+ 0.066856 y_2^2	\\
&&+y_1(3-y_1-y_2)0.193610 
+ y_2(3-y_1-y_2)0.076443 \Big)
\end{eqnarray*}

\item $(4a)$. Observe that we have following six cases.
\begin{itemize} 
\item $\alpha_{1}: e_{1} \in M_{2}, e_{2} \in M_{1}, e_{3} \in M_{1}$. $\Pr[\alpha_{1}]=y_{1}^{2}/9$ and $ \Pr[ \EV| \alpha_{1}]=0.08898$.\\

\item $\alpha_{2}: e_{1} \in M_{2}, e_{2} \in M_{2}, e_{3} \in M_{2}$. $\Pr[\alpha_{2}]=y_{2}^{2}/9$ and $ \Pr[ \EV| \alpha_{2}]=0.2006$.\\

\item $\alpha_{3}: e_{1} \in M_{2}, e_{2} \in M_{0}, e_{3} \in M_{0}$. $\Pr[\alpha_{3}]=(1-y_{1}/3-y_{1}/3)^{2},$ and $ \Pr[ \EV| \alpha_{3}]=0.2642$.\\

\item $\alpha_{4}$: $e_{1} \in M_{2}$ while either $e_{2} \in M_{1}, e_{3} \in M_{2}$ or $e_{2} \in M_{2}, e_{3} \in M_{1}$. 
$\Pr[\alpha_{2}]=2y_{1}y_{2}/9$ and $ \Pr[ \EV| \alpha_{4}]=0.1317$.\\

\item $\alpha_{5}$: $e_{1} \in M_{2}$ while either $e_{2} \in M_{1}, e_{3} \in M_{0}$ or $e_{2} \in M_{0}, e_{3} \in M_{1}$. 
$\Pr[\alpha_{5}]=2y_{1}/3(1-y_{1}/3-y_{2}/3)$ and $\Pr[ \EV| \alpha_{5}]=0.14849$.\\

\item $\alpha_{6}$: $e_{1} \in M_{2}$ while either $e_{2} \in M_{2}, e_{3} \in M_{0}$ or $e_{2} \in M_{0}, e_{3} \in M_{2}$. 
$\Pr[\alpha_{5}]=2y_{2}/3(1-y_{1}/3-y_{2}/3)$ and $\Pr[ \EV|\alpha_{6}]=0.22933$.\\

\end{itemize}

Therefore we have
\begin{eqnarray*}
\Pr[\EV] &\ge &\frac{1}{3} 
				\Big( 0.029661y_1^2 + 2 \ast 0.043903y_1y_2 + 0.066856y_2^2  + 2 y_1(3-y_1-y_2) 0.0494997	\\
						&&  + 2 y_2(3-y_1-y_2)(0.076443) +(3 - y_1-y_2)^2 0.0880803 \Big)
\end{eqnarray*}

\item $(4b)$. Observe that in this configuration, we have additional six cases to the ones discussed in $(4a)$. Let $\alpha_{i}$ be the cases defined in $(4a)$ for each $1 \le i \le 6$. Notice that each $\Pr[\alpha_{i}]$ has a multiplicative factor of $y_{2}/3$.  
Now, consider the six new cases. 

\begin{itemize} 

\item $\beta_{1}: e_{1} \in M_{1}, e_{2} \in M_{1}, e_{3} \in M_{1}$. $\Pr[\alpha_{1}]=y_{1}^{3}/27$ and $ \Pr[ \EV| \alpha_{1}]=0.3167$. \\

\item $\beta_{2}: e_{1} \in M_{1}, e_{2} \in M_{2}, e_{3} \in M_{2}$. $\Pr[\alpha_{2}]=y_{1}y_{2}^{2}/27$ and $ \Pr[ \EV| \alpha_{2}]=0.5374$.\\

\item $\beta_{3}: e_{1} \in M_{1}, e_{2} \in M_{0}, e_{3} \in M_{0}$. $\Pr[\alpha_{3}]=y_{1}/3 \ast (1-y_{1}/3-y_{2}/3)^{2}$ and $ \Pr[ \EV| \alpha_{3}]=0.632$.\\

\item $\beta_{4}$: $e_{1} \in M_{1}$ and either $e_{2} \in M_{1}, e_{3} \in M_{2}$ or $e_{2} \in M_{2}, e_{3} \in M_{1}$. 
$\Pr[\alpha_{2}]=2y_{1}^{2}y_{2}/27$ and $ \Pr[ \EV| \alpha_{4}]=0.4057$.\\

\item $\beta_{5}$: $e_{1} \in M_{1}$ and either $e_{2} \in M_{1}, e_{3} \in M_{0}$ or $e_{2} \in M_{0}, e_{3} \in M_{1}$. 
$\Pr[\alpha_{5}]=2y_{1}^{2}/9 \ast (1-y_{1}/3-y_{2}/3)$ and $ \Pr[ \EV| \alpha_{5}]=0.4323$.\\

\item $\beta_{6}$: $e_{1} \in M_{1}$ and either $e_{2} \in M_{2}, e_{3} \in M_{0}$ or $e_{2} \in M_{0}, e_{3} \in M_{2}$. 
$\Pr[\alpha_{5}]=2y_{1}y_{2}/9 \ast (1-y_{1}/3-y_{2}/3)$ and $\Pr[ \EV|\alpha_{6}]=0.58083$.\\

Hence, we have
\begin{eqnarray*}
\Pr[\EV]&\ge& \frac{1}{3}\Big(0.632 y_{1} - 0.133133 y_{1}^2 + 0.0093y_{1}^3 + 0.264241 y_{2} \\
  &&
 -0.11127 y_{1} y_{2} + 0.01170 y_{1}^2 y_{2} - 0.0232746 y_{2}^2
  + 
 0.00488 y_{1} y_{2}^2 + 0.00068 y_{2}^3\Big) 
\end{eqnarray*}
	
\end{itemize}
\end{itemize}	

Setting $y_{1}=0.687,~y_{2}=1$, we get that the competitive ratio for a small edge is $0.44550$. The bottleneck cases are configurations $(1a)$ and $(1b)$.


\subsection{Supplemental materials in section \ref{sec:vweight}}
\label{apx:c}


	


\subsubsection{Proof of Lemma \ref{lem:offline} (Vertex-weighted and Unweighted) }
\label{apx:offline-vw}

When $H_u=1$ and $u$ is in the cycle $C_{1}$, \cite{bib:Jaillet} show that the competitive ratio of $u$ is $1-2\euler^{-2}$.  Hence, for the remaining cases, we use the following Claims. \\

\begin{claim} \label{cl:vw1a}
	If $H_u=1$ and $u$ is not in $C_{1}$, then we have 
	$\mathsf{R}[\RLA,1]\ge 0.735622$.\\
\end{claim}

\begin{claim}\label{cl:vw2}
	$\RR{\RLA}{2/3} \ge 0.7870$.\\
\end{claim}

\begin{claim}\label{cl:vw3}
	$\RR{\RLA}{1/3} \ge 0.8107$.\\
\end{claim}

Recall that $A_{u,1}$ is the event that among the $n$ random lists, there exists a list starting with $u$ and $A_{u,2}^v$ is the event that among the $n$ lists, there exist successive lists such that (1) all start with some $u'$ which are different from $u$ but are neighbors of $v$; and (2) they ensure $u$ will be matched.

Notice that $A_u$ is the probability that $u$ gets matched in $\mathsf{RLA}[\HH']$. For each $u$, we compute $\Pr[A_{u,1}]$ and $\Pr[A_{u,2}^v]$ for all possibilities of $v \sim u$ and using Lemma \ref{lem:cycle} we get $A_u$. We  first discuss two different ways to calculate $\Pr[A_{u,2}^v]$ when $v$ has different neighboring structures. 

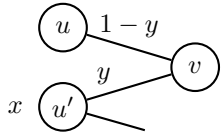
\begin{figure}[h]
	\centering
	\begin{tikzpicture}
[
	xscale=1,yscale=1,auto,thick,font=\small,
  	gray node/.style={circle,
  		inner sep=0pt,minimum size=14pt, 
  		fill=black!20,draw,font=\small},  
  	white node/.style={circle,
  		inner sep=0pt,minimum size=14pt, 
  		fill=white,draw,font=\small},
  node distance=24pt 
]


	\node [white node] (v) at (0,0) {$v$};
	\node [white node] (u1) [left of=v, xshift=-30, yshift=15] {$u$};
	\node [white node] (u2) [left of=v, xshift=-30, yshift=-15] {$u'$};
	\node (v1) [right of=u2, xshift=5, yshift=-10] {};
	\draw[] (u1)  -- (v) node [midway, above] {$1-y$};
	\draw[] (u2)  -- (v) node [above, pos=0.2] {$y$};
	\draw[] (u2)  -- (v1);
	\node (fu2) [left of=u2, xshift=12, yshift=0] {$x$};

\end{tikzpicture}
	\caption{Case 1 in calculation of $\Pr[A_{u,2}^v]$}
	\label{fig:2C1}
\end{figure}

\begin{figure}[h]
	\centering
	\begin{tikzpicture}
[
	xscale=1,yscale=1,auto,thick,font=\small,
  	gray node/.style={circle,
  		inner sep=0pt,minimum size=14pt, 
  		fill=black!20,draw,font=\small},  
  	white node/.style={circle,
  		inner sep=0pt,minimum size=14pt, 
  		fill=white,draw,font=\small},
  	white node/.style={circle,
  		inner sep=0pt,minimum size=14pt, 
  		fill=white,draw,font=\small},
  node distance=24pt 
]


	\node [white node] (v) at (0,0) {$v$};
	\node [white node] (u1) [left of=v, xshift=-30, yshift=30] {$u$};
	\node [white node] (u2) [left of=v, xshift=-30, yshift=0] {$u_1$};
	\node [white node] (u3) [left of=v, xshift=-30, yshift=-30] {$u_2$};
	\draw[] (u1)  -- (v) node [above right, pos=0.1] {$d=0.75$};
	\draw[] (u2)  -- (v) node [above, pos=0.37] {$b=0.1$};
	\draw[] (u3)  -- (v) node [below right, pos=0.1] {$c=0.15$};
	\node (fu1) [left of=u1, xshift=12, yshift=0] {$1$};
	\node (fu2) [left of=u2, xshift=8, yshift=0] {$1/3$};
	\node (fu3) [left of=u3, xshift=8, yshift=0] {$2/3$};

\end{tikzpicture}
	\caption{Case 2 in calculation of $\Pr[A_{u,2}^v]$}
	\label{fig:2C2}
\end{figure}
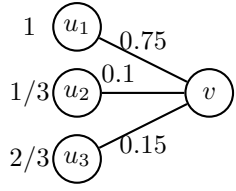

\xhdr{Two ways to compute the value $\Pr[A_{u,2}^v]$.}

\begin{enumerate}
	
	\item \textbf{Case 1: $v$ has two neighbors.}
	Consider the case when $v$ has two neighbors as shown in Figure \ref{fig:2C1}. In this case we choose a slightly direct approach to computing $\Pr[A_{u,2}^v]$. \\
	
Assume $v$ has two neighbors $u$ and $u'$ as shown in Figure \ref{fig:2C1}. After modifications, assume $H'_{(u',v)}=y$, $H'_{(u,v)}=1-y$ and $H'_{u'}=x$. Thus, the second certificate event $A_{u,2}^v$ corresponds to the event (1) a list starting with $u'$ comes at some time $1 \le i <n$; (2) the list $\mathcal{R}_{v}=(u',u)$ comes for a second time at some $j$ with $i<j \le n$. Note that the arrival rate of a list starting with $u'$ is $H'_{u'}=x/n$ and the rate of list $\mathcal{R}_{v}=(u',u)$ is $y/n$. Therefore we have
	\begin{eqnarray}\label{eqn:vw}
	\Pr[A_{u,2}^v] &=&\sum_{i=1}^{n-1}\Big(x/n (1 - x/n)^{(i- 1)} (1 - (1 - y/n)^{(n - i)} \Big)\\
	&\sim  &\frac{x - \euler^{-y} x + (-1 + \euler^{-x}) y}{x - y}  \text{~~~ (if $x \neq y$)}\\
	&\sim  &1- \euler^{-x}(1+x) \text{~~~ (if $x = y$)}
	\end{eqnarray}
	
	\item 
	\textbf{Case 2: $v$ has three neighbors.}
	Consider the case when $v$ has three neighbors  as shown in Figure \ref{fig:2C2}. In this case, we approximate the value $\Pr[A_{u,2}^v]$ using the Markov Chain method, similar to \cite{bib:Jaillet}. \\
	
Focus on the case shown in Figure \ref{fig:2C2} where $v$ has three neighbors $u$, $u_{1}$ and $u_{2}$ with 
	$H_u=1, H_{u_{1}}=1/3$ and $H_{u_{2}}=2/3$. Recall that after modifications, we have 
	$H'_{(u_{1},v)}=b=0.1, H'_{(u_{2},v)}=c=0.15$ and $H'_{(u,v)}=d=0.75$. We simulate the process of $u$ getting matched resulting from several successive random lists starting from either $u_{1}$ or $u_{2}$ by an $n$-step Markov Chain as follows. We have $5$ states: $s_{1}=(0,0,0), s_{2}=(0,1,0), s_{3}=(0,0,1), s_{4}=(0,1,1)$ and $s_{5}=(1,*,*)$ and the three numbers in each triple correspond to $u$, $u_{1}$ and $u_{2}$ being matched(or not) respectively. State $s_{5}$ corresponds to $u$ being matched; the matched status of $u_{1}$ and $u_{2}$ is irrelevant.
	The chain initially starts in $s_{1}$ and the probability of being in state $s_{5}$ after $n$ steps gives an approximation to $\Pr[A_{u,2}^v]$. 
	The one-step transition probability matrix $M$ is shown as follows. 
	\begin{eqnarray*}
		&&M_{1,2}=\frac{b}{n}, M_{1,3}= \frac{c+1/3}{n}, M_{1,1}=1-M_{1,2}-M_{1,3} \\
		&&M_{2,4}=\frac{c+1/3}{n}+\frac{bc}{(c+d)n}, M_{2,5}=\frac{bd}{(c+d)n}, \\
		&& M_{2,3}=1-M_{2,4}-M_{2,5}\\
		&& M_{3,4}=\frac{b}{n}+\frac{cb}{(b+d)n} ,  M_{3,5}=\frac{cd}{(b+d)n}\\
		&&M_{3,3}=1-M_{3,4}-M_{3,5}\\
		&& M_{4,5}=\frac{b+c}{n}, M_{4,4}=1-M_{4,5}\\
		&&M_{5,5}=1 \\
		&& M_{i,j} = 0 \text{ for all other $i, j$}
	\end{eqnarray*}
	
	Notice that $M_{1,3}=\frac{c+1/3}{n}$ and not $\frac{2}{3n}$ since after modifications, the arrival rate of a list starting with $u_{2}$ decreases correspondingly.
	
\end{enumerate}

Let us now prove the three Claims \ref{cl:vw1a}, \ref{cl:vw2} and \ref{cl:vw3}. 
Here we give the explicit analysis for the case when $H_u=1$. For the remaining cases, similar methods can be applied. Hence, we omit the analysis and only present the related computational results which leads to the conclusion. \\


\xhdr{Proof of Claim \ref{cl:vw1a}.}
	Notice that $u$ is not in the cycle $C_{1}$ and thus Lemma \ref{lem:cycle} can be used. Figure \ref{fig:fu=1} describes all possible cases when a node $u \in U$ has $H_u = 1$. (We ignore all those cases when $H_{u}<1$, since they will not appear in the $\WS$.)

	\begin{figure}[h]
		\centering
		\input{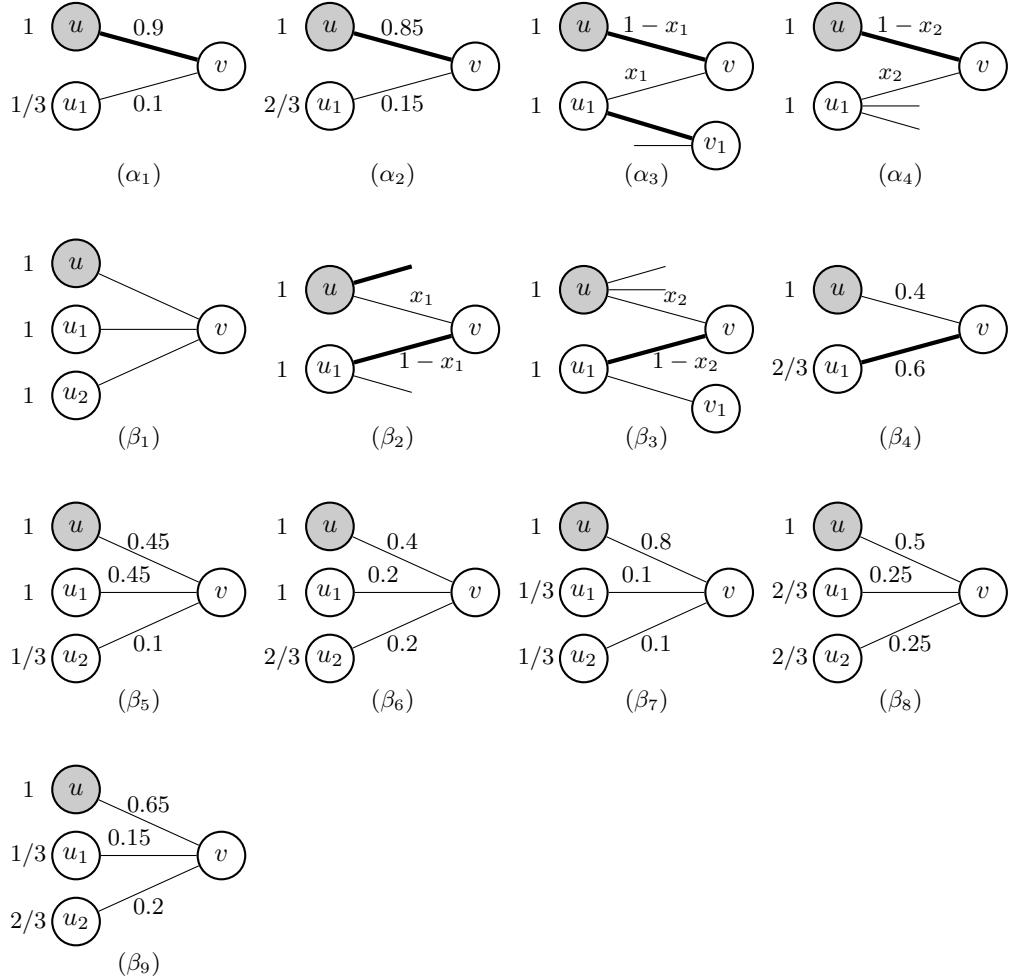}
		\caption{Vertex-weighted $H_u=1$ cases. The value assigned to each edge represents the value after the second modification. No value indicates no modification. Here, $x_1 = 0.2744$ and $x_2 = 0.15877$.}
		\label{fig:fu=1}
	\end{figure}
	
	Let $v_{1}$ and $v_{2}$ be the two neighbors of $u$ with $H_{(u,v_{1})}=2/3$ and $H_{(u,v_{2})}=1/3$. 
	In total, there are $4 \times 10$ combinations, where $v_{1}$ is chosen from some $\alpha_{i}, 1 \le i \le 4$ and $v_{2}$ is chosen from some $\beta_{i}, 1 \le i \le 9$. For $H_u=1$, we need to find the worst combination among these such that the value $A_u$ is minimized. 
	We can find this $\WS$ using the Lemma \ref{lem:cycle}.

	For each type of $\alpha_{i}, \beta_{j}$, we compute the values it will contribute to the term $(1-A_{u,1})\prod_{v \sim u}(1-\Pr[A_{u,2}^v])$. For example, assume $v_{1}$ is of type $\alpha_{1}$, denoted by $v_{1}(\alpha_{1})$. It contributes $\euler^{-0.9}$ to the term $(1-A_{u,1})$ and  $(1-\Pr[A_{u,2}^{v_1}])$ to  $\prod_{v \sim u}(1-\Pr[A_{u,2}^v])$, thus the total value it contributes is $\gamma(v_{1}, \alpha_{1})=\euler^{-0.9}(1-\Pr[A_{u,2}^{v_1}])$. Similarly, we can compute all $\gamma(v_{1}, \alpha_{i})$ and $\gamma(v_{2}, \beta_{j})$. Let 
	$i^{*}=\argmax_{i}\gamma(v_{1}, \alpha_{i}) $  and $j^{*}=\argmax_{j}\gamma(v_{2}, \beta_{j}) $. The $\WS$ is for the combination $\{v_{1}(\alpha_{i^{*}}), v_{2}(\beta_{j^{*}})\}$ and the resulting  value of $A_u$ and $\RR{\RLA}{1}$ is as follows:
	$$A_u=1-\gamma(v_{1}, \alpha_{i^{*}})\gamma(v_{2}, \beta_{j^{*}})$$
	$$\RR{\RLA}{1}=A_u/H_u=A_u$$

	Here is a list of $\gamma(v_{1}, \alpha_{i})$ and $\gamma(v_{2}, \beta_{j})$, for each $1 \le i \le 4$ and $1 \le j \le 9$.

	\begin{itemize}
		\item $\alpha_{1}$: We have $\Pr[A_{u,2}^v]=1-\euler^{-0.1}*1.1$ and
		$\gamma(v, \alpha_{1})=\euler^{-0.1}*1.1*\euler^{-0.9}=0.404667$.\\
		
		\item $\alpha_{2}$: $\Pr[A_{u,2}^v] \ge 1-\euler^{-0.15}*1.15$ and
		$\gamma(v, \alpha_{2}) \le  0.423$.
		
		Notice that after modifications, $H'_{u_{1}} \ge 0.15$. 
		Hence, we use this and Equation \eqref{eqn:vw} to compute the lower bound of $\Pr[A_{u,2}^v]$.\\
		
		\item $\alpha_{3}$: $\Pr[A_{u,2}^v]  \ge 0.0916792 $ and
		$\gamma(v, \alpha_{3}) \le 0.439667$. 
		
		Notice that for any large edge $e$ incident to a node $u$ with $H_u=1$ (before modification), we have after modification, $H'_{e} \ge 1-0.2744=0.7256$. Thus we have $H'_{(u_{1},v_{1})} \ge 0.7256$ and $ H'_{u_{1}} \ge 1$. From Equation \eqref{eqn:vw}, we get $\Pr[A_{u,2}^v]  \ge 0.0916792 $.\\
		
		\item $\alpha_{4}$: $\Pr[A_{u,2}^v]  \ge 0.0307466 $ and
		$\gamma(v, \alpha_{4}) \le 0.417923$. 
		
		Notice that for any small edge $e$ incident to a node $u$ with $H_u=1$ (before modification), we have after modification, $H'_{e} \ge 0.15877$. Thus, we have $H'_{u_{1}} \ge 3*0.15877$. \\
		
		\item $\beta_{1}$: $\Pr[A_{u,2}^v] =0.1608$ and
		$\gamma(v, \beta_{1})=0.601313$.\\
		
		\item $\beta_{2}$: $\Pr[A_{u,2}^v]  \ge 0.208812$ and
		$\gamma(v, \beta_{2}) \le 0.601313 $. \\

		After modifications, we have $H'_{(u_{1},v_{1})} \ge 0.2744$ and thus we get $H'_{u_{1}} \ge 1$. 
		
		\item $\beta_{3}$: $\Pr[A_{u,2}^v]  \ge 0.251611$ and
		$\gamma(v, \beta_{2}) \le 0.63852$.

		After modifications, we have $H'_{(u_{1},v_{1})} \ge 0.2744$ and thus we get $H'_{u_{1}} \ge 1-0.15877+0.2744$. \\

		\item $\beta_{4}$: $\Pr[A_{u,2}^v] =0.121901$ and
		$\gamma(v, \beta_{4})=0.588607$.\\
		
		\item $\beta_{5}$: $\Pr[A_{u,2}^v] =0.1346$ and
		$\gamma(v, \beta_{5})=0.551803$.\\
		
		\item $\beta_{6}$: $\Pr[A_{u,2}^v]  \ge 0.1140$ and
		$\gamma(v, \beta_{6}) \le 0.593904$.\\
		
		\item $\beta_{7}$: $\Pr[A_{u,2}^v]  =  0.0084$ and
		$\gamma(v, \beta_{7}) =0.4455$.\\
		
		\item $\beta_{8}$: $\Pr[A_{u,2}^v]   \ge 0.0397 $ and 
		$\gamma(v, \beta_{8})  \le 0.582451$.\\
		
		\item $\beta_{9}$: $\Pr[A_{u,2}^v]   \ge 0.0230$ and
		$\gamma(v, \beta_{9})  \le 0.510039$.\\
	\end{itemize}
	
	
	Using the computed values above, let us compute the ratio of a node $u$ with $H_u=1$.
	\begin{itemize}
		\item If $u$ has three neighbors, then the $\WS$ configuration is when each of the three neighbors of $u$ is of type $\beta_{3}$. This is because, the value of $\gamma(v, \beta_{3})$ is the largest.  The resultant ratio is $0.73967$.\\
		
		\item If $u$ has two neighbors, then the $\WS$ configuration is when one of the neighbor is of type $\beta_{1}$ (or $\beta_{2}$) and the other is of type $\alpha_{3}$. The resultant ratio is $0.735622$. \\
\end{itemize}

%
%

\xhdr{Proof of Claim \ref{cl:vw2}.}
	The proof is similar to that of Claim \ref{cl:vw1a}. The Figure \ref{fig:Fu=2/3} shows all possible configurations of a node $u$ with $H_u=2/3$. Note that the $\WS$ cannot have $F(v)<1$ and hence we omit them here. For a neighbor $v$ of $u$, if $H_{(u,v)}=2/3$, then  $v$ is in one of $\alpha_{i}, 1 \le i \le 3$; if  $H_{(u,v)}=1/3$, then  $v$ is in one of $\beta_{i}, 1 \le i \le 8$. 
	We now list the values $\gamma(v, \alpha_{i})$ and $\gamma(v, \beta_{j})$, for each $1 \le i \le 3$ and $1 \le j \le 8$.
	
	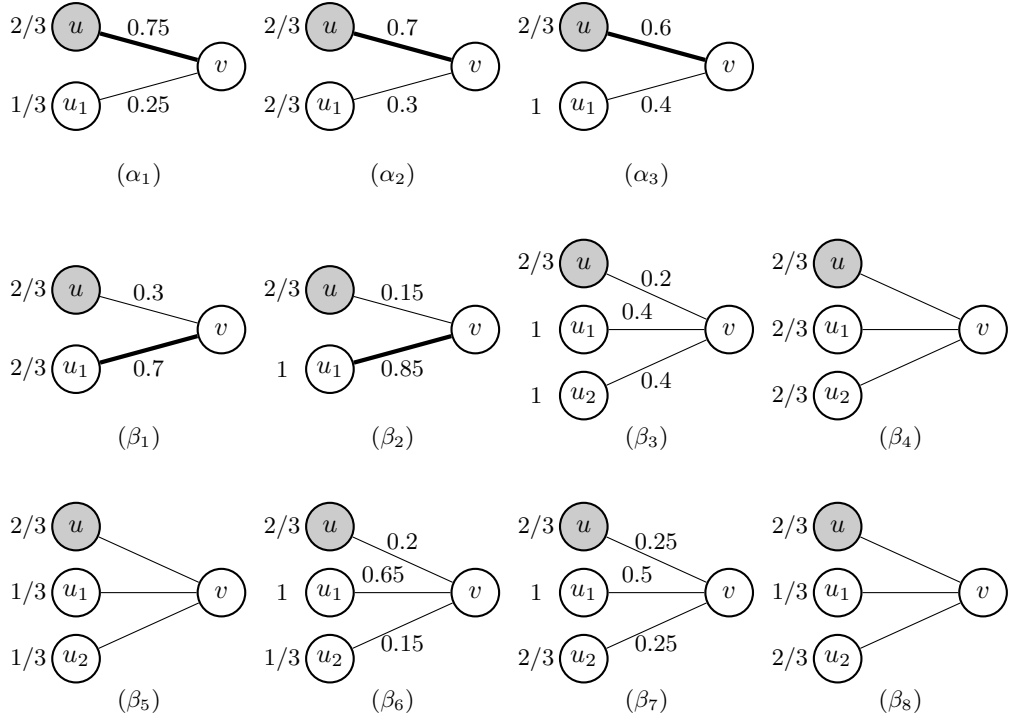
\begin{figure}[h]
		\centering
		\begin{tikzpicture}
[
	xscale=0.8,yscale=1,auto,thick,font=\footnotesize,
  	gray node/.style={circle,
  		inner sep=0pt,minimum size=14pt, 
  		fill=black!20,draw,font=\small},  
  	white node/.style={circle,
  		inner sep=0pt,minimum size=14pt, 
  		fill=white,draw,font=\small},
  node distance=24pt 
]


	\node [white node] (a1v) at (0,0) {$v$};
	\node [gray node] (a1u) [left of=a1v, xshift=-25, yshift=15] {$u$};
	\node [white node] (a1u1) [left of=a1v, xshift=-25, yshift=-15] {$u_1$};
	\draw[ultra thick] (a1u)  -- (a1v) node [midway, above] {$0.75$};
	\draw[thin] (a1u1)  -- (a1v) node [midway, below] {$0.25$};
	\node (a1fu) [left of=a1u, xshift=8, yshift=0] {$2/3$};
	\node (a1fu1) [left of=a1u1, xshift=8, yshift=0] {$1/3$};
	\node [below left of=a1v, xshift=-10, yshift=-20] {($\alpha_1$)};


	\node [white node] (a2v) at (3.75,0) {$v$};
	\node [gray node] (a2u) [left of=a2v, xshift=-25, yshift=15] {$u$};
	\node [white node] (a2u1) [left of=a2v, xshift=-25, yshift=-15] {$u_1$};
	\draw[ultra thick] (a2u)  -- (a2v) node [midway, above] {$0.7$};
	\draw[thin] (a2u1)  -- (a2v) node [midway, below] {$0.3$};
	\node (a2fu) [left of=a2u, xshift=8, yshift=0] {$2/3$};
	\node (a2fu1) [left of=a2u1, xshift=8, yshift=0] {$2/3$};
	\node [below left of=a2v, xshift=-10, yshift=-20] {($\alpha_2$)};


	\node [white node] (a3v) at (7.5,0) {$v$};
	\node [gray node] (a3u) [left of=a3v, xshift=-25, yshift=15] {$u$};
	\node [white node] (a3u1) [left of=a3v, xshift=-25, yshift=-15] {$u_1$};
	\draw[ultra thick] (a3u)  -- (a3v) node [midway, above] {$0.6$};
	\draw[thin] (a3u1)  -- (a3v) node [midway, below] {$0.4$};
	\node (a3fu) [left of=a3u, xshift=8, yshift=0] {$2/3$};
	\node (a3fu1) [left of=a3u1, xshift=12, yshift=0] {$1$};
	\node [below left of=a3v, xshift=-10, yshift=-20] {($\alpha_3$)};


	\node [white node] (b1v) at (0,-3.5) {$v$};
	\node [gray node] (b1u) [left of=b1v, xshift=-25, yshift=15] {$u$};
	\node [white node] (b1u1) [left of=b1v, xshift=-25, yshift=-15] {$u_1$};
	\draw[thin] (b1u)  -- (b1v) node [midway, above] {$0.3$};
	\draw[ultra thick] (b1u1)  -- (b1v) node [midway, below] {$0.7$};
	\node (b1fu) [left of=b1u, xshift=8, yshift=0] {$2/3$};
	\node (b1fu1) [left of=b1u1, xshift=8, yshift=0] {$2/3$};
	\node [below left of=b1v, xshift=-10, yshift=-20] {($\beta_1$)};


	\node [white node] (b2v) at (3.75,-3.5) {$v$};
	\node [gray node] (b2u) [left of=b2v, xshift=-25, yshift=15] {$u$};
	\node [white node] (b2u1) [left of=b2v, xshift=-25, yshift=-15] {$u_1$};
	\draw[thin] (b2u)  -- (b2v) node [midway, above] {$0.15$};
	\draw[ultra thick] (b2u1)  -- (b2v) node [midway, below] {$0.85$};
	\node (b2fu) [left of=b2u, xshift=8, yshift=0] {$2/3$};
	\node (b2fu1) [left of=b2u1, xshift=12, yshift=0] {$1$};
	\node [below left of=b2v, xshift=-10, yshift=-20] {($\beta_2$)};


	\node [white node] (b3v) at (7.5,-3.5) {$v$};
	\node [gray node] (b3u) [left of=b3v, xshift=-25, yshift=25] {$u$};
	\node [white node] (b3u1) [left of=b3v, xshift=-25, yshift=0] {$u_1$};
	\node [white node] (b3u2) [left of=b3v, xshift=-25, yshift=-25] {$u_2$};
	\draw[thin] (b3u)  -- (b3v) node [midway, above] {$0.2$};
	\draw[thin] (b3u1)  -- (b3v) node [above, pos=0.3] {$0.4$};
	\draw[thin] (b3u2)  -- (b3v) node [midway, below] {$0.4$};
	\node (b3fu) [left of=b3u, xshift=8, yshift=0] {$2/3$};
	\node (b3fu1) [left of=b3u1, xshift=12, yshift=0] {$1$};
	\node (b3fu2) [left of=b3u2, xshift=12, yshift=0] {$1$};
	\node [below left of=b3v, xshift=-10, yshift=-20] {($\beta_3$)};


	\node [white node] (b4v) at (11.25,-3.5) {$v$};
	\node [gray node] (b4u) [left of=b4v, xshift=-25, yshift=25] {$u$};
	\node [white node] (b4u1) [left of=b4v, xshift=-25, yshift=0] {$u_1$};
	\node [white node] (b4u2) [left of=b4v, xshift=-25, yshift=-25] {$u_2$};
	\draw[thin] (b4u)  -- (b4v);
	\draw[thin] (b4u1)  -- (b4v);
	\draw[thin] (b4u2)  -- (b4v);
	\node (b4fu) [left of=b4u, xshift=8, yshift=0] {$2/3$};
	\node (b4fu1) [left of=b4u1, xshift=8, yshift=0] {$2/3$};
	\node (b4fu2) [left of=b4u2, xshift=8, yshift=0] {$2/3$};
	\node [below left of=b4v, xshift=-10, yshift=-20] {($\beta_4$)};


	\node [white node] (b5v) at (0,-7) {$v$};
	\node [gray node] (b5u) [left of=b5v, xshift=-25, yshift=25] {$u$};
	\node [white node] (b5u1) [left of=b5v, xshift=-25, yshift=0] {$u_1$};
	\node [white node] (b5u2) [left of=b5v, xshift=-25, yshift=-25] {$u_2$};
	\draw[thin] (b5u)  -- (b5v);
	\draw[thin] (b5u1)  -- (b5v);
	\draw[thin] (b5u2)  -- (b5v);
	\node (b5fu) [left of=b5u, xshift=8, yshift=0] {$2/3$};
	\node (b5fu1) [left of=b5u1, xshift=8, yshift=0] {$1/3$};
	\node (b5fu2) [left of=b5u2, xshift=8, yshift=0] {$1/3$};
	\node [below left of=b5v, xshift=-10, yshift=-20] {($\beta_5$)};


	\node [white node] (b6v) at (3.75,-7) {$v$};
	\node [gray node] (b6u) [left of=b6v, xshift=-25, yshift=25] {$u$};
	\node [white node] (b6u1) [left of=b6v, xshift=-25, yshift=0] {$u_1$};
	\node [white node] (b6u2) [left of=b6v, xshift=-25, yshift=-25] {$u_2$};
	\draw[thin] (b6u)  -- (b6v) node [midway, above] {$0.2$};
	\draw[thin] (b6u1)  -- (b6v) node [above, pos=0.3] {$0.65$};
	\draw[thin] (b6u2)  -- (b6v) node [midway, below] {$0.15$};
	\node (b6fu) [left of=b6u, xshift=8, yshift=0] {$2/3$};
	\node (b6fu1) [left of=b6u1, xshift=12, yshift=0] {$1$};
	\node (b6fu2) [left of=b6u2, xshift=8, yshift=0] {$1/3$};
	\node [below left of=b6v, xshift=-10, yshift=-20] {($\beta_6$)};


	\node [white node] (b7v) at (7.5,-7) {$v$};
	\node [gray node] (b7u) [left of=b7v, xshift=-25, yshift=25] {$u$};
	\node [white node] (b7u1) [left of=b7v, xshift=-25, yshift=0] {$u_1$};
	\node [white node] (b7u2) [left of=b7v, xshift=-25, yshift=-25] {$u_2$};
	\draw[thin] (b7u)  -- (b7v) node [midway, above] {$0.25$};
	\draw[thin] (b7u1)  -- (b7v) node [above, pos=0.3] {$0.5$};
	\draw[thin] (b7u2)  -- (b7v) node [midway, below] {$0.25$};
	\node (b7fu) [left of=b7u, xshift=8, yshift=0] {$2/3$};
	\node (b7fu1) [left of=b7u1, xshift=12, yshift=0] {$1$};
	\node (b7fu2) [left of=b7u2, xshift=8, yshift=0] {$2/3$};
	\node [below left of=b7v, xshift=-10, yshift=-20] {($\beta_7$)};


	\node [white node] (b8v) at (11.25,-7) {$v$};
	\node [gray node] (b8u) [left of=b8v, xshift=-25, yshift=25] {$u$};
	\node [white node] (b8u1) [left of=b8v, xshift=-25, yshift=0] {$u_1$};
	\node [white node] (b8u2) [left of=b8v, xshift=-25, yshift=-25] {$u_2$};
	\draw[thin] (b8u)  -- (b8v);
	\draw[thin] (b8u1)  -- (b8v);
	\draw[thin] (b8u2)  -- (b8v);
	\node (b8fu) [left of=b8u, xshift=8, yshift=0] {$2/3$};
	\node (b8fu1) [left of=b8u1, xshift=8, yshift=0] {$1/3$};
	\node (b8fu2) [left of=b8u2, xshift=8, yshift=0] {$2/3$};
	\node [below left of=b8v, xshift=-10, yshift=-20] {($\beta_8$)};

\end{tikzpicture}
		\caption{Vertex-weighted $H_u=2/3$ cases. The value assigned to each edge represents the value after the second modification. No value indicates no modification.}
		\label{fig:Fu=2/3}
	\end{figure}
	
	\begin{itemize}
		\item $\alpha_{1}$: We have $\Pr[A_{u,2}^v]=1-\euler^{-0.25}*1.25$ and
		$\gamma(v, \alpha_{1})=\euler^{-0.25}*1.25*\euler^{-0.75}=0.459849$.\\
		
		\item $\alpha_{2}$: We have $\Pr[A_{u,2}^v]\ge 0.0528016$ and
		$\gamma(v, \alpha_{1}) \le 0.470365$. \\
		
		
		\item $\alpha_{3}$. We have $\Pr[A_{u,2}^v]\ge 0.13398$ and
		$\gamma(v, \alpha_{3}) \le 0.475282$.\\
		
		
		\item $\beta_{1}$: We have $\Pr[A_{u,2}^v]=1-\euler^{-0.7}*1.7 $ and
		$\gamma(v, \beta_{1}) =0.625395$.\\
		
		\item $\beta_{2}$: We have $\Pr[A_{u,2}^v] \ge 0.226356 $ and
		$\gamma(v, \beta_{2}) \le 0.665882$.\\
		
		\item $\beta_{3}$: We have $\Pr[A_{u,2}^v] \ge 0.1819  $ and
		$\gamma(v, \beta_{3}) \le 0.669804$.\\
		
		\item $\beta_{4}$: We have $\Pr[A_{u,2}^v] \ge 0.1130   $ and
		$\gamma(v, \beta_{4}) \le 0.635563$.\\
		
		\item $\beta_{5}$: We have $\Pr[A_{u,2}^v] \ge 0.0587   $ and
		$\gamma(v, \beta_{5}) \le 0.674471$.\\
		
		\item $\beta_{6}$: We have $\Pr[A_{u,2}^v] \ge 0.1688  $ and
		$\gamma(v, \beta_{6}) \le 0.680529$.\\
		
		\item $\beta_{7}$: We have $\Pr[A_{u,2}^v] \ge 0.1318  $ and
		$\gamma(v, \beta_{7}) \le 0.676155$.\\
		
		\item $\beta_{8}$: We have $\Pr[A_{u,2}^v] \ge 0.0587  $ and
		$\gamma(v, \beta_{8}) \le 0.674471$.\\
		
	\end{itemize}
	
	Hence, the $\WS$ structure is when $u$ is such that $H_u=2/3$ and has one neighbor of type $\alpha_3$. The resultant ratio is $0.7870$.

\xhdr{Proof of Claim \ref{cl:vw3}.}
	The Figure \ref{fig:Fu=1/3} shows the possible configurations of a node $u$ with $H_u=1/3$.
	Again, we omit those cases where $H_{v}<1$. 
	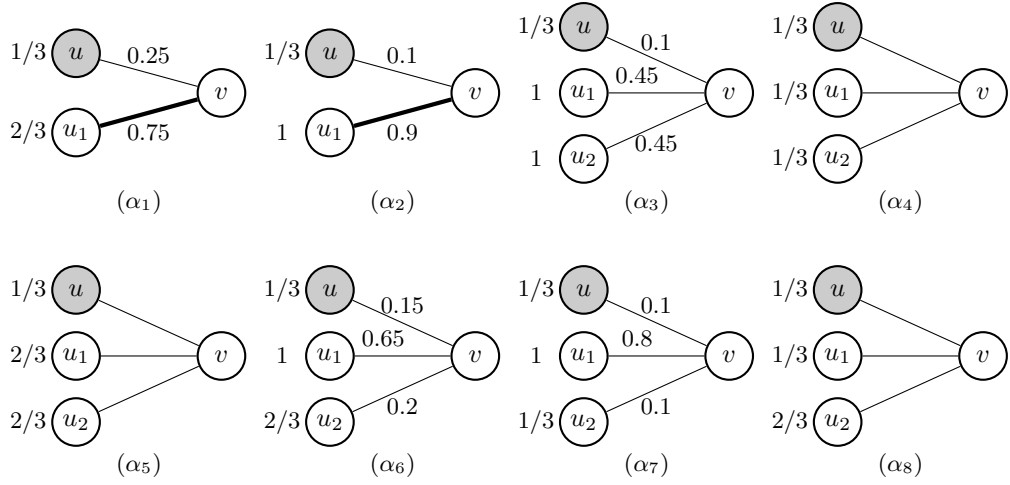
\begin{figure}[h]
		\centering
		\begin{tikzpicture}
[
	xscale=0.8,yscale=1,auto,thick,font=\footnotesize,
  	gray node/.style={circle,
  		inner sep=0pt,minimum size=14pt, 
  		fill=black!20,draw,font=\small},  
  	white node/.style={circle,
  		inner sep=0pt,minimum size=14pt, 
  		fill=white,draw,font=\small},
  node distance=24pt 
]


	\node [white node] (a1v) at (0,0) {$v$};
	\node [gray node] (a1u) [left of=a1v, xshift=-25, yshift=15] {$u$};
	\node [white node] (a1u1) [left of=a1v, xshift=-25, yshift=-15] {$u_1$};
	\draw[thin] (a1u)  -- (a1v) node [midway, above] {$0.25$};
	\draw[ultra thick] (a1u1)  -- (a1v) node [midway, below] {$0.75$};
	\node (a1fu) [left of=a1u, xshift=8, yshift=0] {$1/3$};
	\node (a1fu1) [left of=a1u1, xshift=8, yshift=0] {$2/3$};
	\node [below left of=a1v, xshift=-10, yshift=-20] {($\alpha_1$)};


	\node [white node] (a2v) at (3.75,0) {$v$};
	\node [gray node] (a2u) [left of=a2v, xshift=-25, yshift=15] {$u$};
	\node [white node] (a2u1) [left of=a2v, xshift=-25, yshift=-15] {$u_1$};
	\draw[thin] (a2u)  -- (a2v) node [midway, above] {$0.1$};
	\draw[ultra thick] (a2u1)  -- (a2v) node [midway, below] {$0.9$};
	\node (a2fu) [left of=a2u, xshift=8, yshift=0] {$1/3$};
	\node (a2fu1) [left of=a2u1, xshift=12, yshift=0] {$1$};
	\node [below left of=a2v, xshift=-10, yshift=-20] {($\alpha_2$)};


	\node [white node] (a3v) at (7.5,0) {$v$};
	\node [gray node] (a3u) [left of=a3v, xshift=-25, yshift=25] {$u$};
	\node [white node] (a3u1) [left of=a3v, xshift=-25, yshift=0] {$u_1$};
	\node [white node] (a3u2) [left of=a3v, xshift=-25, yshift=-25] {$u_2$};
	\draw[thin] (a3u)  -- (a3v) node [midway, above] {$0.1$};
	\draw[thin] (a3u1)  -- (a3v) node [above, pos=0.3] {$0.45$};
	\draw[thin] (a3u2)  -- (a3v) node [midway, below] {$0.45$};
	\node (a3fu) [left of=a3u, xshift=8, yshift=0] {$1/3$};
	\node (a3fu1) [left of=a3u1, xshift=12, yshift=0] {$1$};
	\node (a3fu2) [left of=a3u2, xshift=12, yshift=0] {$1$};
	\node [below left of=a3v, xshift=-10, yshift=-20] {($\alpha_3$)};


	\node [white node] (a4v) at (11.25,0) {$v$};
	\node [gray node] (a4u) [left of=a4v, xshift=-25, yshift=25] {$u$};
	\node [white node] (a4u1) [left of=a4v, xshift=-25, yshift=0] {$u_1$};
	\node [white node] (a4u2) [left of=a4v, xshift=-25, yshift=-25] {$u_2$};
	\draw[thin] (a4u)  -- (a4v);
	\draw[thin] (a4u1)  -- (a4v);
	\draw[thin] (a4u2)  -- (a4v);
	\node (a4fu) [left of=a4u, xshift=8, yshift=0] {$1/3$};
	\node (a4fu1) [left of=a4u1, xshift=8, yshift=0] {$1/3$};
	\node (a4fu2) [left of=a4u2, xshift=8, yshift=0] {$1/3$};
	\node [below left of=a4v, xshift=-10, yshift=-20] {($\alpha_4$)};


	\node [white node] (a5v) at (0,-3.5) {$v$};
	\node [gray node] (a5u) [left of=a5v, xshift=-25, yshift=25] {$u$};
	\node [white node] (a5u1) [left of=a5v, xshift=-25, yshift=0] {$u_1$};
	\node [white node] (a5u2) [left of=a5v, xshift=-25, yshift=-25] {$u_2$};
	\draw[thin] (a5u)  -- (a5v);
	\draw[thin] (a5u1)  -- (a5v);
	\draw[thin] (a5u2)  -- (a5v);
	\node (a5fu) [left of=a5u, xshift=8, yshift=0] {$1/3$};
	\node (a5fu1) [left of=a5u1, xshift=8, yshift=0] {$2/3$};
	\node (a5fu2) [left of=a5u2, xshift=8, yshift=0] {$2/3$};
	\node [below left of=a5v, xshift=-10, yshift=-20] {($\alpha_5$)};


	\node [white node] (a6v) at (3.75,-3.5) {$v$};
	\node [gray node] (a6u) [left of=a6v, xshift=-25, yshift=25] {$u$};
	\node [white node] (a6u1) [left of=a6v, xshift=-25, yshift=0] {$u_1$};
	\node [white node] (a6u2) [left of=a6v, xshift=-25, yshift=-25] {$u_2$};
	\draw[thin] (a6u)  -- (a6v) node [midway, above] {$0.15$};
	\draw[thin] (a6u1)  -- (a6v) node [above, pos=0.3] {$0.65$};
	\draw[thin] (a6u2)  -- (a6v) node [midway, below] {$0.2$};
	\node (a6fu) [left of=a6u, xshift=8, yshift=0] {$1/3$};
	\node (a6fu1) [left of=a6u1, xshift=12, yshift=0] {$1$};
	\node (a6fu2) [left of=a6u2, xshift=8, yshift=0] {$2/3$};
	\node [below left of=a6v, xshift=-10, yshift=-20] {($\alpha_6$)};


	\node [white node] (a7v) at (7.5,-3.5) {$v$};
	\node [gray node] (a7u) [left of=a7v, xshift=-25, yshift=25] {$u$};
	\node [white node] (a7u1) [left of=a7v, xshift=-25, yshift=0] {$u_1$};
	\node [white node] (a7u2) [left of=a7v, xshift=-25, yshift=-25] {$u_2$};
	\draw[thin] (a7u)  -- (a7v) node [midway, above] {$0.1$};
	\draw[thin] (a7u1)  -- (a7v) node [above, pos=0.3] {$0.8$};
	\draw[thin] (a7u2)  -- (a7v) node [midway, below] {$0.1$};
	\node (a7fu) [left of=a7u, xshift=8, yshift=0] {$1/3$};
	\node (a7fu1) [left of=a7u1, xshift=12, yshift=0] {$1$};
	\node (a7fu2) [left of=a7u2, xshift=8, yshift=0] {$1/3$};
	\node [below left of=a7v, xshift=-10, yshift=-20] {($\alpha_7$)};


	\node [white node] (a8v) at (11.25,-3.5) {$v$};
	\node [gray node] (a8u) [left of=a8v, xshift=-25, yshift=25] {$u$};
	\node [white node] (a8u1) [left of=a8v, xshift=-25, yshift=0] {$u_1$};
	\node [white node] (a8u2) [left of=a8v, xshift=-25, yshift=-25] {$u_2$};
	\draw[thin] (a8u)  -- (a8v);
	\draw[thin] (a8u1)  -- (a8v);
	\draw[thin] (a8u2)  -- (a8v);
	\node (a8fu) [left of=a8u, xshift=8, yshift=0] {$1/3$};
	\node (a8fu1) [left of=a8u1, xshift=8, yshift=0] {$1/3$};
	\node (a8fu2) [left of=a8u2, xshift=8, yshift=0] {$2/3$};
	\node [below left of=a8v, xshift=-10, yshift=-20] {($\alpha_8$)};

\end{tikzpicture}
		\caption{Vertex-weighted $H_u=1/3$ cases. The value assigned to each edge represents the value after the second modification. No value indicates no modification.}
		\label{fig:Fu=1/3}
	\end{figure}
	
	 We now list the values $\gamma(v, \alpha_{i})$, for each $1 \le i \le 8$. 
	
	\begin{itemize}
		
		\item $\alpha_{1}$: We have $\Pr[A_{u,2}^v] = 1-\euler^{-0.75}*1.75 $ and
		$\gamma(v, \alpha_{1}) =0.643789$.\\
		
		\item $\alpha_{2}$: We have $\Pr[A_{u,2}^v] \ge 0.282256$ and
		$\gamma(v, \alpha_{2})  \le 0.649443$.\\
		
		\item $\alpha_{3}$: We have $\Pr[A_{u,2}^v] \ge 0.1935$ and
		$\gamma(v, \alpha_{3})  \le 0.729751$.\\
		
		\item $\alpha_{4}$: We have $\Pr[A_{u,2}^v] \ge 0.0587$ and
		$\gamma(v, \alpha_{4})  \le 0.674471$.\\
		
		\item $\alpha_{5}$: $\gamma(v, \alpha_{5})  \le 0.674471$.\\
		
		\item $\alpha_{6}$: We have $\Pr[A_{u,2}^v] \ge 0.1546$ and
		$\gamma(v, \alpha_{6})  \le 0.727643$.\\
		
		\item $\alpha_{7}$: We have $\Pr[A_{u,2}^v] \ge 0.1938$ and
		$\gamma(v, \alpha_{7})  \le 0.72948$.\\

		\item $\alpha_{8}$: $\gamma(v, \alpha_{8})  \le 0.674471$.\\
	\end{itemize}
	
	Hence, the $\WS$ for node $u$ with $H_u=1/3$ is when $u$ has one neighbor of type $\alpha_3$.
	The resultant ratio is $0.8107$.

\end{document}